\documentclass{amsart}
\pdfoutput=1

\usepackage{amssymb}
\usepackage{amsthm}
\usepackage{amsfonts}
\usepackage{amsmath}
\usepackage{pmboxdraw}
\usepackage{verbatim} % for comment
\usepackage{graphicx}
\usepackage{color}
\usepackage[colorlinks=true, citecolor=blue, filecolor=black, linkcolor=black, urlcolor=black]{hyperref}
\usepackage{cite}
\usepackage[normalem]{ulem}
\usepackage{subcaption}
\usepackage{todonotes}
\usepackage{bbm}
\allowdisplaybreaks
\newcommand{\RN}[1]{%
	\textup{\uppercase\expandafter{\romannumeral#1}}%
}

\def\wh{\widehat}

\def\Im{ \mathrm{Im}}

\def\OO{\mathcal{O}}

\def\CC{\mathbb{C}}

\def\C{\mathbb{C}}

\def\H{\mathbb{H}}
\def\P{\mathbf{P}}
\def\R{\mathbb{R}}

\newcommand{\re}{\operatorname{Re}}
\newcommand{\im}{\operatorname{Im}}
\newcommand{\erfc}{\operatorname{erfc}}

\newcommand{\ZZ}{\mathbb{Z}}

\newcommand{\RR}{\mathbb{R}}

\newcommand{\al}{\alpha}
\newcommand{\Del}{\Delta}
\newcommand{\del}{\delta}

\newcommand{\sig}{\sigma}

\newcommand{\mc}{\mathcal}
\newcommand{\mf}{\mathfrak}

\newcommand{\HH}{\mathbb{H}}
\newcommand{\II}{\mathcal{I}}
\newcommand{\JJ}{\mathcal{J}}

\newcommand{\mm}{\mathfrak{m}}
\newcommand{\PP}{\mathbb{P}}
\newcommand{\PPP}{\mathcal{P}}

\newcommand{\ind}{\mathbbm{1}}

\renewcommand{\Im}{\operatorname{Im}}
\newcommand{\floor}[1]{\lfloor#1\rfloor}
\newcommand{\ceil}[1]{\lceil#1\rceil}

\theoremstyle{plain}
\newtheorem*{thm*}{Theorem}
\newtheorem{thm}{Theorem}[section]
\newtheorem{lem}[thm]{Lemma}
\newtheorem{cor}[thm]{Corollary}
\newtheorem{prop}[thm]{Proposition}
\newtheorem*{prop*}{Proposition}
\newtheorem*{lem*}{Lemma}

\theoremstyle{definition}
\newtheorem*{eg*}{Example}
\newtheorem*{egs*}{Examples}

\theoremstyle{remark}
\newtheorem*{rmk*}{Remark}
\newtheorem*{rmks*}{Remarks}
\newtheorem{rem}[thm]{Remark}
\numberwithin{equation}{section}

\begin{document}
\title[Large gap probabilities of the spherical ensembles]{Large gap probabilities of complex and symplectic \\ spherical ensembles with point charges}

%%%%%%%%%%%%%%%%%%%%%%%%%%%%% author %%%%%%%%%%%%%%%%%%%%%%%%%%%%

\author{Sung-Soo Byun}
\address{Department of Mathematical Sciences and Research Institute of Mathematics, Seoul National University, Seoul 151-747, Republic of Korea}
\email{sungsoobyun@snu.ac.kr}

\author{Seongjae Park}
\address{Department of Mathematical Sciences, Seoul National University, Seoul 151-747, Republic of Korea}
\email{psj\_0708@snu.ac.kr}

%%%%%%%%%%%%%%%%%%%%%%%%%%%%% author %%%%%%%%%%%%%%%%%%%%%%%%%%%%

\date{\today}

\thanks{Sung-Soo Byun was partially supported by the POSCO TJ Park Foundation (POSCO Science Fellowship), by the New Faculty Startup Fund at Seoul National University and by the National Research Foundation of Korea funded by the Korea government (NRF-2016K2A9A2A13003815, RS-2023-00301976). 
}


%%%%%%%%%%%%%%%%%%%%%%Abstract%%%%%%%%%%%%%%%%%%%%%%%%%%%%%%%%%%

%%%%%%%%%%%%%%%%%%%%%%Abstract%%%%%%%%%%%%%%%%%%%%%%%%%%%%%%%%%%%
\begin{abstract}
We consider $n$ eigenvalues of complex and symplectic induced spherical ensembles, which can be realised as two-dimensional determinantal and Pfaffian Coulomb gases on the Riemann sphere under the insertion of point charges.
For both cases, we show that the probability that there are no eigenvalues in a spherical cap around the poles has an asymptotic behaviour as $n\to \infty$ of the form
$$
\exp\Big( c_1 n^2 + c_2 n\log n + c_3 n + c_4 \sqrt n + c_5 \log n + c_6 + \OO(n^{-\frac1{12}}) \Big) 
$$
and determine the coefficients explicitly. 
Our results provide the second example of precise (up to and including the constant term) large gap asymptotic behaviours for two-dimensional point processes, following a recent breakthrough by Charlier. 
\end{abstract}
%%%%%%%%%%%%%%%%%%%%%%Abstract%%%%%%%%%%%%%%%%%%%%%%%%%%%%%%%%%%%

\maketitle
%\tableofcontents
%%%%%%%%%%%%%%%%%%%%%%%introduction%%%%%%%%%%%%%%%%%%%%%%%%%%%%%%%%%

\section{Introduction and main results}

\subsection{Spherical ensembles}

We consider two-dimensional Coulomb gases \cite{Fo10} (also known as one-component plasmas) $\boldsymbol{x}=\{x_j\}_{j=1}^n$ on the Riemann sphere $\mathbb{S}= \{ x \in \mathbb{R}^3: \| x \| =1 \}$, under the insertion of point charges at the south and north poles, ${\rm{S}}=(0,0,-1)$ and ${\rm{N}}=(0,0,1)$ respectively. Here, $\| \cdot \|$ is the Euclidean metric in $\R^3$.  
In addition to the point charge insertions, we consider both determinantal and Pfaffian Coulomb gases. 
Consequently, for given fixed non-negative integers $\alpha$ and $c$ determining point charges at the poles, the joint probability distributions of the models are given by 
\begin{align}
 \label{Gibbs cplx sphere}
d\boldsymbol{\mathcal{P}}_{n}^\C(\boldsymbol{x})&=\frac{1}{ \mathcal{Z}_{n }^{ \mathbb{C} }   } \prod_{j>k=1}^n \| x_j-x_k \|^{2} \prod_{j=1}^{n}  \| x_j-{\rm{N}} \|^{2\alpha}  \,  \| x_j-{\rm{S}} \|^{2c}  \,dA_{ \mathbb S }(x_j),
\\
d\boldsymbol{\mathcal{P}}_{n}^\H(\boldsymbol{x})&=\frac{1}{   \mathcal{Z}_{n }^{ \mathbb{H} }  } \prod_{j>k=1}^n \| x_j-x_k \|^{2} \| x_j-\overline{x}_k \|^{2} \prod_{j=1}^{n} \|x_j-\overline{x}_j\|^2   \| x_j-{\rm{N}} \|^{4\alpha}  \,  \| x_j-{\rm{S}} \|^{4c}  \,dA_{ \mathbb S }(x_j),  \label{Gibbs symplectic sphere}
\end{align}
where $dA_{ \mathbb S }$ is the area measure on $\mathbb{S}$, and  we use the convention $\overline{x}=(x_1,-x_2,x_3)$ for $x=(x_1,x_2,x_3)$. 
Here, $\mathcal{Z}_{n}^{ \mathbb{C} }$ and $ \mathcal{Z}_{n}^{ \mathbb{H} }$ are partition functions that make $\boldsymbol{\mathcal{P}}_{n}^\C$ and $\boldsymbol{\mathcal{P}}_{n}^\H$ probability measures. 
The superscripts $\mathbb{C}$ and $\mathbb{H}$ denoting determinantal and Pfaffian Coulomb gases, respectively, will become clear below. 

For readers who are inclined towards the statistical physics perspective, it suffices to regard \eqref{Gibbs cplx sphere} and \eqref{Gibbs symplectic sphere} as the primary subjects of our investigations. 
On the other hand, for readers from the random matrix theory community, let us stress that these models have realisations of eigenvalues of random matrix models called complex \cite{FF11,FK09,Kr09} and symplectic \cite{BF23a,FM12,Ma13,MP17} spherical induced ensembles. 
We refer to \cite[Section 2.5]{BF22} and \cite[Section 6.3]{BF23} for comprehensive reviews, see also a recent work \cite{No23} and references therein.
A simple way to define these random matrix models is as $n \times n$ complex or quaternion matrices, denoted by $\textbf{G}$, whose matrix probability distribution function is proportional to
\begin{equation} \label{mat dist}
\frac{ \det (\textbf{G} \textbf{G}^*)^{c \beta/2} }{ \det ( \mathbbm{1}_n + \textbf{G} \textbf{G}^* )^{ (2n+c+\alpha) \beta/2 } }, \qquad \beta= \begin{cases}
2 & \textup{for the complex case},
\smallskip 
\\
4 & \textup{for the quaternionic case}. 
\end{cases}
\end{equation}
For the special case when $c=\alpha=0$ (i.e. without point charges), these models are called the complex and symplectic spherical ensembles. 
In this case, an element-wise realisation of random matrices distributed as \eqref{mat dist} is due to Krishnapur \cite{Kr09}.  
To be more concrete, the matrix distribution \eqref{mat dist} follows from $\textbf{G}_1 \textbf{G}_2^{-1}$, where $\textbf{G}_1$ and $\textbf{G}_2$ are independent complex or symplectic Ginibre matrices whose elements are given by i.i.d. complex or quaternionic Gaussian random variables.  
It is evident from the construction that the spherical ensembles are closely related to generalised eigenvalue problems \cite{FM12,EKS94}, which find applications in random plane geometry.
Beyond the spherical cases, the presence of point charges can also be realised at the level of matrix models via a so-called inducing procedure, see e.g. \cite{FF11,MP17,FBKSZ12}. This construction requires the use of rectangular Ginibre matrices as well as Haar-distributed matrices within the symmetry classes under consideration.
In this inducing procedure, one needs to assume that $\alpha$ and $c$ are non-negative integer values since they are related to the rectangular parameters.

It follows from \eqref{mat dist} that the joint distributions of eigenvalues $\boldsymbol{z}=\{z_j\}_{j=1}^n$ are given by
\begin{align} \label{Gibbs cplx}
d\P_{n}^\C(\boldsymbol{z})&=\frac{1}{Z_{n}^\C } \prod_{j>k=1}^n |z_j-z_k|^{2} \prod_{j=1}^{n}  e^{-n Q_n(z_j) }  \,dA(z_j),
\\
d\P_{n}^\H(\boldsymbol{z})&=\frac{1}{Z_{n}^\H } \prod_{j>k=1}^n |z_j-z_k|^{2} |z_j-\overline{z}_k|^2  \prod_{j=1}^{n}|z_j-\overline{z}_j|^2  e^{-2n Q_{n}(z_j) }  \,dA(z_j),  \label{Gibbs symplectic}
\end{align}
where $dA(z)=d^2z/\pi$ and the external potential $Q_n$ is given by 
\begin{equation} \label{potential spherical}
 Q_n(z):= \frac{ n+\alpha+c+1 }{n} \log(1+|z|^{2})-\frac{ 2c }{n} \log |z|.  
\end{equation}
The models \eqref{Gibbs cplx} and \eqref{Gibbs symplectic} are again two-dimensional Coulomb gas ensembles in the complex plane, with an $n$-dependent potential.  
Due to their special integrable structures, they are determinantal and Pfaffian Coulomb gas ensembles, respectively.   
We mention that the partition functions $Z_n^{ \mathbb{C} } $ and $Z_n^{ \mathbb{H} } $ can be expressed in terms of the gamma functions as 
\begin{equation} \label{partition product expression}
Z_n^{ \mathbb{C} } = n! \prod_{k=0}^{n-1} \frac{  \Gamma(k+c+1) \Gamma(n+\alpha-k) }{ \Gamma(n+\alpha+c+1) } , \qquad   Z_n^{ \mathbb{H} }  = n! \,2^n\, \prod_{k=0}^{n-1} \frac{  \Gamma(2k+2c+2) \Gamma(2n+2\alpha-2k) }{ \Gamma(2n+2\alpha+2c+2) }. 
\end{equation} 

To observe the equivalence between \eqref{Gibbs cplx sphere} and \eqref{Gibbs cplx}, as well as \eqref{Gibbs symplectic sphere} and \eqref{Gibbs symplectic}, recall that the stereographic projection $\phi: \mathbb{S}^2 \to \C \cup \{ \infty \}$ is given by 
\begin{equation}
\phi(x) = 
\begin{cases}
\dfrac{ x_1+ix_2 }{ 1-x_3 }, &\textup{if } x \in \mathbb{S}^2 \setminus {\rm{N}},
\smallskip 
\\
\infty, &\textup{if }x= {\rm{N}},
\end{cases} \qquad \phi^{-1}(z)= \Big( \frac{2 \re z}{1+|z|^2}, \frac{2\im z}{1+|z|^2}, \frac{|z|^2-1}{1+|z|^2} \Big),
\end{equation}
where $x=(x_1,x_2,x_3).$ 
Using these, one can easily check that the measures \eqref{Gibbs cplx sphere} and \eqref{Gibbs symplectic sphere} are a pull-back of \eqref{Gibbs cplx} and \eqref{Gibbs symplectic}, whence the name \emph{spherical} for the random matrix $\textbf{G}$ was first coined \cite{Kr09,FK09}. 
Furthermore, the partition functions $ \mathcal{Z}_{n}^{\mathbb{C}} $ and $ \mathcal{Z}_{n}^{\mathbb{H}} $ in \eqref{Gibbs cplx sphere} and \eqref{Gibbs symplectic sphere} are also given by \eqref{partition product expression}, up to explicit constants. %of the form $2^{n(n-1)}$ or $2^{2n^2}$.
From this viewpoint, the potential \eqref{potential spherical} in the complex plane can also be realised as point charge insertions at the origin and infinity, cf. \cite{Sam17}.

The potential $Q_n$ in \eqref{potential spherical} is indeed a prominent example of \emph{weakly confining potential}, see e.g. \cite{BGNW21} and references therein. 
Contrary to Coulomb gas ensembles with a confining potential that makes the particles lie in a compact set in the large-$n$ limit, the ensembles \eqref{Gibbs cplx} and \eqref{Gibbs symplectic} tend to be distributed in the whole complex plane. 
To be more precise, due to the Laplacian growth property of the two-dimensional Coulomb gas ensemble, 
for given fixed $\alpha$ and $c$, the ensembles \eqref{Gibbs cplx} and \eqref{Gibbs symplectic} tend to occupy the whole complex plane with the limiting density
\begin{equation} \label{def of rho}
\rho^{\rm sp}(z)= \frac{1}{(1+|z|^2)^2} . 
\end{equation}
As a consequence, using the inverse stereographic projection, we infer that as $n \to \infty$, the ensembles \eqref{Gibbs cplx sphere} and \eqref{Gibbs symplectic sphere} tend to be uniformly distributed on the whole Riemann sphere $\mathbb{S}$, see Figure~\ref{Fig_spherical}. 
We also mention that the equilibrium measure problems associated with spherical Coulomb gases with point charges have been recently studied in \cite{LD21,CK22}. 

\begin{figure}[t]
	\begin{subfigure}{0.48\textwidth}
	\begin{center}	
		\includegraphics[width=0.7\textwidth]{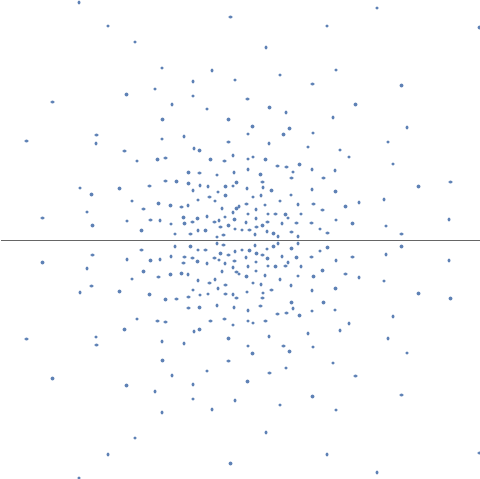}
	\end{center}
	\subcaption{ $\{z_j\}_{j=1}^n$, $z_j \in \C$ }
\end{subfigure}	
	\begin{subfigure}{0.48\textwidth}
	\begin{center}	
		\includegraphics[width=0.7\textwidth]{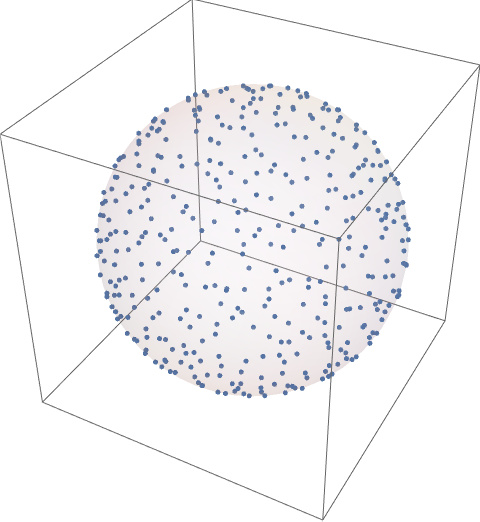}
	\end{center}
	\subcaption{ $\{x_j\}_{j=1}^n$, $x_j \in \mathbb{S}$ }
\end{subfigure}	
	\caption{The plot (A) displays eigenvalues of the spherical symplectic ensembles, where $n=200$. The plot (B) shows its projection to the sphere.} \label{Fig_spherical}
\end{figure}

\subsection{Main results: large gap probabilities}

In this work, we study the asymptotic behaviour of the large gap probabilities as $n \to \infty$ of the ensembles \eqref{Gibbs cplx} and \eqref{Gibbs symplectic}: 
\begin{align}
\begin{split}
 \label{Gap prob def inner}
\mathcal{P}_n(R; \alpha,c) & :=\mathbb{P}\Big[ \textup{there is no eigenvalues in } \mathbb{D}(0,R)  \Big] .
\end{split}
\end{align}
This problem was already considered a decade ago in the work \cite{AZ15}, where the leading term for the complex case with $\alpha=c=0$ was obtained (i.e. the constant $A_1$ in Theorem~\ref{thm:1.1}). 

More generally, obtaining large gap asymptotics is a classical and challenging problem in random matrix theory with a long history, which has attracted considerable attention over the years. The literature pertaining to this topic will be reviewed in Subsection~\ref{Subsection_related work} below.
Note that $\mathcal{P}_n$ can also be written as 
\begin{equation}
\mathcal{P}_n(R; \alpha,c) = \mathbb{P} \Big[ \min\{|z_{1}|,\ldots,|z_{n}|\} \geq R \, \Big]. 
\end{equation}
Hence $\mathcal P_n$ can be viewed as a heavy-tail distribution for the distribution of the smallest moduli. 
Analogous extreme distributions have been widely studied for various one-dimensional log-correlated point processes (such as the Tracy-Widom distribution for the Airy point process); this will be discussed in Subsection~\ref{subsection_1D process}.
In our case, the moduli $\{|z_j|\}_{j=1}^n$ forms a permanantal point process (see e.g. \cite{AIP14}), and the quantity $\mathcal{P}_n$ can also be interpreted as the large deviation probability of this process.
Furthermore, $\mathcal{P}_n$ is intimately connected to an energy minimisation (electrostatics) problem, which is a fundamental aspect of potential theory, see Subsection~\ref{Subsection_balayage}. 
Additionally, it is equivalent to the free energy of a one-component plasma confined by hard walls, see Subsection~\ref{Subsection_free energy} for further details.  

For the models \eqref{Gibbs cplx sphere} and \eqref{Gibbs symplectic sphere} on the sphere, the probability \eqref{Gap prob def inner} coincides with the probability that there are no particles in a spherical cap (whose center is the south pole) of area
\begin{equation}
4\pi \frac{R^2}{1+R^2},
\end{equation}
see Figure~\ref{Fig_spherical_gap} for some illustrations.
Note that instead of \eqref{Gap prob def inner}, one may also consider 
\begin{align}
\begin{split}
 \label{Gap prob def outer}
 \widetilde{\mathcal{P}}_n(R; \alpha,c) & :=\mathbb{P}\Big[ \textup{there is no eigenvalues in } \mathbb{D}(0,R)^c  \Big] = \mathbb{P} \Big[  \max\{|z_{1}|,\ldots,|z_{n}|\} \le R \, \Big]. 
\end{split}
\end{align}
Due to the sphere geometry, there is a duality relation between $\mathcal{P}_n$ and $\widetilde{\mathcal{P}}_n$:
\begin{equation} \label{duality}
\mathcal{P}_n(R; \alpha,c)  =  \widetilde{\mathcal{P}}_n(1/R; c,\alpha).  
\end{equation} 
Alternatively, \eqref{duality} can also be directly obtained using Lemma~\ref{Lem_finite expression} below. 

In the sequel, we add the superscripts and write $\mathcal{P}_n^\C$ and $\mathcal{P}_n^{ \mathbb{H} }$ to distinguish \eqref{Gibbs cplx} and \eqref{Gibbs symplectic}. 
To state our main results, we need some elementary special functions.
Recall that the complementary error function is defined by \cite[Chapter 7]{NIST}
\begin{equation}
\erfc(z) := \frac{2}{ \sqrt{\pi} } \int_z^\infty e^{-t^2} \,dt
\end{equation}
and that the Barnes $G$-function is defined recursively by \cite[Section 5.17]{NIST}
\begin{equation} \label{Barnes G def}
G(z+1)=\Gamma(z)G(z),\qquad G(1)=1, 
\end{equation}
where $\Gamma$ is the standard gamma function. 
We write 
  \begin{align}
        \mathcal{I}_1 &  =  \sqrt{2} \int_{-\infty}^0 \log\Big(\tfrac12 \erfc(t)\Big) \,dt +  \sqrt{2} \int_{0}^\infty \bigg[ \log\Big(\tfrac12 \erfc(t)\Big)  + t^2+\log t + \log (2\sqrt{\pi}) \bigg]\,dt \approx -2.301,   \label{I1}
        \\
        \II_2 &= \int_{-\infty}^\infty \bigg[\frac{(1+t^2)e^{-t^2}}{\sqrt{\pi}\erfc(t)} - \Big(t^3+\frac32 t\Big)\ind_{(0,\infty)}(t)\bigg] \, dt \approx  0.928.  %0.92787. 
        \label{I2}
        \end{align}
        
We are now ready to state the asymptotics as $n\to\infty$ of the large gap probabilities of the complex induced spherical ensemble \eqref{Gibbs cplx}. 

\begin{thm}[\textbf{Large gap probabilities of the complex ensemble}] 
\label{thm:1.1}
  Let $\al, c \in \ZZ_{\geq 0}$ and $R>0$ be fixed.
  Then as $n \to \infty$, we have
    \begin{align}
        \log \PPP_n^\C(R;\al,c) = A_{1}  n^2 + A_2 n\log n + A_3 n + A_4 \sqrt n + A_5 \log n + A_6 + \OO(n^{-\frac1{12}}),
    \end{align} 
  where  
    \begin{align}
        A_1 &= \frac12 \Big( \frac{R^2}{1+R^2} -\log(1+R^2) \Big), 
        \\
        A_2 &= -\frac12  \frac{R^2}{1+R^2} ,
        \\
        A_3 &= \frac{R^2}{1+R^2}\Big( 1+\al+c - \frac12 \log(2\pi) -\log R \Big)- \al \log (1+R^2), 
        \\
        A_4 &=\frac{R}{1+R^2}\, \II_1,
        \\
        A_5 & = \frac{1+\al+c}{2(1+R^2)} - \frac16 -\frac{\al + c^2}{2}, 
        \\
        \begin{split}
        A_6 &= \Big(\frac{1+\al+c}{1+R^2}- \frac13 - \al - c^2 \Big) \log R - \frac{\al^2 - c^2}{2}\log (1+R^2)
        \\
        & \quad +\Big(\frac12 - (1+\al+c) \frac{R^2}{1+R^2}\Big) \frac{\log (2\pi)}{2} + \frac{(\al+c)(1+\al+c)}{2} \frac{R^2}{1+R^2}
        \\
        &\quad  - \frac23 \frac{1-R^2}{1+R^2} \, \II_2 + \frac{1}{12} \Big( 1+\frac{11}{1+R^2}\Big) -\zeta'(-1) + \log G(c+1).
        \end{split}
    \end{align} 
    Here, $\zeta$ is the Riemann zeta function, $G$ is the Barnes $G$-function, and the constants $\II_1$ and $\II_2$ are given by \eqref{I1} and \eqref{I2}. 
\end{thm}

\begin{rem}[Conjecture for the terms proportional to $n \log n$ and $\sqrt{n}$]
In view of Theorem~\ref{thm:1.1} and \cite[Theorem 1.7]{Ch23}, we formulate the following conjecture for $A_2$ and $A_4$, valid for a general radially symmetric potential $V$ such that $\{z\in \C: |z|< R\}$ is contained in the droplet:   
\begin{equation} 
A_2= -\frac12 \int_{ |z| <R } \,d\sigma_V (z) \qquad \mbox{and} \qquad  \label{universal sqrt n}
A_4 =R \sqrt{\rho(R) } \, \mathcal{I}_1 ,  
\end{equation}
where $d\sigma_V = \Delta V \,dA= \rho \,dA$ is the equilibrium measure \eqref{eq msr form} associated with the potential $V$, see also Subsection~\ref{Subsection_balayage}. 
This conjecture is consistent with Theorem~\ref{thm:1.1}, and also with the result \cite[Theorem 1.7]{Ch23} on the Mittag-Leffler ensemble. 
In our case, $\rho$ is given by \eqref{def of rho}, whereas for the Mittag-Leffler ensemble, $\rho$ is given by \eqref{rho ML}.
\end{rem}

As previously mentioned, when $\alpha=c=0$, $A_{1}$ was already obtained in \cite[Proposition 3.1]{AZ15}. However, a more precise expansion has not been discovered; indeed even the second term $A_2$ is new to our knowledge.

The symplectic counterpart of Theorem~\ref{thm:1.1} is as follows.

\begin{thm}[\textbf{Large gap probabilities of the symplectic ensemble}] \label{thm:1.2}
 Let $\al, c \in \ZZ_{\geq 0}$ and $R>0$ be fixed. 
  Then as $n \to \infty$, we have
    \begin{align}
        \log \PPP_n^{\mathbb H}(R;\al,c) = \widehat{A}_1 n^2 +  \widehat{A}_2 n\log n +  \widehat{A}_3 n +  \widehat{A}_4 \sqrt n +  \widehat{A}_5 \log n +  \widehat{A}_6 + \OO(n^{-\frac1{12}}),
    \end{align} 
where  
 \begin{align}
  \widehat{A}_1 &= \frac{R^2}{1+R^2} -\log(1+R^2),
  \\
   \widehat{A}_2 &= -\frac12 \frac{R^2}{1+R^2},
   \\
    \widehat{A}_3 &= \frac{R^2}{1+R^2}\Big(2+2\al+2c-\frac12\log(4\pi)-\log R \Big)-\Big(2\al+\frac12\Big)\log(1+R^2),
        \\
         \widehat{A}_4 &= \frac{1}{\sqrt{2}} \frac{R}{1+R^2} \, \II_1   ,
         \\
          \widehat{A}_5 & = \frac{1+\al+c}{2(1+R^2)}-\frac{5}{24}-\frac{\al+c}{2}-c^2 , 
          \\
  \begin{split}
       \widehat{A}_6 &= \Big(\frac{1+\al+c}{1+R^2}-\frac{5}{12}-\al-c-2c^2\Big)\log R - (\al-c)\Big(\frac12+\al+c\Big)\log(1+R^2) 
         \\
        & \quad -\Big(c+(1+\al+c)\frac{R^2}{1+R^2}\Big)\frac{\log(4\pi)}{2}+\Big(c+\frac14\Big)\log 2 +\Big(\frac12+\al+c\Big)(1+\al+c)\frac{R^2}{1+R^2}\\
        &\quad -\frac{1}{3}\frac{1-R^2}{1+R^2}\II_2 + \frac1{24}\Big(1+\frac{11}{1+R^2}\Big)-2\zeta'(-1)+\log\Big(G(c+1)G(c+\tfrac32)\Big).    
  \end{split}
 \end{align}
    Here, %$\zeta$ is the Riemann zeta function, $G$ is the Barnes $G$-function, and 
    the constants $\II_1$ and $\II_2$ are given by \eqref{I1} and \eqref{I2}. 
\end{thm}

The numerical verifications of Theorems~\ref{thm:1.1} and \ref{thm:1.2} are provided in Figure~\ref{Fig_numerics}. Here, we present the plots for the case where $\alpha=c=0$, and similar figures can be observed for other values of $\alpha$ and $c$. 

In general, for Pfaffian point processes, the large gap problems have been much less developed, even in dimension one, see Subsection~\ref{subsection_1D process}. Indeed, Theorem~\ref{thm:1.2} provides one of the very few results on precise large gap asymptotics for Pfaffian point processes.

\begin{figure}[h!]
	\begin{subfigure}{0.48\textwidth}
	\begin{center}	
		\includegraphics[width=\textwidth]{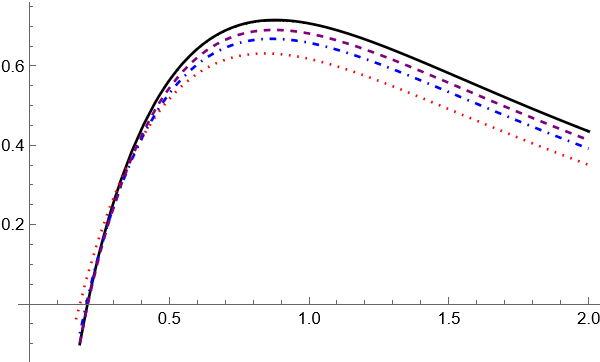}
	\end{center}
	\subcaption{ Complex }
\end{subfigure}	 \quad 
	\begin{subfigure}{0.48\textwidth}
	\begin{center}	
		\includegraphics[width=\textwidth]{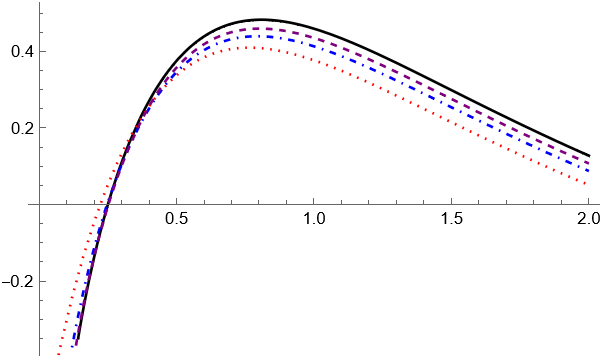}
	\end{center}
	\subcaption{ Symplectic }
\end{subfigure}	
	\caption{The plot (A) shows $R \mapsto A_6$ (black full line) and its comparison $ R \mapsto   \log \PPP_n^\C(R;\al,c) - ( A_{1}  n^2 + A_2 n\log n + A_3 n + A_4 \sqrt n + A_5 \log n ),$ where $\alpha=c=0$. Here, $n=10$ (red, dotted line), $n=40$ (blue, dot-dashed line) and $n=160$ (purple, dashed line), cf. Lemma~\ref{Lem_finite expression}. The plot (B) shows the same figure for the symplectic case.  } \label{Fig_numerics}
\end{figure}

By comparing Theorems~\ref{thm:1.1} and \ref{thm:1.2}, one can observe
\begin{equation}
2\,A_1= \widehat{A}_1, \qquad A_2= \widehat{A}_2. 
\end{equation}
The agreement of the first two terms (up to a factor of 2 for the first term) turns out to be a universal feature of free energy expansions of determinantal and Pfaffian Coulomb gases. 
We refer the reader to \cite{BKS23}, where such the same phenomenon was established for partition functions with soft edges.

\subsection{Related works and further results} \label{Subsection_related work}

In this subsection, we shall discuss the progress made regarding large gap probabilities and related topics, providing further motivations for our study and elucidating our contributions to this direction. In addition, we present further consequences arising from our main results.

\subsubsection{Large gap probabilities for the Ginibre ensemble} \label{subsection_gap Ginibre}

As previously mentioned, Ginibre matrices are non-Hermitian random matrices with i.i.d. Gaussian entries, see \cite{BF22,BF23} for recent reviews. 
The joint probability distributions of their eigenvalues are of the form \eqref{Gibbs cplx} and \eqref{Gibbs symplectic}, where the external potential is instead given by $|z|^2$. Its two-parameter extension, the Mittag-Leffler ensemble, has a potential of the form 
\begin{equation}
Q^{ \rm ML }(z) = |z|^{2b}- \frac{2c}{n} \log |z|, \qquad (b>0,\, c>-1). 
\end{equation}
The associated limiting spectrum is a centred disc of radius $b^{ -\frac{1}{2b} }$ and the density is given by 
\begin{equation} \label{rho ML}
\rho^{ \rm ML }(z) = b^2 |z|^{2b-2}. 
\end{equation}
In particular, if $b=1$, this gives rise to the well-known circular law. 

An early study on the expansion of large gap probabilities of the complex Ginibre ensemble was due to Forrester \cite{Fo92}, where he derived the first \emph{four terms} in the expansion, i.e. $C_j$ ($j=1,2,3,4$) in \eqref{large gap Ginibre} below. (See also \cite{GHS88,JLM93} for the first two terms). 
Subsequently, using the theory of skew-orthogonal polynomials, the symplectic counterpart was addressed in \cite{APS09}, where the authors obtained the \emph{first two terms}, and partially conjectured the third one. (See also \cite{AS13} for the extension to the product of Ginibre matrices.) 

For a considerable period, obtaining a more precise expansion has been a longstanding open problem.
It is only now that the specialised methods of asymptotic analysis required for this task have been mastered in the work \cite{Ch23} of Charlier. 
Consequently, it has been established that the large gap probabilities of the Mittag-Leffler ensemble take the form:
\begin{equation} \label{large gap Ginibre}
\exp\Big( C_1 n^2 + C_2 n\log n + C_3 n + C_4 \sqrt n + C_5 \log n + C_6 + \OO(n^{-\frac1{12}}) \Big). 
\end{equation}
We mention that even though the large gap probabilities were stated in  \cite{Ch23} only for the complex ensembles, the symplectic counterparts also follow by combining the results in \cite{Ch23} with \cite[Proposition 5.10]{BF23}. 
%By \cite[Theorem 1.7]{Ch23} with $g=1$, 

The key ingredient for the analysis of Charlier is a uniform asymptotic expansion of the incomplete gamma function due to Temme, see \cite[Section 11.2.4]{Te96} and \cite[Section 8.12]{NIST}.  
The basic concept of this method involves dividing the required summation based on different asymptotic regimes of the incomplete gamma function. Each division is meticulously chosen to minimise error bounds during asymptotic analysis, ensuring the smallest possible cumulative errors. 
Even though some of these ideas trace back to Forrester’s work \cite{Fo92} more than 30 years ago, their implementation demands technical mastery and a very systematic approach, particularly to achieve precision up to the constants $C_5$ and $C_6$ \cite{Ch23}. This method was also recently used in e.g. \cite{Ch22,ACCL24} to obtain precise results on counting statistics.

Contrary to \cite{Ch23,Ch22,ACCL24}, the gap probability $\mathcal P_n$ is not expressed in terms of the incomplete gamma function but in terms of the \emph{incomplete beta function} (cf. Lemma~\ref{Lem_finite expression}). 
This difference with \cite{Ch23,Ch22,ACCL24} has far-reaching consequences in the proofs of Theorems~\ref{thm:1.1} and ~\ref{thm:1.2}.  
Indeed, our analysis requires very detailed asymptotics of the incomplete beta functions for various regimes of the parameters, and in particular we rely on some results in \cite[Section 11.3.3]{Te96}.  
We illustrate our proofs' strategy in Figures~\ref{Fig_sum division} and ~\ref{Fig_sum division H} below. 
Interestingly, after a very detailed and intricate analysis, the resulting error bounds in Theorems~\ref{thm:1.1} and ~\ref{thm:1.2} are $\OO(n^{-\frac{1}{12}})$, i.e. of the same order as the bounds for the Mittag-Leffler ensemble \eqref{large gap Ginibre}.
As a side note, one of the additional technical (or practical) difficulties in our work stems from a typo present in Temme's original book \cite[Section 11.3.3]{Te96}. 
This typo is also present in NIST \cite[Eq.(8.18.9)]{NIST}, and we are not aware of a literature where these typos have been reported, see Lemma~\ref{beta} for the corrected version.

\subsubsection{Leading order asymptotics and balayage measure} \label{Subsection_balayage}

The leading order asymptotic behaviour of the large gap probabilities can alternatively be derived using a potential theoretic method \cite{ST97} by investigating the associated balayage measure. 
For this purpose, let us recall that for a given probability measure $\mu$ on $\mathbb{C}$, the weighted logarithmic energy $I_V[\mu]$ associated with the external potential $V$ is given by 
\begin{equation} \label{energy}
I_V[\mu]:= \int_{ \mathbb{C}^2 } \log  \frac{1}{|z-w|} \, d\mu(z)\, d\mu(w) +\int_{ \mathbb{C} } V \,d\mu .
\end{equation}
Assuming that $V$ is lower semi-continuous and finite on some set of positive capacity, $I_V[\mu]$ has a unique minimiser $\sigma_V$ with a support $S$ called the droplet.  
Furthermore, if $V$ is $C^2$-smooth in a neighbourhood of $S$, then $\sigma_V$ is absolutely continuous with respect to $dA$, and takes the form
\begin{equation} \label{eq msr form}
d\sigma_V(z) = \Delta V(z) \,\mathbbm{1}_{S_V}(z) \, dA(z), \qquad \Delta =\partial \bar{\partial}.
\end{equation}
This property is often called the Laplacian growth, especially in the context of the Hele-Shaw flow.  
We now consider a gap $\Omega \subset \mathbb{C}$ and redefine the potential by 
\begin{equation} \label{V gap gen}
V_{\Omega}(z)= \begin{cases}
V(z) & \textup{if } z \not \in \Omega,
\smallskip 
\\
+\infty & \textup{if } z \in \Omega.
\end{cases}
\end{equation}
If $S \cap \Omega \neq \emptyset$, then the equilibrium measure associated with $V_\Omega$ is no longer absolutely continuous with respect to $dA$, but rather takes the form
\begin{equation}
d\sigma_{ V_{\Omega} }(z) =  \Delta V(z) \,\mathbbm{1}_{S_V \setminus \Omega }(z) \, dA(z) + \nu(z) \,\mathbbm{1}_{ \partial \Omega }(z) |dz|. 
\end{equation}

As the term \emph{balayage} means \emph{sweeping} in French, from an electrostatic perspective, the constraint that there are no particles in $\Omega$ pushes them to the boundary of $\Omega$ with a non-trivial distribution $\nu(z)$, cf. Figure~\ref{Fig_spherical_gap}.
The balayage measure can be used to derive the leading order of the large gap probabilities. In particular, the leading order of the (logarithm of the) large gap probabilities is proportional to 
\begin{equation} \label{large rate function}
I_{ V }[ \sigma_V ] - I_{ V_\Omega }[ \sigma_{ V_\Omega } ], 
\end{equation}
see e.g. \cite{Ad18,CMV16}. 
From a probabilistic point of view, the leading term \eqref{large rate function} is a large deviation rate function, see e.g. \cite{Se24}. An advantage of this approach is that it can be applied not only to the $\beta=2$ Coulomb gas ensembles but also to general $\beta>0$. However, for a general shape of domain, it is far from obvious to explicitly compute the balayage measure, and some non-trivial examples have been recently obtained in \cite{AR17,Ch23a}. 
On one hand, this approach usually has its limitations in deriving only the leading order asymptotic expansion.

In our present case, one can consider the balayage measure associated with the potential 
\begin{equation} \label{potential gap}
Q^R_n(z):= \begin{cases} 
+\infty & \textup{if } |z|\le R,
\smallskip 
\\
Q_n(z) & \textup{if }|z|>R, 
\end{cases}
\end{equation}
where $Q_n$ is given by \eqref{potential spherical}. In this case, due to the rotational symmetry, the balayage measure becomes uniform distribution on $|z|=R$, see Figure~\ref{Fig_spherical_gap}. 
Then the leading order of our results also follows from computing the associated energy, see \cite[Eq.(2.9)]{Ch23a}.

\begin{figure}[t]
	\begin{subfigure}{0.48\textwidth}
	\begin{center}	
		\includegraphics[width=0.8\textwidth]{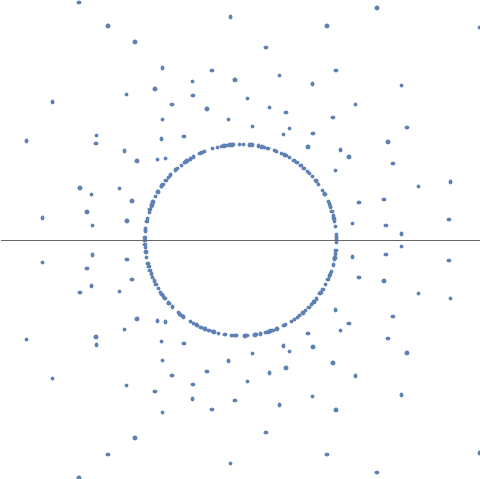}
	\end{center}
	\subcaption{ $\{z_j\}_{j=1}^n$, $z_j \in \C$ }
\end{subfigure}	
	\begin{subfigure}{0.48\textwidth}
	\begin{center}	
		\includegraphics[width=0.8\textwidth]{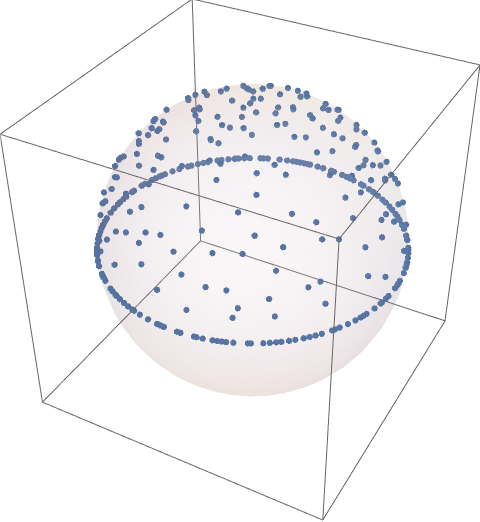}
	\end{center}
	\subcaption{ $\{x_j\}_{j=1}^n$, $x_j \in \mathbb{S}$ }
\end{subfigure}	
	\caption{The plot (A) displays eigenvalues of the spherical symplectic ensembles conditioned to have a gap $\{z\in \C : |z|<1\}$, where $n=200$. The plot (B) shows its projection to the sphere.} \label{Fig_spherical_gap}
\end{figure}

Let us also mention that beyond the two-dimensional Coulomb gases, the zeros of Gaussian analytic functions have been studied in the literature, see e.g. \cite{GN18} and references therein. A characteristic feature when conditioning on a hole event for this model is that it produces a macroscopic region outside the hole region \cite{GN19}. Such a forbidden region has been further investigated in recent literature \cite{NW24}.

\subsubsection{Partition functions for two-dimensional point processes} \label{Subsection_free energy}

For a given external potential $V$, the partition functions of determinantal and Pfaffian Coulomb gases are given by 
\begin{align} \label{ZN V general complex}
Z_{n,V}^{\mathbb C} &:= \int_{\C^n} \prod_{j>k=1}^n |z_j-z_k|^{2} \prod_{j=1}^{n}  e^{-n \, V(z_j) }  \,dA(z_j),
\\
Z_{n,V}^{\mathbb H} &:= \int_{\C^n} \prod_{j>k=1}^n |z_j-z_k|^{2} |z_j-\overline{z}_k|^2  \prod_{j=1}^{n}|z_j-\overline{z}_j|^2  e^{-2n V(z_j) }  \,dA(z_j),   \label{ZN V general symplectic}
\end{align} 
respectively. 
The asymptotic expansion of the partition function stands as one of the cornerstones in Coulomb gas theory. 
For a regular potential $V$, substantial progress has been made in this direction \cite{BBNY19,CFTW15,ZW06,LS17}. Furthermore, for comprehensive literature and the latest advancements, we refer the reader to \cite{ACC23,BKS23,BSY24,Se24}.

For the partition functions $ Z_n^{ \mathbb{C} }$ and $ Z_n^{ \mathbb{H} }$ of induced spherical ensembles, their asymptotic expansions can be derived from well-known formulas involving the Barnes $G$-function.

\begin{prop}[\textbf{Free energy expansions of the induced spherical ensembles}]  \label{Prop_free energy soft}
As $n \to \infty$, the following asymptotic expansion holds. 
\begin{itemize}
    \item \textbf{\textup{(Complex ensemble)}} We have 
\begin{align}
\log Z_n^{ \mathbb{C} } = \mathcal{A}_1 n^2 + \mathcal{A}_2 n\log n + \mathcal{A}_3 n + \mathcal{A}_4 \sqrt n + \mathcal{A}_5 \log n + \mathcal{A}_6 + \OO(n^{-1}),
\end{align}
where 
\begin{gather*}
\mathcal{A}_1 = -\frac12 , \qquad \mathcal{A}_2 = \frac12, \qquad \mathcal{A}_3= \frac{\log (2\pi)}{2} -1  -\alpha-c ,
%\\
\qquad \mathcal{A}_4 =0, \qquad \mathcal{A}_5=   \frac{\alpha^2+c^2}{2} + \frac{1}{3}, 
\\
\mathcal{A}_6 =  \frac{\log(2\pi)}{2} -\frac{1}{12} +2\zeta'(-1) -\frac12 (\alpha+c)(\alpha+c+1-\log(2\pi))  -\log \Big(G(\alpha+1)G(c+1) \Big). 
\end{gather*} 
\smallskip 
   \item \textbf{\textup{(Symplectic ensemble)}} We have 
\begin{align}
\log Z_n^{ \mathbb{H} } = \widehat{\mathcal{A}}_1 n^2 +  \widehat{\mathcal{A}}_2 n\log n +  \widehat{\mathcal{A}}_3 n +  \widehat{\mathcal{A}}_4 \sqrt n +  \widehat{\mathcal{A}}_5 \log n +  \widehat{\mathcal{A}}_6 + \OO(n^{-1}),
\end{align}
where 
\begin{gather*}
 \widehat{\mathcal{A}}_1 = -1 ,\qquad  \widehat{\mathcal{A}}_2= \frac12, \qquad  \widehat{\mathcal{A}}_3 = \frac{\log (4\pi)}{2}-2-2\alpha-2c, \qquad  \widehat{\mathcal{A}}_4= 0, \qquad  \widehat{\mathcal{A}}_5= \alpha^2+\frac{\alpha}{2}+c^2+\frac{c}{2} +\frac5{12},
\\
\begin{split}
 \widehat{\mathcal{A}}_6 &=  \log(2\pi)  -\frac{13}{24}+4\zeta'(-1)  
- (\alpha+c) (\alpha+c+\tfrac{3}{2}-\log(2\pi)) -\log \Big( G(\alpha+1) G(\alpha+\tfrac32) G(c+1)G(c+\tfrac32)  \Big).
\end{split}
\end{gather*}
\end{itemize}
\end{prop}

For the complex case, we also refer to \cite{FF11,Kl14}. See also \cite{BH19} and references therein for the asymptotic behaviours of the logarithmic energy of the configurations.

Let us stress here that $\mathcal{A}_4=\widehat{\mathcal{A}}_4=0$. Such an absence of the $\OO(\sqrt{n})$-term is a general feature of the free energy expansion for $\beta=2$ ensembles, see \cite{BKS23,ACC23} for the rotationally symmetric case and \cite{BSY24} for a non-rotationally symmetric example. 
Furthermore, note that for $\alpha=c=0$, we have
\begin{equation}
\mathcal{A}_5 \Big|_{ \alpha=c=0 } = \frac12-\frac{\chi}{12}, \qquad \widehat{\mathcal{A}}_5  \Big|_{ \alpha=c=0 }  = \frac{1}{2}- \frac{\chi}{24}, 
\end{equation}
where $\chi=2$ is the Euler index of the droplet, in our case, the Riemann sphere. This form of the coefficient of the $\log n$ term was introduced in the work of Jancovici et al. \cite{JMP94,TF99} and is expected to hold for more general droplets \cite{BKS23,ACC23,BSY24}.
Note further that for $\alpha=c=0$, we have  
\begin{equation}
\mathcal{A}_6 \Big|_{ \alpha=c=0 }= \frac{\log(2\pi)}{2} -\frac{1}{12} +\chi \,\zeta'(-1), \qquad  \widehat{\mathcal{A}}_6 \Big|_{ \alpha=c=0 } = \frac{\log(2\pi)}{2} -\frac{13}{24}+ \frac{\chi}{2} \Big( \frac{5 \log 2}{12}+ \zeta'(-1) \Big). 
\end{equation}
Here, again the $\chi$-dependent terms are expected to be universal, see \cite{BKS23}. Other than these $\chi$-dependent terms, the other parts of the $\OO(1)$-term are related to the conformal geometric properties of the equilibrium measure, see e.g. the introduction of \cite{BSY24}.
For the ensembles on the sphere with general radially symmetric potentials, the associated free energy expansions will be addressed in a forthcoming work. 

Contrary to the regular case, when $V$ contains certain singularities, there has been less understanding of the asymptotic expansions. 
Among the various types of singularities one may consider, the following types are particularly noteworthy due to various motivations and have been extensively investigated in the literature.

\begin{itemize}
    \item \textbf{(Jump type singularity)} This is the case where the weight function has a discontinuity. This type of singularity naturally arises in the context of the moment generating function of the disc counting function, see  \cite{Ch22,ACCL24,ACCL23,ABES23,CL23,ABE23} and references therein for recent progress. In particular, for the spherical ensembles, this has been addressed in a recent work \cite{Mo24}. 
    \smallskip 
    \item \textbf{(Root type singularity)} The pointwise root type singularity is equivalent to the point charge insertions. This type of singularities also finds application in the context of moments of characteristic polynomials \cite{DS22,WW19,BSY24}.   
    We also refer to \cite{BBLM15,KLY23,LY17,LY19,LY23,BKP23,BGM18,BGM17} for extensive studies on the associated orthogonal polynomials. 
  On the other hand, the circular (or non-isolated in general) root type singularity naturally arises in the context of the truncations of Haar unitary matrices \cite{ZS99,ACM23} or finite-rank perturbation \cite{FK99}.
      \smallskip 
    \item \textbf{(Hard wall type singularity)} This is the type we are exploring in our present work, namely the potential takes the form \eqref{V gap gen}. 
    While one might consider it as another jump type singularity, unlike the one seen in counting statistics, in the context of gap probabilities, the potential reveals an extreme discontinuity as it takes the value $+\infty$ in a certain region. 
    As far as we are aware, the only instance prior to our current work where the partition function expansion, up to the constant term, was achieved for two-dimensional point processes, is the work \cite{Ch23} of Charlier. 
\end{itemize}
   
In the context we have discussed, our main results can then be stated as free energy expansions of the partition functions \eqref{ZN V general complex} and \eqref{ZN V general symplectic}, associated with the potential \eqref{potential spherical}, which exhibits both root and hard wall type singularities.

\begin{cor}[\textbf{Free energy expansions with singularities}] \label{Cor_free}
Let 
\begin{equation}
\mathsf{A}_j=A_j+\mathcal{A}_j, \qquad \widehat{\mathsf{A}}_j=\widehat{A}_j+\widehat{\mathcal{A}}_j, \qquad (j=1,2,\dots,6), 
\end{equation}  
where $A_j$ and $\widehat{A}_j$ are given in Theorems~\ref{thm:1.1} and ~\ref{thm:1.2} respectively, and $\mathcal{A}_j$ and $\widehat{\mathcal{A}}_j$ are given in Proposition~\ref{Prop_free energy soft}. 
Let $V=Q^R$, where $Q^R$ is given by \eqref{potential gap}. Then as $n \to \infty$, we have 
\begin{align} 
\log Z_{n,V}^{\mathbb C} &= \mathsf{A}_1 n^2 + \mathsf{A}_2 n\log n +\mathsf{A}_3 n + \mathsf{A}_4 \sqrt n + \mathsf{A}_5 \log n + \mathsf{A}_6 + \OO(n^{-\frac1{12}}), 
\\
\log Z_{n,V}^{\mathbb H} &= \widehat{\mathsf{A}}_1 n^2 + \widehat{\mathsf{A}}_2 n\log n +\widehat{\mathsf{A}}_3 n + \widehat{\mathsf{A}}_4 \sqrt n + \widehat{\mathsf{A}}_5 \log n + \widehat{\mathsf{A}}_6 + \OO(n^{-\frac1{12}}).
 \end{align}
\end{cor}
\begin{proof}
Note that by the definition of $Q^R$ in \eqref{potential gap}, we have
\begin{equation} \label{gap partition rel}
\PPP_n^\C(R;\al,c) = \frac{ Z_{n,V}^{\mathbb C} }{ Z_n^{ \mathbb C } }, \qquad \PPP_n^\H(R;\al,c) = \frac{ Z_{n,V}^{\mathbb H} }{ Z_n^{ \mathbb H } }.
\end{equation}
Then the corollary immediately follows from Theorems~\ref{thm:1.1}, ~\ref{thm:1.2} and Proposition~\ref{Prop_free energy soft}. 
\end{proof}

Contrary to the regular case stated in Proposition~\ref{Prop_free energy soft} or in e.g. \cite{BKS23,BSY24,ACC23}, in Corollary~\ref{Cor_free}, a non-trivial coefficient of the $\OO(\sqrt{n})$-term does arise when considering the hard wall constraints of the potential. This is due to the hard wall constraints of the potential.

\subsubsection{Large gap probabilities for one-dimensional point processes} \label{subsection_1D process}

Comparing with two-dimensional point processes, there have been far more investigations into the gap probabilities for one-dimensional point processes. For the reader's convenience, let us list some of the literature. 

Some early works where the leading order term of large gap probabilities for the eigenvalues of classical Hermitian random matrices include \cite{BDG01,DM06,DM08,KC10,MS14,MV09,VMB07,FW12}. Some of these works also find interesting connections with spin glass model as well as free fermionic systems.
As in the two-dimensional cases, the leading order can be obtained using potential-theoretic computations, sometimes referred to as the Coulomb gas method. 

For determinantal point processes, large gap problems can be interpreted as asymptotic expansions for structured determinants. 
For one-dimensional point processes, the Riemann-Hilbert approach has been particularly successful in obtaining precise asymptotic behaviours. (On the other hand, in dimension two, developing a Riemann-Hilbert approach to large gap problems breaking the radial symmetry is still an outstanding open problem.)
However, multiplicative constants of large gap problems are in general very difficult to obtain and are still open in some cases (see below).

\begin{itemize}
    \item \textbf{(Hankel determinants)} The form of the asymptotic expansion depends on the external potential of the Hermitian random unitary matrices, particularly on their edge behaviours of the associated equilibrium measure. For the generalised Gaussian unitary ensembles, a general result can be found in \cite{DS17,CD18}, and for the Laguerre and Jacobi types, these can be found in \cite{CG21}. Furthermore, the large gap asymptotics of the Muttalib-Borodin model with constant potential involve a generalisation of Hankel determinants and were established in \cite{Ch22a}.
    \smallskip 
    \item \textbf{(Toeplitz determinants)} For the circular unitary ensemble, the probability of an arc being empty of eigenvalues, up to a multiplicative constant, was obtained in \cite{Wi71}. This was later extended in \cite[Proposition 1.1.3]{Fa17} to situations where the gap region consists of multiple arcs of the same length symmetrically located on the unit circle, and higher order correction terms were obtained in \cite{Ma23}.   
    \smallskip 
    \item \textbf{(Fredholm determinants)} In this case, the asymptotic expansions have been investigated for different limiting point processes.
    \begin{itemize}
        \item \emph{(Sine point process)} After the asymptotics of gap probabilities on a single interval were conjectured in \cite{Dy76}, they were proved simultaneously and independently in \cite{Eh06,Kra04}. The case of several interval gaps was investigated in \cite{Wi95}, where later the results were made more explicit in \cite{DIZ97}. In particular, the oscillations were expressed in terms of Jacobi theta functions in \cite{DIZ97} for the first time. Nonetheless, the multiplicative constant was still missing in \cite{DIZ97}. For the two-interval case, this constant was recently obtained in \cite{FK24}.
        \smallskip 
        \item \emph{(Airy point process)} The conjecture on the large gap asymptotics made in \cite{TW94} was shown in \cite{DIK08, BBD08} up to and including the multiplicative constant. For the case of two intervals, this was obtained in \cite{BCL22} without a multiplicative constant, which was later obtained in \cite{KM24}. The gap probabilities in the bulk were also obtained in \cite{BCL22a}.
           \smallskip 
        \item \emph{(Bessel point process)} The large gap asymptotics on a single interval were obtained in \cite{TW94a} up to, but not including, the multiplicative constant. The multiplicative constant conjectured in \cite{TW94a} was later established in \cite{Eh10,DKV11} using different approaches. Recently, the case of several intervals has also been addressed in \cite{BCL23} with the multiplicative constants left undetermined.
    \end{itemize}
  Beyond the standard limiting point processes, there have been further developments on other types of point processes. These include the Wright's generalised Bessel and Meijer-$G$ point processes \cite{CGS19,CLS21,CLS21a}; the Pearcey point process \cite{DXZ21}; the hard edge Pearcey point process  \cite{YZ24a}; the tacnode point process \cite{YZ24}; and the Freud point process \cite{CKM23}.
\end{itemize}

\subsection*{Organisation of the paper} The rest of this paper is organised as follows. 
In Section~\ref{Section_Outline}, we provide an outline of proofs together with necessary preliminaries. Section~\ref{Section_complex} is devoted to the analysis of the complex case, where we prove Theorem~\ref{thm:1.1}. Section~\ref{Section_symplectic} is for the symplectic counterpart, where we provide the proof of Theorem~\ref{thm:1.2}. Additionally, this article contains an appendix where we compile some tail asymptotic behaviours of the binomial distributions.

\section{Outline of proofs}  \label{Section_Outline}

In this section, we provide the preliminaries and outline of proofs of our main results. 

The key ingredients for the proofs are given as follows. 

\begin{itemize}
    \item We first obtain the expression of gap probabilities in terms of products of incomplete beta functions (Lemma~\ref{Lem_finite expression}). 
    This follows from the determinantal/Pfaffian structures of the underlying models and recent development \cite{AEP22} on the planar skew-orthogonal polynomials.  
    \smallskip 
    \item We make use of the uniform asymptotic behaviours of the incomplete beta function (Lemma~\ref{beta}). Although this formula is provided in Temme's book \cite[Section 11.3.3.2]{Te96}, it turns out that there is a typo that can be corrected. We note that the same typo appears in \cite[Eq.(8.18.9)]{NIST}. 
    \smallskip 
    \item The remainder of the proof entails the asymptotic analysis of the finite-$n$ expression utilising the uniform asymptotic behaviours of the incomplete beta function, see Figures~\ref{Fig_sum division} and ~\ref{Fig_sum division H}. This constitutes the main aspect of our proofs, and we apply the method developed by Charlier \cite{Ch23,Ch22} to our model.
\end{itemize}

%\medskip 

\begin{figure}[h!]

\begin{center}
\begin{tikzpicture}[scale = 3.75]
\draw[purple, dotted] (0,0) to [out=35,in=145] node[below] {$\log \PPP_n^{ \mathbb{C} }=\sum_{k=0}^{n-1}(*)$} (4,0);

\filldraw[blue] (0,0) circle (0.5pt) node[above left]{$0$};
\filldraw[red] (1,0) circle (0.5pt) node[above]{$k_-$};
\filldraw[red] (2,0) circle (0.5pt) node[above]{$k_+$};
\filldraw[blue] (3,0) circle (0.5pt) node[above]{$k_M$};
\filldraw[blue] (4,0) circle (0.5pt) node[above right]{$n-1$};
 
\filldraw[red] (0,-0.5) circle (0.5pt) node[below left]{$k_-$};
\filldraw[red] (4/3,-0.5) circle (0.5pt) node[below]{$g_-$};
\filldraw[red] (8/3,-0.5) circle (0.5pt) node[below]{$g_+$};
\filldraw[red] (4,-0.5) circle (0.5pt) node[below right]{$k_+$};

\draw[red, dashed] (1,0) -- (0,-0.5);
\draw[red, dashed] (2,0) -- (4,-0.5);

\draw (0,0) -- ++(1,0)  node[midway,above] {$S_0$ (Lem.~\ref{S0})};
\draw (1,0) -- ++(1,0)  node[midway,above] {$S_1$ (Lem.~\ref{S1})};
\draw (2,0) -- ++(1,0)  node[midway,above] {$S_2$ (Lem.~\ref{S2})};
\draw (3,0) -- ++(1,0)  node[midway,above] {$S_3$ (Lem.~\ref{S3})};

\draw (0,-0.5) -- ++(4/3,0)  node[midway,below] {$S_1^{(1)}$ (Lem.~\ref{S1^(1)})};
\draw (4/3,-0.5) -- ++(4/3,0)  node[midway,below] {$S_1^{(2)}$ (Lem.~\ref{S1^(2)})};
\draw (8/3,-0.5) -- ++(4/3,0)  node[midway,below] {$S_1^{(3)}$ (Lem.~\ref{S1^(3)})}; 
\end{tikzpicture}
\end{center}
    \caption{Illustration of the divisions of the summations in \eqref{sum division S0123} and \eqref{sum division S1 123}, together with lemmas on the asymptotic expansions in each sub-summation.}
    \label{Fig_sum division}
\end{figure}
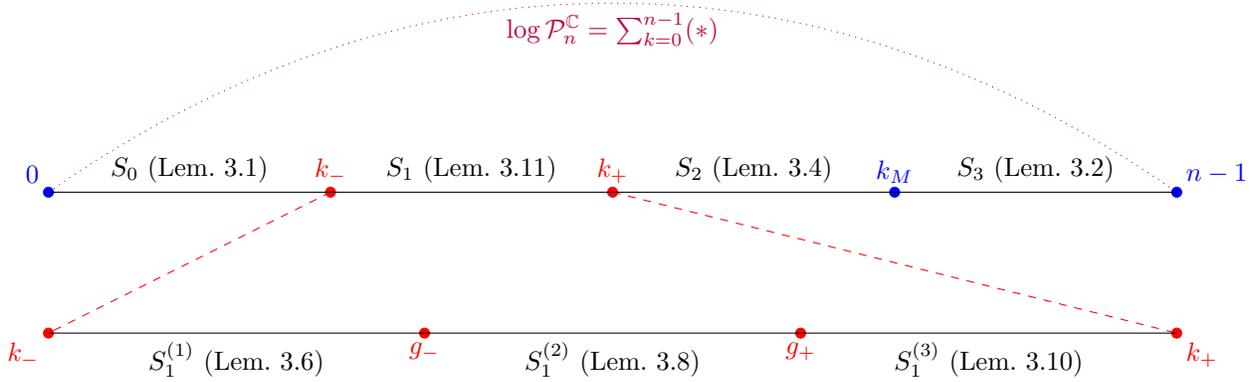

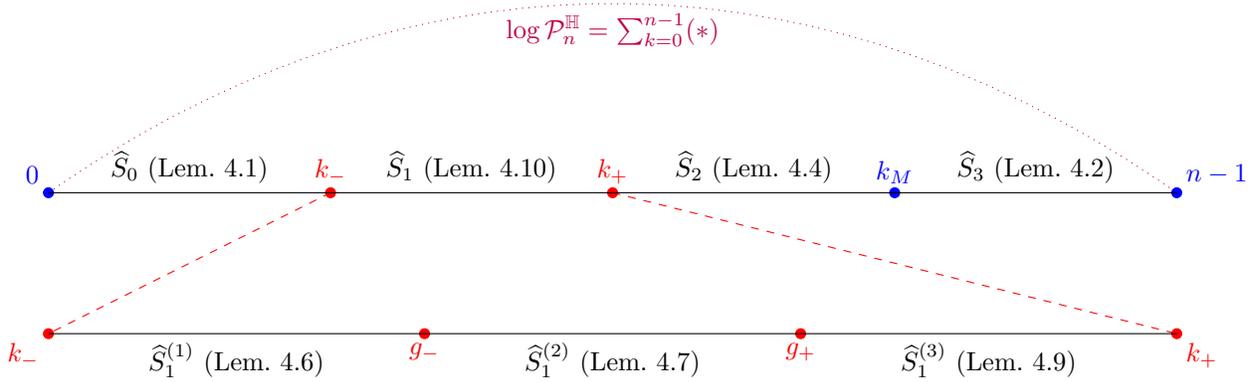
\begin{figure}[h!]

\begin{center}
\begin{tikzpicture}[scale = 3.75]
\draw[purple, dotted] (0,0) to [out=35,in=145] node[below] {$\log \PPP_n^{ \mathbb{H} }=\sum_{k=0}^{n-1}(*)$} (4,0);

\filldraw[blue] (0,0) circle (0.5pt) node[above left]{$0$};
\filldraw[red] (1,0) circle (0.5pt) node[above]{$k_-$};
\filldraw[red] (2,0) circle (0.5pt) node[above]{$k_+$};
\filldraw[blue] (3,0) circle (0.5pt) node[above]{$k_M$};
\filldraw[blue] (4,0) circle (0.5pt) node[above right]{$n-1$};
 
\filldraw[red] (0,-0.5) circle (0.5pt) node[below left]{$k_-$};
\filldraw[red] (4/3,-0.5) circle (0.5pt) node[below]{$g_-$};
\filldraw[red] (8/3,-0.5) circle (0.5pt) node[below]{$g_+$};
\filldraw[red] (4,-0.5) circle (0.5pt) node[below right]{$k_+$};

\draw[red, dashed] (1,0) -- (0,-0.5);
\draw[red, dashed] (2,0) -- (4,-0.5);

\draw (0,0) -- ++(1,0)  node[midway,above] {$\wh S_0$ (Lem.~\ref{H:S0})};
\draw (1,0) -- ++(1,0)  node[midway,above] {$\wh S_1$ (Lem.~\ref{H:S1})};
\draw (2,0) -- ++(1,0)  node[midway,above] {$\wh S_2$ (Lem.~\ref{H:S2})};
\draw (3,0) -- ++(1,0)  node[midway,above] {$\wh S_3$ (Lem.~\ref{H:S3})};

\draw (0,-0.5) -- ++(4/3,0)  node[midway,below] {$\wh S_1^{(1)}$ (Lem.~\ref{H:S1^(1)})};
\draw (4/3,-0.5) -- ++(4/3,0)  node[midway,below] {$\wh S_1^{(2)}$ (Lem.~\ref{H:S1^(2)})};
\draw (8/3,-0.5) -- ++(4/3,0)  node[midway,below] {$\wh S_1^{(3)}$ (Lem.~\ref{H:S1^(3)})}; 
\end{tikzpicture}
\end{center}
    \caption{The analogue of Figure~\ref{Fig_sum division} for the symplectic counterpart, i.e. illustration of the summations in \eqref{H:sum division S0123} and \eqref{H:sum division S1 123}. }
    \label{Fig_sum division H}
\end{figure}

Before introducing the strategy in more detail, let us first provide the proof of Proposition~\ref{Prop_free energy soft}.

\begin{proof}[Proof of Proposition~\ref{Prop_free energy soft}]
By using the definition of the Barnes $G$-function, one can rewrite \eqref{partition product expression} as 
\begin{align}
 Z_n^{ \mathbb{C} } &=   \frac{ n! }{  \Gamma(n+\alpha+c+1)^n  }  \frac{ G(n+c+1) }{ G(c+1) } \frac{ G(n+\alpha+1) }{ G(\alpha+1) },
 \\
 Z_n^{ \mathbb{H} } &=  n! \Big( \frac{2^{2n+2\alpha+2c+1} }{  \pi \Gamma(2n+2\alpha+2c+2)  } \Big)^n \frac{ G(n+c+1) }{ G(c+1) } \frac{ G(n+c+3/2) }{ G(c+3/2) } \frac{ G(n+\alpha+1) }{ G(\alpha+1) } \frac{ G(n+\alpha+3/2) }{ G(\alpha+3/2) }. 
\end{align}
Then the desired asymptotic expansion follows from the well-known asymptotic behaviours of the gamma function \cite[Eq.(5.11.1)]{NIST} 
\begin{equation} \label{log N!}
\log N!= N \log N-N+\frac12 \log N+\frac12 \log(2\pi)+\frac1{12N}+\mathcal{O}\Big(\frac{1}{N^3}\Big), \qquad (N \to \infty)
\end{equation}
and of the Barnes $G$-function \cite[Eq.(5.17.5)]{NIST}
\begin{align}
\begin{split} \label{Barnes G asymp}
\log G(z+1) =\frac{z^2 \log z}{2} -\frac34 z^2+\frac{ \log(2\pi) z}{2}-\frac{\log z}{12}+\zeta'(-1)-\frac{1}{240z^2}+\mathcal{O}\Big( \frac{1}{z^4} \Big), \qquad (z \to \infty). 
\end{split}
\end{align}
\end{proof}

Recall that the regularised incomplete beta function $I_x$ is given by 
\begin{equation}
I_x(a,b):= \frac{B_x(a,b)}{B(a,b)}, \qquad   B_x(a,b):=\int_0^x t^{a-1}(1-t)^{b-1}\,dt 
\end{equation}
where 
\begin{equation}
B(a,b)= B_1(a,b):= \frac{\Gamma(a) \Gamma(b) }{ \Gamma(a+b) }
\end{equation}
is the beta function, see \cite[Chapter 8]{NIST}.  
Note that by definition, we have  
\begin{equation}
I_x(a,b)=1- I_{1-x}(b,a). 
\end{equation}
We also note that for integer values $n \ge m$, we have 
\begin{equation}
I_x(m,n-m+1)= \sum_{j=m}^n \binom{n}{j} x^j (1-x)^{n-j}. 
\end{equation}
This in turn implies that the regularised incomplete beta function is the cumulative distribution function of the binomial distribution, i.e. for $X \sim B(n,p)$, 
\begin{equation} \label{beta binomial cdf}
\mathbb{P}(X \le k) = I_{1-p}(n-k,k+1)= 1-I_p(k+1,n-k). 
\end{equation}

In the following lemma, we express the gap probabilities in terms of the incomplete beta functions. 

\begin{lem}[\textbf{Evaluation of the gap probabilities at finite-$n$}] \label{Lem_finite expression}
Let 
\begin{equation} \label{def of x}
x= \frac{1}{1+R^2}. 
\end{equation}
Then we have 
\begin{align}
\mathcal{P}_n^\C(R; \alpha,c) &= \prod_{k=0}^{n-1} I_x(k+\alpha+1,n+c-k)=  \prod_{k=0}^{n-1} I_x(n+\alpha-k,k+c+1) ,  \label{C:gap_prob}
\\
\mathcal{P}_n^{ \mathbb{H} }(R; \alpha,c) &= \prod_{k=0}^{n-1}  I_{x}(2k+2\alpha+2,2n+2c-2k) = \prod_{k=0}^{n-1}  I_{x}(2n+2\alpha-2k,2k+2c+2). \label{H:gap_prob}
\end{align}
\end{lem}
\begin{proof}
Recall that $Q_n^R$ is given by \eqref{potential gap}. 
Let 
\begin{align}
h_{k,m}:= \int_\C |z|^{2k}\,e^{-m Q_n(z)}\,dA(z),
\qquad 
h_{k,m}(R) := \int_{ \mathbb{D}(0,R)^c } |z|^{2k}\,e^{-m Q_n(z)}\,dA(z).
\end{align}
It follows from the general theory of determinantal and Pfaffian point processes that the associated partition functions \eqref{ZN V general complex} and \eqref{ZN V general symplectic} can be expressed in terms of (skew)-norms of the associated (skew)-orthogonal polynomials.
Combining this with \eqref{gap partition rel}, we have 
\begin{equation}
\mathcal{P}^\C_n(R) = \prod_{k=0}^{n-1} \frac{ h_{k,n}(R) }{ h_{k,n} }, \qquad \mathcal{P}^{ \mathbb{H} }_n(R) = \prod_{k=0}^{n-1} \frac{ h_{2k+1,2n}(R) }{ h_{2k+1,2n} },
\end{equation}
see e.g. \cite[Eq.(1.37)]{BKS23}. In particular, we have used the construction of the skew-orthogonal polynomials \cite[Corollary 3.3]{AEP22}. On the other hand, by using the Euler beta integrals, we have 
\begin{align*}
h_{k,m}(R)  & = 2 \int_R^\infty  \frac{ r^{2k+1+2c \frac{m}{n}} }{ (1+r^2)^{ (n+\alpha+c+1) \frac{m}{n} } } \,dr
=    \int_R^\infty   \Big( \frac{r^2}{1+r^2} \Big)^{k+c \,\frac{m}{n}} \Big( \frac{1}{1+r^2} \Big)^{ (n+\alpha+1)\,\frac{m}{n}-k-2 }    \frac{2r}{(1+r^2)^2} \,dr 
\\
&= \int_0^x t^{ (n+\alpha+1)\,\frac{m}{n}-k-2 }  (1-t)^{k+c \,\frac{m}{n}} \,dt = B_x( (n+\alpha+1)\,\tfrac{m}{n}-k-1, k+c \,\tfrac{m}{n}+1 ), 
\end{align*}
which gives 
\begin{align*}
\frac{ h_{k,n}(R) }{ h_{k,n}  } =  I_x( n+\alpha-k, k+c +1 ), \qquad \frac{ h_{2k+1,2n}(R) }{ h_{2k+1,2n}  } =  I_x( 2n+2\alpha-2k, 2k+2c+2 ). 
\end{align*}
This completes the proof. 
\end{proof}

To analyse the expressions in Lemma~\ref{Lem_finite expression}, it is convenient to introduce the notation 
\begin{equation} \label{def of N}
N  = \begin{cases}
n+\al+c &\textup{ for the complex case,}
\smallskip 
\\
n+\al+c+\frac12 & \textup{ for the symplectic case.}
\end{cases}
\end{equation}
We further define 
\begin{equation} \label{k_pm}
    k_- = \floor{\del N}-\al-1, \qquad k_+ = \ceil{(1-\del)N}-\al-1, 
\end{equation}
and 
\begin{equation} \label{def of kM}
k_M = 
\begin{cases}
N-\ceil{M}-\al - 1  &\textup{ for the complex case,}
\smallskip 
\\ 
 N-\ceil{M}-\al - \frac 32 & \textup{ for the symplectic case.} 
\end{cases}
\end{equation}
Here, we set 
\begin{equation} \label{def of M}
M = N^{\frac13} 
\end{equation}
and choose $\del \in (0,\min\{x,1-x\})$ small enough so that $0<k_-<k_+<k_M<N$ holds for large $N$.
 
By using \eqref{k_pm} and \eqref{def of kM}, we divide the sum in \eqref{C:gap_prob} by the following four parts:
\begin{equation} \label{sum division S0123}
\log \PPP_n^\CC(R;\al,c) = S_0 + S_1 + S_2 + S_3,
\end{equation}
where
\begin{equation} \label{def of S0123}
\begin{alignedat}{2}
        S_0 &= \sum_{k=0}^{k_-} \log I_x(k+\al+1,n+c-k), \qquad &&S_1 = \sum_{k=k_-+1}^{k_+-1} \log I_x(k+\al+1,n+c-k), 
        \\
        S_2 &= \sum_{k=k_+}^{k_M} \log I_x(k+\al+1,n+c-k), \qquad &&S_3 = \sum_{k=k_M+1}^{n-1} \log I_x(k+\al+1,n+c-k).
\end{alignedat}
\end{equation}
Similarly, we divide the sum in \eqref{H:gap_prob} by  
\begin{equation} \label{H:sum division S0123}
    \log \PPP_n^{\mathbb H}(R;\al,c) = \wh{S}_0 + \wh S_1 + \wh S_2 + \wh S_3,
\end{equation}
where
\begin{equation} \label{H:def of S0123}
    \begin{alignedat}{2}
        \wh S_0 &= \sum_{k=0}^{k_-} \log I_x(2k+2\al+2,2n+2c-2k), \quad &&\wh S_1 = \sum_{k=k_-+1}^{k_+-1} \log I_x(2k+2\al+2,2n+2c-2k), \\
        \wh S_2 &= \sum_{k=k_+}^{k_M} \log I_x(2k+2\al+2,2n+2c-2k), \quad &&\wh S_3 = \sum_{k=k_M+1}^{n-1} \log I_x(2k+2\al+2,2n+2c-2k).
\end{alignedat}
\end{equation}
See Figures~\ref{Fig_sum division} and ~\ref{Fig_sum division H}. 
As previously mentioned, the divisions in \eqref{def of S0123} and \eqref{H:def of S0123} are chosen in a way to minimise cumulative errors. For the asymptotic behaviours of $S_0$, $S_2$, and $S_3$, as well as $\wh{S}_0$, $\wh{S}_2$, and $\wh{S}_3$, we can apply various tail behaviours of the binomial distribution (see Appendix~\ref{Appendix_binomial tails}) together with a version of the Euler-Maclaurin formula (Lemma~\ref{lem:3.3}).
The most technical part involves $S_1$ and $\wh{S}_1$, which, from the viewpoint of the asymptotic behaviour of the incomplete beta function, are referred to as the \emph{symmetric} regime. In this case, it is necessary to further divide the summations as in \eqref{sum division S1 123} and \eqref{H:sum division S1 123}. For this regime, the following uniform asymptotics of the incomplete beta function play a key role.

\begin{lem}[\textbf{Uniform asymptotics of the incomplete beta function}] \label{beta} 
Let $x_0 = a/N$. Then as $N \to \infty$, we have
\begin{equation} \label{Incomplete beta Temme}
I_x(a,N-a) = \frac12 \erfc\Big(-\eta \sqrt{\frac{N}{2}}\Big)+R_a(\eta), 
\end{equation}
where 
\begin{equation}
R_a(\eta) = e^{-\frac12 \eta^2 N} \frac{1}{2\pi i}\int_{-\infty}^\infty e^{-\frac12 N u^2} g(u) \, du.
\end{equation}
Here, $\eta$ is given by
\begin{equation} \label{def of eta in lemma}
-\frac12 \eta^2 = x_0 \log\Big(\frac{x}{x_0}\Big) + (1-x_0) \log\Big(\frac{1-x}{1-x_0}\Big), \qquad \frac{\eta}{x-x_0}>0,
\end{equation}
and $g$ is given by 
\begin{equation}
g(u) = \frac{dt}{du}\frac{1}{t-x} - \frac{1}{u+i\eta},
\end{equation}
where the variables $t$ and $u$ are correlated by the bijection 
\begin{equation}
\frac12 u^2 = x_0 \log \Big( \frac{t}{x_0} \Big) +(1-x_0)\log \Big( \frac{1-t}{1-x_0} \Big), \qquad \frac{\Im t}{u}>0, \qquad t \in \mc L,
\end{equation}
on  
\begin{equation}
\mc L= \Big\{ t =  \frac{\sin(\rho\sig)}{\sin[(1+\rho)\sig]}e^{i\sig}: |\sig|<\frac{\pi}{1+\rho} \Big\} \mapsto u \in \RR, \qquad \Big( \rho = \frac{x_0}{1-x_0} \Big). 
\end{equation}

Furthermore, we have
\begin{equation}
R_a(\eta) \sim \frac1{\sqrt{2\pi N}} e^{-\frac12\eta^2N}\sum_{k=0}^\infty \frac{(-1)^k c_k(\eta)}{N^k}
\end{equation}
uniformly for $x \in (0,1)$ and $x_0 \in [\del,1-\del]$, where $0<\del < \frac 12$ is an arbitrarily fixed small constant. 
Here, the coefficients $c_n(\eta)$ is given by 
\begin{equation}
c_n(\eta) = -i (2n-1)!! g_{2n}, \qquad g_n= \frac{1}{n!} \frac{d^n}{du^n} g(u) \Big|_{u=0}.
\end{equation}
In particular, the first two coefficients $c_0$ and $c_1$ are given by
\begin{align}
    c_0(\eta) &= \frac1\eta - \frac{\sqrt{x_0(1-x_0)}}{x-x_0}, \label{def of c0}
    \\
    c_1(\eta) &= -\frac{1}{\eta^3}+ \frac{1-13x_0+13x_0^2}{12\sqrt{x_0(1-x_0)}}\frac{1}{x-x_0}+\frac{(1-2x_0)\sqrt{x_0(1-x_0)}}{(x-x_0)^2}+\frac{(x_0(1-x_0))^{3/2}}{(x-x_0)^3}. \label{def of c1}
\end{align}
All of the $c_k(\eta)$ are analytic at $\eta = 0$.
\end{lem}

\begin{rem} \label{Rem_typo}
In \cite[Section 11.3.3.2]{Te96}, the leading term of the right-hand side of \eqref{Incomplete beta Temme} is incorrectly written as $\frac12 \erfc (-\eta \sqrt{\frac{N-a}{2}} )$. 
The corrected version follows from the integral
    \[
    e^{-\frac12 \eta^2 N}\frac{1}{2\pi i}\int_{-\infty}^\infty e^{-\frac12 N u^2}\frac{1}{u+i\eta}\,du = \frac12 \erfc\Big(-\eta \sqrt{\frac N2}\Big). 
    \] 
This typo still appears in \cite[Eq.(8.18.8)]{NIST}, where the first term of the right-hand side there should be corrected as 
\begin{equation}
I_x(a,b) \sim \frac12 \erfc\Big(- \eta \sqrt{ \frac{a+b}{2} } \Big). 
\end{equation}
The subleading terms in \cite{Te96,NIST} are correct. 
We note that the corrected version \eqref{Incomplete beta Temme} can also be easily verified numerically.
\end{rem}

\section{Gap probabilities of the complex ensemble} \label{Section_complex}

In this section, we prove Theorem~\ref{thm:1.1}. As previously mentioned, we shall use the splitting \eqref{sum division S0123}. In Subsection~\ref{Subsection_S023}, we perform the analysis of $S_0$, $S_2$ and $S_3$. Subsection~\ref{Subsection_S1} is devoted to the analysis of a more complicated part $S_1$, where we further need the division \eqref{sum division S1 123}.

\subsection{Analysis of $S_0$, $S_2$ and $S_3$}  \label{Subsection_S023}

We provide the asymptotic formulas of the $S_j \, (j=0,2,3)$ in increasing order of difficulty: $S_0$ (Lemma~\ref{S0}), $S_3$ (Lemma~\ref{S3}), and then $S_2$ (Lemma~\ref{S2}).

To simplify the notations, let us write
\begin{equation} \label{def of m_k}
    m_k := k+\al+1 
\end{equation}
so that the summand of \eqref{def of S0123} can be rewritten as
\[
I_x(k+\al+1,n+c-k) = I_x(m_k,N+1-m_k).
\]

Recall that $S_0$ is given by \eqref{def of S0123}. 

\begin{lem}[\textbf{Asymptotic behaviour of $S_0$}] \label{S0}
    There exists $C>0$ such that
    \begin{equation}
    S_0 = \OO(e^{-CN}), \qquad \text{as } N\to \infty.
    \end{equation}
\end{lem}

\begin{proof} %[Proof of Lemma~\ref{S0}] 
    Recall that $k_-$ is given by \eqref{k_pm}. Thus for $0 \leq k \leq k_-$, we have $\al+1 \leq m_k \leq \floor{\del N}$. Let $X \sim B(N,x)$, where $B$ is the binomial distribution. Here $N$ and $x$ are given by \eqref{def of N} and \eqref{def of x}.
    Note that by \eqref{beta binomial cdf}, we have
    \begin{equation}
        \log I_x(m_k,N+1-m_k) = \log\Big(1- \PP(X\leq m_k-1)\Big) = \log\Big(1-\PP(X<m_k)\Big).
    \end{equation}
    For every $\al+1 \leq m \leq \floor{\del N}$, by using Lemma \ref{Lemma:2.1}, we have
    \begin{align*}
        \PP(X<m) \leq \PP(X<\del N) \le \frac12 \erfc\Big(\frac{x-\del}{\sqrt{2x}}\sqrt N\Big) \leq \frac 12 e^{-\frac{(x-\del)^2}{2x}N} = \OO(e^{-CN}), %\frac12 e^{-CN}.
    \end{align*}
    where $C$ is independent of $m$. Hence, 
    the summands of $S_0$ in \eqref{def of S0123} are of order $\OO(e^{-CN})$ uniformly, which completes the proof.
\end{proof}

Next, we analyse the sum $S_3$ in \eqref{def of S0123}. 
For this purpose, let us write 
\begin{equation} \label{theta_Mdel}
    \theta_M := \lceil M \rceil - M, \qquad \theta_{\del} := \del N - \lfloor \del N \rfloor ,
\end{equation}
where $M$ is given by \eqref{def of M}. 

\begin{lem}[\textbf{Asymptotic behaviour of $S_3$}] \label{S3}
    As $N\to \infty$, we have
    \begin{equation} \label{S3 asymptotics}
    S_3 = D_3N + D_5\log N + D_6 - \frac{1}{6} \frac{M^3}{N} + \OO\Big(\frac1M+\frac{M^2}{N}+\frac{M^4}{N^2}\Big),
    \end{equation}
    where
    \begin{align*}
        D_3 &= (M + \theta_M - c) \log x,\qquad D_5 = \frac12 (M^2 + (2\theta_M -1) M + \theta_M^2 - \theta_M -c^2 + c), \\
        D_6 &= \frac12\bigg(M^2+(2\theta_M -1)M + \theta_M^2 -\theta_M - c^2 + c\bigg)\log\Big(\frac{1-x}{x}\Big) - \Big(\frac12 M^2 +M\theta_M + \frac12 \theta_M^2 - \frac1{12}\Big) \log M\\
        & \quad  + \frac34 M^2 + \Big(\theta_M - \frac12 \log 2\pi\Big) M - \Big(\frac12 \log 2\pi \Big) \theta_M - \zeta'(-1) + \log G(c+1). 
    \end{align*}
    Here, $\theta_M$ is given by \eqref{theta_Mdel}, $\zeta$ is the Riemann zeta function and $G$ is the Barnes $G$-function.
\end{lem}

\begin{proof}
    By the definitions \eqref{def of m_k}, \eqref{def of kM}, and \eqref{def of N}, as well as \eqref{beta binomial cdf},  
    one can write 
    \begin{equation}
      S_3 = \sum_{\ell = c}^{\ceil{M}-1}\log \PP(Y \leq \ell), \qquad Y \sim B(N,1-x). 
    \end{equation}
    Observe that
    \begin{align}\label{eq:3.2.1}
        \begin{split}
      &\quad  \log \PP(Y \leq \ell) = (N-\ell) \log x + \log \bigg[\sum_{j=0}^\ell \binom{N}{j} (1-x)^jx^{\ell-j}\bigg] 
      \\
        &= (N-\ell)\log x + \log \bigg[\binom N\ell (1-x)^\ell\bigg] +\log\bigg(1+\sum_{j=0}^{\ell-1}\Big(\frac{x}{1-x}\Big)^{\ell-j} \frac{(j+1)\cdots \ell}{(N-\ell+1)\cdots(N-j)}\bigg).
        \end{split}
    \end{align}
    Notice here that the last term in the second line is of order $\OO(\frac \ell N)$. 
    Note also that 
    \begin{equation}\label{eq:3.2.2}
        \log \binom N\ell = \ell \log N + \sum_{j=0}^{\ell-1} \log\Big(1-\frac jN\Big) - \log \ell! = \ell \log N -\frac{1}{N}\frac{\ell(\ell-1)}{2}- \log \ell! + \OO\Big(\frac{\ell^3}{N^2}\Big).
    \end{equation}
   Then, by combining \eqref{eq:3.2.1} and \eqref{eq:3.2.2}, we obtain
    \begin{equation} \label{eq:3.2.3}
        \log \PP(Y \leq \ell) = N \log x + \ell \log N + \ell \log \Big( \frac{1-x}{x} \Big)- \log \ell! - \frac1N \frac{\ell^2}{2} + \OO\Big(\frac{\ell}{N}+\frac{\ell^3}{N^2}\Big).
    \end{equation}
    By taking the sum of \eqref{eq:3.2.3} over $c \leq \ell \leq \ceil M-1$, it follows that
    \begin{align} \label{C:S3_with_ceilM}
    \begin{split}
        S_3 &= N(\ceil M - c) \log x + \Big(\frac{\ceil M^2-\ceil M}{2}-\frac{c^2-c}{2}\Big) \log\Big(\frac{1-x}{x}N\Big)
        \\
        &\quad - \log G(\ceil M +1) + \log G(c+1) - \frac{1}{N}\frac{\ceil{M}^3}{6} + \OO\Big(\frac{M^2}{N}+\frac{M^4}{N^2}\Big).
    \end{split}
    \end{align}
 Then the desired asymptotic formula \eqref{S3 asymptotics} follows from  \eqref{Barnes G asymp} and $\lceil M \rceil = M + \theta_M$. 
\end{proof}

Next, we analyse $S_2$. For this, we need the following version of the Euler-Maclaurin formula taken from \cite[Lemma 3.4]{Ch23}.

\begin{lem}[cf. Lemma 3.4 in \cite{Ch23}] \label{lem:3.3}
    Let $A$, $a_0$, $B$ and $b_0$ be bounded functions of $n \in \ZZ^+$ such that $$ a_n:=An+a_0, \qquad b_n:=Bn+b_0$$ 
    are integer-valued. 
    Furthermore, assume that $B-A$ is positive and bounded away from 0. Let $f$ be a function independent of $n$, and which is $C^3([\min\{\frac{a_n}{n},A\}, \max\{\frac{b_n}{n},B\}])$ for all $n \in \ZZ^+$. Then as $n \to \infty$, we have
    \begin{align}
    \begin{split}
        \sum_{j=a_n}^{b_n} f\Big(\frac jn\Big) &= n\int_A^B f(x) \, dx + \frac{(1-2a_0)f(A)+(1+2b_0)f(B)}{2} \\
        &\quad + \frac{(-1+6a_0-6a_0^2)f'(A)+(1+6b_0+6b_0^2)f'(B)}{12n} \\
        &\quad + \OO\bigg(\frac{\mm_{A,n}(f'')+\mm_{B,n}(f'')}{n^2}+\sum_{j=a_n}^{b_n-1}\frac{\mm_{j,n}(f''')}{n^3}\bigg).    
    \end{split}
    \end{align}
    Here, for any continuous function $g$ on $[\min\{\frac{a_n}{n},A\}, \max\{\frac{b_n}{n},B\}]$,  
    \[\mm_{A,n}(g) := \max_{x \in [\min\{\frac{a_n}{n},A\},\max\{\frac{a_n}{n},A\}]}|g(x)|, \qquad \mm_{B,n}(g) := \max_{x \in [\min\{\frac{b_n}{n},B\},\max\{\frac{b_n}{n},B\}]}|g(x)| .\]
\end{lem}

We now derive the asymptotic behaviour of $S_2$. Compared to $S_0$ and $S_3$, the asymptotics of $S_2$ contribute to the leading order of the total summation \eqref{sum division S0123}.

\begin{lem}[\textbf{Asymptotic behaviour of $S_2$}] \label{S2}
    As $N \to \infty$, we have
    \begin{equation}
        S_2 = C_1N^2 + C_2N\log N +C_3N + C_5\log N + C_6 + \OO\Big(\frac1M+\frac{M}{N}\Big),
    \end{equation}
    where 
    \begin{align*}
        C_1 &= -\int_{M/N}^\del f_1(t) \,dt, \qquad C_2 = -\frac{\del}{2}, \qquad C_3 = \Big(\theta_M - \frac12\Big) f_1\Big(\frac MN \Big) + \Big(\theta_{\del}-\frac12\Big) f_1(\del) - \int_{M/N}^\del f_2(t) \,dt, \\
        C_5 &= \frac12(M+\theta_{\del}+\theta_M -1) -\frac1{12}, \\
        C_6 &= \frac{1-6\theta_M+6\theta_M^2}{12} f_1'\Big(\frac MN\Big) - \frac{1-6\theta_\del+6\theta_\del^2}{12}f_1'(\del) + \Big(\theta_M - \frac12 \Big) f_2\Big(\frac MN\Big) + \Big(\theta_{\del}-\frac12\Big) f_2(\del) \\
        &\quad + \frac1{12}\log M+ \Big[\frac1{12}\Big(\del - \log \Big(\frac{\del}{1-\del}\Big)\Big) - \frac{\del x}{1-x-\del}-x\log\Big(\frac{1-x-\del}{1-x}\Big)\Big],
    \end{align*}
    where $\theta_M$ and $\theta_{\del}$ are given by \eqref{theta_Mdel}, and the functions $f_1$, $f_2$ are defined as
    \begin{equation} \label{def:f1f2}
        f_1(t) := t \log \Big(  \frac t{1-x} \Big)+ (1-t) \log \Big(\frac{1-t}{x}\Big), \qquad f_2(t) := \log \Big(\frac{1-x-t}{(1-x)(1-t)}\Big) + \frac12 \log \Big(2\pi t(1-t) \Big), 
    \end{equation}
    for $0<t<1-x$.
\end{lem}

\begin{proof}
Let $X \sim B(N,x)$. Then by \eqref{k_pm} and \eqref{def of kM}, we have 
\begin{equation} \label{S2_summation}
    S_2 = \sum_{k=k_+}^{k_M} \log \PP(X \geq k+\al+1) = \sum_{j=\ceil M}^{\floor{\del N}}\log \PP(X \ge N-j).
\end{equation}
Let $t = j/N$. By applying Lemma \ref{lem:2.7}, we have
    \[
    \log \PP(X \ge tN) = -Nf_1(t) -\frac12 \log N - f_2(t)
    + \frac{1}{N}f_3(t) + \OO\Big(\frac1{t^3N^3}+\frac1{N^2}\Big),
    \]
    where 
     \begin{equation} \label{def:f3}
        f_3(t) = \frac1{12}-\frac{tx}{(1-x-t)^2}-\frac{1}{12t(1-t)}.
    \end{equation}
    Note that 
    \[
    \sum_{j = \lceil M \rceil}^{\lfloor \del N \rfloor} \OO\Big(\frac{1}{j^3}+\frac{1}{N^2}\Big) = \OO\Big(\frac{1}{M^2}+\frac{1}{N}\Big).
    \]
    By using Lemma \ref{lem:3.3} with $A(N) = M/N$, $a_0(N) = \theta_M$, $B(N) = \del$ and $b_0(N) = -\theta_{\del}$, we obtain
    \begin{align*}
        -N\sum_{j = \lceil M \rceil}^{\lfloor \del N \rfloor} f_1\Big(\frac jN \Big) &= -N^2 \int_{M/N}^\del f_1(t)\, dt + \Big[\Big(\theta_M - \frac12\Big) f_1\Big(\frac MN\Big) + \Big(\theta_\del - \frac12\Big) f_1(\del) \Big] N\\
        &\quad + \frac{1-6\theta_M + 6\theta_M^2}{12}f_1'\Big(\frac MN\Big) - \frac{1-6\theta_{\del}+6\theta_{\del}^2 }{12}f_1'(\del) + \OO\Big(\frac 1{M}\Big), \\
        -\sum_{j = \lceil M \rceil}^{\lfloor \del N \rfloor} f_2\Big(\frac jN\Big) &= -N \int_{M/N}^\del f_2(t)\, dt +\Big(\theta_M-\frac12\Big)f_2\Big(\frac{M}{N}\Big) +\Big(\theta_\del-\frac12\Big)f_2(\del) + \OO\Big(\frac1M + \frac 1N\Big), \\
        \frac 1N \sum_{j = \lceil M \rceil}^{\lfloor \del N \rfloor} f_3\Big(\frac jN\Big) &= \int_{M/N}^\del f_3(t)\, dt + \OO\Big(\frac{1}{M}\Big) \\
        &= \frac1{12}\Big(\del-\log \Big(\frac{\del}{1-\del}\Big)\Big)-\frac{\del x}{1-x-\del}-x\log \Big(\frac{1-x-\del}{1-x}\Big) + \frac{1}{12}\log \Big(\frac MN \Big)+ \OO\Big(\frac1M+ \frac MN\Big).
    \end{align*}
    Combining all of the above, the proof is complete. 
\end{proof}

We conclude this subsection by presenting the asymptotics of $S_0 + S_2 + S_3$, derived from the previous lemmas.

\begin{lem}[\textbf{Asymptotic behaviour of $S_0+S_2+S_3$}] \label{S2+S3}
   As $N \to \infty$, we have
   \begin{equation} \label{S023 asymp}
  S_0+ S_2 + S_3 = E_1N^2 + E_2 N\log N + E_3 N + E_5 \log N + E_6 + \OO \Big(\frac{1}{M}+\frac{M^2}{N}+\frac{M^4}{N^2}\Big),
   \end{equation}
 where 
    \begin{align*}
        E_1 &= \frac12 \log x -\frac12 \del^2 \log \Big( \frac\del {1-x} \Big) + \frac12 (1-\del)^2 \log \Big( \frac{1-\del}{x} \Big) + \frac\del 2, \qquad E_2 = -\frac\del2, 
        \\
        E_3 &= \Big(\frac12 - c\Big) \log x + (1-x-\del) \log \Big( \frac{1-x-\del}{1-x} \Big) +\del - \frac12 \del \log(2\pi\del) - \frac12(1-\del)\log(1-\del) - \Big(\frac12 - \theta_\del\Big) f_1(\del), \\
        E_5 &= \frac{1}{2}\Big(\theta_{\del} + c - c^2\Big) - \frac 5{12}, \\
        E_6 &= \Big(\frac1{12}-\frac{c-c^2}{2}\Big)\log \Big( \frac{x}{1-x} \Big) -\frac14 \log(2\pi)-\zeta'(-1)+\log G(c+1) \\
        &\quad -\frac{1-6\theta_\del+6\theta_\del^2}{12}f_1'(\del) - \Big(\frac12 - \theta_\del\Big)f_2(\del) + \frac1{12}\del - \frac{\del x}{1-x-\del} -x\log \Big( \frac{1-x-\del}{1-x} \Big) -\frac1{12}\log \Big( \frac{\del}{1-\del} \Big) . 
    \end{align*}
    Here, $f_1$ and $f_2$ are defined as \eqref{def:f1f2}.
\end{lem}
\begin{proof}
First, note that since $S_0$ is exponentially small due to Lemma~\ref{S0}, it suffices to compute $S_2 + S_3.$
Note that by definition \eqref{def:f1f2}, we have 
\begin{equation}\label{integral of f1}
\int f_1(t) \, dt  = \frac12 t^2 \log \Big( \frac t{1-x} \Big) - \frac12 (1-t)^2 \log \Big( \frac{1-t}{x} \Big) - \frac t2 , \qquad f_1'(t) = \log \Big( \frac{tx}{(1-t)(1-x)} \Big) ,
  \end{equation}
  and 
    \begin{align}
    \begin{split} \label{integral of f2}
        \int f_2(t) \, dt = -(1-x-t) \log(1-x-t) + t \Big(  \log \Big( \frac{\sqrt{2\pi}}{1-x}\Big)  - 1\Big) 
         + \frac{1}{2} \Big( t \log t +  (1-t) \log(1-t)  \Big) ,
    \end{split}
    \end{align}
    up to integration constants.
    Combining these formulas with Lemma~\ref{S2}, we have
    \begin{align*}
        C_1N^2 &= \bigg[-\frac12 \del^2 \log \Big(\frac{\del}{1-x}\Big)+\frac12(1-\del)^2\log \Big( \frac{1-\del}{x} \Big) +\frac\del2 + \frac12 \log x \bigg]N^2 - (M\log x) N- \frac12 M^2 \log N  
        \\
        &\quad + \bigg[\frac12 M^2 \log M +\frac12 M^2 \log\Big( \frac{x}{1-x} \Big)-\frac34 M^2\bigg] + \frac{M^3}{6N} + \OO\Big(\frac{M^4}{N^2}\Big),\\
        C_3N &=  \bigg[(1-x-\del)\log(1-x-\del) - \del\log \Big( \frac{\sqrt{2\pi}}{1-x} \Big) +\del - \frac12 \del \log \del - \frac12 (1-\del)\log(1-\del)-\Big(\frac12-\theta_\del\Big)f_1(\del) \bigg]N 
        \\
        &\quad + \bigg[\Big(\frac12 -\theta_M\Big)\log x - (1-x)\log(1-x)\bigg] N -M\theta_M \log N + M\theta_M \log M 
        \\
        &\quad + \bigg[\frac12 \log(2\pi) - \theta_M + \Big(\frac12 - \theta_M\Big) \log \Big(\frac{1-x}{x}\Big)\bigg] M + \OO\Big(\frac{M^2}{N}\Big),
        \\
        C_6&= \Big(\frac16-\frac12 \theta_M^2\Big)\log N +\Big(\frac12 \theta_M^2-\frac1{12} \Big)\log M+ \Big(\frac12\theta_M - \frac14 \Big) \log(2\pi)  + \frac{1-6\theta_M+6\theta_M^2}{12}\log \Big( \frac{x}{1-x} \Big) + \frac1{12}\del
        \\ 
        & \quad -\frac1{12}\log\Big(\frac{\del}{1-\del}\Big)- \frac{1-6\theta_\del+6\theta_\del^2}{12}f_1'(\del) - \Big(\frac12 - \theta_\del\Big)f_2(\del) - \frac{\del x}{1-x-\del}-x\log\Big( \frac{1-x-\del}{1-x} \Big) + \OO\Big(\frac MN\Big).
    \end{align*}
   Then after lengthy but straightforward computations, one can observe that all terms depending on $M$ in $D_j$ and $C_j$ from Lemmas \ref{S3} and \ref{S2} cancel each other out, up to an error term $\mathcal{O}(M^{-1} + M^2 N^{-1} + M^4 N^{-2})$. The only terms remaining are those stated in the desired asymptotic formula \eqref{S023 asymp}. 
\end{proof}

\subsection{Analysis of $S_1$}  \label{Subsection_S1}

By Lemma~\ref{S2+S3} in the previous subsection, it remains to derive the asymptotic behaviour of $S_1$. The aim of this subsection is to perform this remaining task and demonstrate Lemma~\ref{S1} below.
This requires to further divide the summation $S_1$. 
First, recall that
\begin{equation}
S_1 = \sum_{k=k_-+1}^{k_+-1} \log I_x(k+\al+1, n+c-k) = \sum_{m=\lfloor \del N\rfloor +1}^{\lceil (1-\del) N \rceil -1} \log I_x(m,N+1-m).
\end{equation}
We denote
\begin{equation} \label{def of N1 M1 x_k}
N_1 = N+1, \qquad M_1 = N_1^{\frac{1}{12}}, \qquad x_k = \frac{m_k}{N_1},
\end{equation}
where $N$ and $m_k$ are given by \eqref{def of N} and \eqref{def of m_k}, respectively.
Define
\begin{equation} \label{g_pm}
g_- = \lceil N_1x - M_1\sqrt{N_1}-\al - 1\rceil, \qquad g_+ = \lfloor N_1x + M_1\sqrt{N_1}-\al - 1\rfloor.
\end{equation}
We split $S_1$ into the following three parts: 
\begin{equation} \label{sum division S1 123}
S_1 = S_1^{(1)}+S_1^{(2)}+S_1^{(3)},
\end{equation} 
where 
\begin{align}
\begin{split}
S_1^{(1)} & = \sum_{k=k_-+1}^{g_--1} \log I_x(m_k,N_1-m_k),  \label{def of S1(123)}
\\
S_1^{(2)} & = \sum_{k=g_-}^{g_+} \log I_x(m_k,N_1-m_k),     
\\
S_1^{(3)} & = \sum_{k=g_++1}^{k_+-1} \log I_x(m_k,N_1-m_k).   
\end{split}
\end{align}
Note that the sizes of $x_k$ in $S_1^{(1)}$, $S_1^{(2)}$ and $S_1^{(3)}$ are roughly $[\del,x-\frac{M_1}{\sqrt{N_1}})$, $[x-\frac{M_1}{\sqrt{N_1}}, x+\frac{M_1}{\sqrt{N_1}}]$ and $(x+\frac{M_1}{\sqrt{N_1}},1-\del]$, respectively. 
On the other hand, since $k_{\pm}$ is defined using $\lfloor \delta N \rfloor$, not $\lfloor \delta N_1 \rfloor$, one cannot guarantee $\del \le x_k \le 1-\del$ on the boundary of $S_1$. Nevertheless, this distinction is not important, as the approximate size of $x_k$ determines the asymptotic behaviour.

We denote the $\eta$ in \eqref{def of eta in lemma} corresponding to $x_0 = x_k$ by $\eta_k$, i.e.
\begin{equation} \label{def of eta}
    -\frac12\eta_k^2 = x_k \log\Big(\frac{x}{x_k}\Big) + (1-x_k) \log \Big(\frac{1-x}{1-x_k}\Big), \qquad \frac{\eta_k}{x-x_k}>0.
\end{equation}
Note that $\eta_k>0$ for $x_k<x$, whereas $\eta_k<0$ for $x_k>x$. Then by Lemma \ref{beta}, we have 
\begin{equation} \label{summand of S1}
    I_x(m_k,N_1-m_k) = \frac12 \erfc\Big(-\eta_k \sqrt{\frac{N_1}{2}}\Big) + \frac1{\sqrt{2\pi N_1}}e^{-\frac12 \eta_k^2 N_1}\Big(c_0(\eta_k) + \frac{1}{N}c_1(\eta_k)+ \OO\Big(\frac1{N_1^2}\Big)\Big).
\end{equation}

We first compute the asymptotic behaviour of $S_1^{(1)}$.

\begin{lem} \label{S1^(1)}
    There exists a constant $C>0$ such that
    \begin{equation}
     S_1^{(1)} = \OO(e^{-CM_1^2}), \qquad \text{ as }N \to \infty.
    \end{equation}
\end{lem}
\begin{proof}%[Proof of Lemma~\ref{S1^(1)}]
   By the Taylor expansion of $\eta_k^2$ at $x_k = x$, we have
    \[
    -\frac12 \eta_k^2 = - \frac{(x_k-x)^2}{2x(1-x)} + \OO((x_k - x)^3), \qquad \text{ as } x_k \to x.
    \]
    This in turn implies that  
    $$
    -\frac12 \eta_k^2N_1 = -\frac{1}{2x(1-x)}M_1^2 + o(1), \qquad \text{ as } x_k \to x-\frac{M_1}{\sqrt{N_1}}.  
    $$
    Therefore, it follows that $-\frac12 \eta_k^2N_1 \to -\infty$ for all $k \in [k_-+1,g_--1]$. 
  
    Recall that by \cite[Eq.(7.12.1)]{NIST}, as $x \to \infty$,
    \begin{align} \label{erfc asymptotic}
        \erfc x \sim \frac{e^{-x^2}}{\sqrt \pi x} \sum_{m=0}^\infty (-1)^m \frac{(2m-1)!!}{(2x^2)^m}, \qquad \erfc(-x) \sim 2 - \frac{e^{-x^2}}{\sqrt \pi x} \sum_{m=0}^\infty (-1)^m \frac{(2m-1)!!}{(2x^2)^m}.
    \end{align}
   Using this, \eqref{summand of S1} can be rewritten as
    \[
    I_x(m_k,N_1-m_k) = 1 - \frac{1}{\sqrt{2\pi N_1}\eta_k}e^{-\frac12\eta_k^2 N_1}\OO(1) + \frac{1}{\sqrt{2\pi N_1}}e^{-\frac{1}{2}\eta_k^2 N_1}\OO(1) = 1+\OO(e^{-cM^2}), 
    \]
    which completes the proof.
\end{proof}

We denote 
\begin{equation} \label{def of big M_k}
M_k := \sqrt{N_1}(x_k - x), 
\end{equation}
where $x_k$ is given by \eqref{def of N1 M1 x_k}. 
We acknowledge that this is a slight abuse of notation, as we already use $M_1$ in \eqref{def of N1 M1 x_k}. However, since $M_k$ with $k=1$ will not be used further in the sequel, we allow this abuse of notation.

To analyse $S_1^{(2)}$, let us introduce the following lemma.

\begin{lem}\label{Riemann sum of M_k}
    Let $h \in C^2(\RR)$. Then as $N \to \infty$,
    \begin{equation}
        \sum_{k = g_-}^{g_+} h(M_k) = \sqrt{N_1}\int_{-M_1}^{M_1} h(t) \,dt + \Big(\frac12-\theta_-\Big) h(-M_1) + \Big(\frac12 -\theta_+\Big) h(M_1)+ \OO\Big(\frac{\mf m(h')}{\sqrt{N_1}}+\frac{M_1\mf m(h'')}{\sqrt{N_1}}\Big), 
    \end{equation}
    where
    \begin{equation}\label{theta_pm}
    \theta_- = g_- - (N_1x - M_1\sqrt{N_1}-\al - 1), \qquad \theta_+ = (N_1x + M_1\sqrt{N_1}-\al - 1) - g_+,
    \end{equation}
    and 
    \begin{equation}
    \mf m(f) := \max \{f(t): M_{g_-}\leq t \leq M_{g_+}\}. 
    \end{equation} 
\end{lem}
\begin{proof}
Note that by \eqref{def of big M_k}, we have $M_{k+1}-M_k = \frac{1}{\sqrt{N_1}}$. 
Therefore, it follows that 
    \begin{align*}
        \int_{M_{g_-}}^{M_{g_+}} h(t) \,dt &= \sum_{k = g_-}^{g_+-1} \int_{M_k}^{M_{k+1}}h(t) \,dt \\
        &= \sum_{k = g_-}^{g_+-1}\Big[h(M_k)(M_{k+1}-M_k) + h'(M_k)\frac{(M_{k+1}-M_k)^2}{2}\Big] + \OO\Big(\sum_{k=M_{g_-}}^{M_{g_+}} (M_{k+1}-M_k)^3 \mf m(h'')\Big) \\
        &= \frac1{\sqrt {N_1}}\sum_{k=g_-}^{g_+-1} h(M_k) + \frac1{2N_1} \sum_{k = g_-}^{g_+-1} h'(M_k) + \OO\Big(\frac{M_1}{N_1} \mf m(h'')\Big), 
    \end{align*} 
which leads to 
    \[
    \sum_{k=g_-}^{g_+-1} h(M_k) = \sqrt{N_1} \int_{M_{g_-}}^{M_{g_+}}h(t) \,dt - \frac1{2\sqrt{N_1}}\sum_{k=g_-}^{g_+-1}h'(M_k) + \OO\Big(\frac{M_1}{\sqrt{N_1}} \mf m(h'') \Big).
    \]
    By replacing $h(t)$ above by $h'(t)$, we obtain
    \[
    \sum_{k=g_-}^{g_+-1} h'(M_k) = \sqrt {N_1} \Big( h(M_{g_+})-h( M_{g_-}) \Big) + \OO\Big(M_1\mf m(h'')\Big).
    \]
   Combining these two equations, it follows that 
    \begin{equation*}%\label{eq3.12}
        \sum_{k=g_-}^{g_+} h(M_k) = \sqrt {N_1}\int_{M_{g_-}} ^{M_{g_+}} h(t) \,dt + \frac12\Big( h(M_{g_+})+ h(M_{g_-}) \Big) + \OO\Big(\frac{M_1}{\sqrt{N_1}} \mf m(h'')\Big).
    \end{equation*}
  Note here that by \eqref{g_pm} and \eqref{def of big M_k}, we have
    \[
    M_{g_+} = M_1 - \frac{\theta_+}{\sqrt{N_1}}, \qquad M_{g_-} = -M_1 + \frac{\theta_-}{\sqrt {N_1}}.
    \]
  By using this, we obtain 
    \[
    \int_{M_{g_-}}^{M_{g_+}} h(t)\, dt = \int_{-M_1}^{M_1} h(t)\, dt - h(-{M_1}) \frac{\theta_-}{\sqrt{N_1}} -h(M_1) \frac{\theta_+}{\sqrt N_1} + \OO\Big(\frac{\mf m(h')}{N_1}\Big)
    \]
 and 
    \[
    \frac12 (h(M_{g_+}) + h(M_{g_-})) = \frac12 (h(M_1) + h(-M_1)) + \OO\Big(\frac{\mf m(h')}{\sqrt{N_1}}\Big).
    \]
    Combining all of the above, the lemma follows. 
\end{proof}

We now compute the asymptotic behaviour of $S_1^{(2)}$. 
For this, it is convenient to introduce the notation
\begin{equation} \label{def of sigma}
    \sigma := \sqrt{x(1-x)}.
\end{equation}
Let us also define 
 \begin{equation} \label{def:g}
        g(t) = \frac{1-2x}{3\sig^3\sqrt{2\pi} \erfc(\frac{t}{\sqrt 2 \sig})}(t^2+2\sig^2)e^{-\frac{t^2}{2\sig^2}}.
    \end{equation}

\begin{lem} \label{S1^(2)}
    As $N \to \infty$, we have
    \begin{equation}
    S_1^{(2)} = B_4^{(2)}\sqrt{N_1} + B_6^{(2)} +  \OO\Big(\frac{M_1^5}{\sqrt {N_1}}+\frac{M_1^{10}}{N_1}\Big),
    \end{equation}
    where 
    \[
    B_4^{(2)} = \int_{-M_1}^{M_1}\log\Big[\frac 12 \erfc\Big(\frac{t}{\sqrt 2 \sig}\Big)\Big]\,dt, \qquad B_6^{(2)} = \int_{-M_1}^{M_1}g(t) \,dt + \Big(\frac12 - \theta_+\Big) \log\Big[\frac12 \erfc\Big(\frac{M_1}{\sqrt2\sig}\Big)\Big].
    \]
    Here $g$ is given by \eqref{def:g}. 
\end{lem}
\begin{proof}
    Note that $|x_k-x|<{M_1}/\sqrt{N_1} \to 0$ as $N_1 \to \infty$. By \eqref{def of eta}, as $N_1 \to \infty$, we have
    \begin{align}\label{eq:3.9.1}
    \begin{split}
        -\frac12 \eta_k^2& = -\frac{1}{2\sigma^2} \frac{M_k^2}{N_1}\Big(1-\frac{1-2x}{3\sigma^2}\frac{M_k}{\sqrt{N_1}} + \frac{1-3\sigma^2}{6\sigma^4}\frac{M_k^2}{N_1} + \OO\Big(\frac{M_k^3}{N_1\sqrt {N_1}}\Big)\Big), \\
         \eta_k & = -\frac{M_k}{\sigma\sqrt{N_1}}\Big(1-\frac{1-2x}{6\sigma^2}\frac{M_k}{\sqrt{N_1}} + \frac{5-14\sigma^2}{72\sigma^4}\frac{M_k^2}{N_1}+\OO\Big(\frac{M_k^3}{N_1\sqrt{N_1}}\Big)\Big), \\
        -\eta_k \sqrt{\frac{N_1}{2}} & = \frac{M_k}{\sqrt{2}\sigma}\Big(1-\frac{1-2x}{6\sigma^2}\frac{M_k}{\sqrt{N_1}} + \frac{5-14\sigma^2}{72\sigma^4}\frac{M_k^2}{N_1}+\OO\Big(\frac{M_k^3}{N_1\sqrt{N_1}}\Big)\Big).
    \end{split}
    \end{align}
    Define
    \[
    \Delta_k := -\eta_k \sqrt{\frac{N_1}{2}}- \frac{M_k}{\sqrt{2}\sigma} = -\frac{1-2x}{6\sqrt 2 \sigma^3}\frac{M_k^2}{\sqrt {N_1}} + \frac{5-14\sigma^2}{72\sqrt 2 \sigma^5} \frac{M_k^3}{N_1} + \OO\Big(\frac{M_k^4}{N_1\sqrt {N_1}} \Big).
    \]
    Note that by using the Taylor expansion  
    \[
    \erfc(a+x) = \erfc a - \frac{2}{\sqrt \pi}e^{-a^2}\Big(x - ax^2 + \frac{2a^2-1}{3}x^3 + \OO(a^3x^4)\Big), 
    \]
    as $x \to 0$, we obtain 
    \begin{align*}
&\quad        \erfc\Big(-\eta_k \sqrt{\frac{N}{2}}\Big) = \erfc\Big(\frac{M_k}{\sqrt 2 \sig}\Big) - \frac{2}{\sqrt \pi}e^{-\frac{M_k^2}{2\sig^2}}\Big[\Del_k - \frac{M_k}{\sqrt 2 \sig}\Del_k^2 + \OO((1+M_k^2) \Del_k^3)\Big] 
\\
 &=
        \erfc\Big(\frac{M_k}{\sqrt 2 \sigma}\Big) - \sqrt {\frac2\pi}e^{-\frac{M_k^2}{2\sigma^2}}\bigg[-\frac{1-2x}{6 \sigma^3}\frac{M_k^2}{\sqrt{N_1}}+\frac{5-14\sigma^2}{72\sigma^5}\frac{M_k^3}{N_1} - \frac{1-4\sigma^2}{72\sigma^7}\frac{M_k^5}{N_1}+\OO\Big(\frac{M_k^4+M_k^8}{N_1\sqrt{N_1}}\Big)\bigg].
    \end{align*}
    On the other hand, by \eqref{def of c0}, we have 
    \[
    c_0(\eta_k) = \frac{1}{\eta_k} - \frac{\sqrt{x_k(1-x_k)}}{x-x_k} = \frac{1-2x}{3\sigma} - \frac{1+\sigma^2}{12\sigma^3} \frac{M_k}{\sqrt{N_1}} + \OO\Big(\frac{M_k^2}{N_1}\Big). 
    \]
    Note also that 
    \[
    e^{-\frac12 \eta_k^2 N_1} = e^{-\frac{M_k^2}{2\sig^2}}\Big[1+\frac{1-2x}{6\sig^4}\frac{M_k^3}{\sqrt{N_1}}+\OO\Big(\frac{M_k^4+M_k^6}{N_1}\Big)\Big].
    \]
    Then by Lemma~\ref{beta}, it follows that 
    \begin{align*}
        I_x(m_k, N_1-m_k) &= \frac12 \erfc\Big(-\eta_k \sqrt{\frac{N_1}{2}}\Big) + \frac{1}{\sqrt{2\pi {N_1}}}e^{-\frac12 \eta_k^2{N_1}}\Big(c_0(\eta_k) + \OO\Big(\frac1{N_1}\Big)\Big) \\
        &= u_1(M_k) + \frac{1}{\sqrt{N_1}} u_2(M_k) + \frac{1}{N_1} u_3(M_k) + \OO\Big(\frac{1+M_k^8}{N_1\sqrt{N_1}}e^{-\frac{M_k^2}{2\sig^2}}\Big),
    \end{align*}
    where
    \begin{align*}
        u_1(t) &= \frac12 \erfc\Big(\frac{t}{\sqrt2 \sig}\Big), \qquad u_2(t) = \frac{1-2x}{\sqrt{2\pi}}e^{-\frac{t^2}{2\sig^2}} \Big[\frac{1}{6\sig^3}t^2 + \frac{1}{3\sig}\Big], \\
        u_3(t) &= \frac{1}{\sqrt{2\pi}}e^{-\frac{t^2}{2\sig^2}}\Big[\frac{1-4\sig^2}{72\sig^7}t^5 - \frac{1+2\sig^2}{72\sig^5}t^3 -\frac{1+\sig^2}{12\sig^3}t\Big].
    \end{align*}
    Therefore we obtain
    \begin{align*}
        \log I_x(m_k, N_1-m_k) &= \log u_1(M_k) + \frac{1}{\sqrt{N_1}}  g(M_k) + \frac{1}{N_1} g_1(M_k) 
        \\
        &\quad + \OO\bigg(\frac{1}{N_1\sqrt N_1}\Big(\frac{u_2(M_k)^3}{u_1(M_k)^3}+\frac{u_2(M_k)u_3(M_k)}{u_1(M_k)^2}+\frac{1+M_k^8}{u_1(M_k)} e^{-\frac{M_k^2}{2\sig^2}}\Big)\bigg),
    \end{align*}
    where 
    \[
    g(t) = \frac{u_2(t)}{u_1(t)}, \qquad g_1(t) = -\frac{u_2(t)^2}{2u_1(t)^2} + \frac{u_3(t)}{u_1(t)}.
    \]

    Recall that $S_1^{(2)}$ is given by \eqref{def of S1(123)}. 
    Notice here that as $t \to \infty,$ we have
    $$
   u_1(t) = \OO(t^{-1}e^{-\frac12 xt^2}), \qquad u_2(t) = \OO(t^2 e^{-\frac12 xt^2}), \qquad u_3(t) = \OO(t^5 e^{-\frac12 xt^2}).
    $$
    Then by applying Lemma~\ref{Riemann sum of M_k} with $h(t) = t^9$, we have 
    \[
    \sum_{k = g_-}^{g_+} \OO\bigg(\frac{1}{N\sqrt N}\Big(\frac{u_2(M_k)^3}{u_1(M_k)^3}+\frac{u_2(M_k)u_3(M_k)}{u_1(M_k)^2}+\frac{1+M_k^8}{u_1(M_k)} e^{-\frac{M_k^2}{2\sig^2}}\Big)\bigg) = \OO\Big(\frac{1+M_1^{10}}{N_1}\Big). 
    \] 
    Furthermore, by using Lemma~\ref{Riemann sum of M_k} again with $h(t)= \log u_1(t)$, $h(t)=g(t)$ and $h(t)= g_1(t)$, the desired lemma follows. Here, we have used  $ g_1(t) = \OO(t^4)$ as $t \to \infty$.
\end{proof}

Finally, we compute $S_1^{(3)}$ in \eqref{def of S1(123)}, where $g_+$ and $k_+$ are given as \eqref{g_pm} and \eqref{k_pm}.
For the analysis of $S_1^{(3)}$, we shall apply the following lemma. 

\begin{lem} \label{lem:3.9}
    Let $h \in C^3((x,1-\del))$. Then as $N \to \infty$,
    \begin{align}
    \begin{split}
        \sum_{k=g_++1}^{k_+-1} h(x_k) &= N_1 \int_{x+\frac{M_1}{\sqrt{N_1}}}^{1-\del} h(t) \,dt + \Big(\theta_+ - \frac12 \Big)h\Big(x+\frac{M_1}{\sqrt{N_1}}\Big) + \Big(\theta_\del'-\frac12\Big) h(1-\del) 
        \\
        &\quad -\frac{1-6\theta_++6\theta_+^2}{12N_1}h'\Big(x+\frac{M_1}{\sqrt{N_1}}\Big) + \frac{1-6\theta_\del'+6\theta_\del'^2}{12N_1}h'(1-\del)+\OO\Big(\frac{\mf m(h''')+\mf m(h'')}{N_1^2}\Big),   
    \end{split}
    \end{align}
    where $\theta_\del' := \theta_{\del}-(1-\del)$ and  
    \begin{equation} \label{mf lem3.9}
    \mf m(f) := \max \Big\{f(t):x+\frac{M_1}{\sqrt{N_1}}\leq t \leq 1-\del \Big \}. 
    \end{equation} 
    Here, $\theta_+$ and $\theta_\del$ are given by \eqref{theta_pm} and \eqref{theta_Mdel}. 
\end{lem}
\begin{proof}
   Let $A = x+\frac{M_1}{\sqrt{N_1}}$, $a_0 = 1-\theta_+$, $B = 1-\del$ and $b_0 = \theta_\del'-1$. Then we have
    \begin{align*}
        a_n &= AN_1 + a_0 = \Big(x+\frac{M_1}{\sqrt{N_1}}\Big)N_1 + 1-\theta_+ = g_++\al+2, \\
        b_n &= BN_1+b_0 = (1-\del)N_1+\theta_\del'-1 =  (1-\del)N + \theta_\del - 1 = k_++\al.
    \end{align*}
    Then the lemma follows from Lemma \ref{lem:3.3}.
\end{proof}

Define
\begin{equation}\label{def:T}
    T(t) = 2\bigg(t\log \Big( \frac tx \Big)+ (1-t) \log \Big( \frac{1-t}{1-x} \Big) \bigg).
\end{equation}
Note that $T(x_k) = \eta_k^2$, where $\eta_k$ is given by \eqref{def of eta}. 

\begin{lem} \label{S1^(3)}
    As $N \to \infty$, we have 
    \begin{align}
    \begin{split}
        S_1^{(3)} &= B_1^{(3)} N_1^2 + B_2^{(3)}N_1\log N_1 + B_3^{(3)} N_1 + B_4^{(3)} \sqrt{N_1} + B_5^{(3)} \log N_1 + B_6^{(3)} \\
        &\quad + B_7^{(3)} \frac 1{N_1} + B_8^{(3)} \frac 1{N_1^2} + \OO\Big(\frac{M_1}{\sqrt{N_1}} + \frac1{M_1^2}+\frac{\sqrt {N_1}}{M_1^7}\Big),
    \end{split}
    \end{align}
    where
    \begin{align*}
        B_1^{(3)} &= -\frac12 \int_{x+\frac{M_1}{\sqrt{N_1}}}^{1-\del} T(t)\, dt, \qquad B_2^{(3)} = -\frac12(1-\del-x), 
        \\
        B_3^{(3)} &= \frac12 \int_{x+\frac{M_1}{\sqrt{N_1}}}^{1-\del} h_1(t)\, dt -\frac12(1-\del-x)\log(2\pi) + \frac12\Big[\Big(\frac12 - \theta_+\Big)T_{M_1}+\Big(\frac12 - \theta_\del'\Big)T_{1-\del}\Big], \\
        B_4^{(3)} &= \frac12 M_1\log(2\pi N_1), \qquad B_5^{(3)} = \frac12(1-\theta_+-\theta_\del'), \\
        B_6^{(3)} &= \frac12 (1-\theta_+-\theta_\del') \log(2\pi) + \frac{1-6\theta_++6\theta_+^2}{24}T_{M_1}' - \frac{1-6\theta_\del'+6\theta_\del'^2}{24}T_{1-\del}' \\
        &\quad + \frac12\Big(\theta_+-\frac12\Big)h_1\Big(x+\frac{M_1}{\sqrt{N_1}}\Big) + \frac12\Big(\theta_\del'-\frac12\Big)h_1(1-\del) - \int_{x+\frac{M_1}{\sqrt{N_1}}}^{1-\del}h_2(t)\, dt, \\
        B_7^{(3)} &= \frac52 \int_{x+\frac{M_1}{\sqrt{N_1}}}^{1-\del} \frac{1}{T(t)^2}\,dt, \qquad B_8^{(3)} = -\frac{37}{3}\int_{x+\frac{M_1}{\sqrt{N_1}}}^{1-\del}\frac1{T(t)^3}\,dt.
    \end{align*}
    Here $T$ is given by \eqref{def:T} and 
    \begin{equation} \label{def of TM1 del}
    T_{M_1} = T(x+\tfrac{M_1}{\sqrt{N_1}}), \qquad T_{1-\del} = T(1-\del), \qquad T_{M_1}' = T'(x+\tfrac{M_1}{\sqrt{N_1}}),  \qquad T_{1-\del}' = T'(1-\del)
    \end{equation} 
    and
    \begin{equation}\label{def:h1h2}
        h_1(t)  = \log \Big( \frac{t(1-t)}{(t-x)^2} \Big), \qquad
        h_2(t) = \frac{1-13t+13t^2}{12t(1-t)} - \frac{1-2t}{t-x}+\frac{t(1-t)}{(t-x)^2}.
    \end{equation}
\end{lem}

\begin{proof}
By Lemma \ref{beta}, we have
\begin{equation}\label{eq:49}
    I_x(m_k,N_1-m_k) = \frac12 \erfc\Big(-\eta_k\sqrt{\frac{N_1}{2}}\Big) + \frac{1}{\sqrt{2\pi N_1}}e^{-\frac12 \eta_k^2 N_1}\Big(c_0(\eta_k) + \OO\Big(\frac1{N_1}\Big)\Big).
\end{equation}
Let us write $x_k = x+d$ with $d\ll 1$. 
Note that 
\[
-\eta_k \sqrt{\frac{N_1}{2}} = \sqrt{\frac{N_1}{2}}\frac d\sig \Big(1-\frac{1-2x}{6\sig^2}d + \OO(d^2)\Big)  \to \infty.  
\]
Then by \eqref{erfc asymptotic}, we obtain 
\begin{equation} \label{S1^(3):Ix}
   I_x(m_k,N_1 - m_k) = -\frac1{\eta_k} \,\frac{1}{\sqrt{2\pi N_1}}e^{-\frac12 \eta_k^2 N_1} \Big[\al_k -\beta_k \frac{1}{N_1}+\frac{3}{\eta_k^4 N_1^2}-\frac{15}{\eta_k^6 N_1^3}+\OO\Big(\frac{1}{\eta_k^8N_1^4}+\frac{\eta_k}{N_1^2}\Big)\Big] , 
\end{equation}
where
\begin{equation} \label{def:alphabeta}
    \al_k := 1- \eta_k c_0(\eta_k) = -\eta_k \frac{\sqrt{x_k(1-x_k)}}{x_k-x}, \qquad \beta_k := \frac{1}{\eta_k^2}+\eta_k c_1(\eta_k).
\end{equation}
Here $c_0$ and $c_1$ are given by \eqref{def of c0} and \eqref{def of c1}.

Note that as $x_k \to x$,
\[
\al_k = 1+\OO(x_k - x), \qquad \beta_k = \frac{\sig^2}{(x_k-x)^2}+\frac{1-2x}{3}\frac{1}{x_k - x}-\frac{1-\sig^2}{18\sig^2}+\OO(x_k-x).
\]
Notice also that $\al_k$ is bounded away from 0 in $x_k \in [x,1-\del]$.
%By taking logarithm on each side of \eqref{S1^(3):Ix}, we obtain
Therefore, we have 
\begin{align} \label{S1^(3):logIx}
\begin{split}
    \log I_x(m_k, N_1-m_k) &= -\frac12 \eta_k^2 N_1 - \frac12 \log N_1 - \frac12 \log \eta_k^2 - \frac12 \log(2\pi) +\log \al_k - \frac{\beta_k}{\al_k N_1}\\
    &\quad +\frac{5}{2\eta_k^4 N_1^2}-\frac{37}{3\eta_k^6 N_1^3}+\OO\Big(\frac{1}{\eta_k^8N_1^4}+\frac{1}{\eta_k^2 N_1^2(x_k-x)}+\frac{x_k-x}{\eta_k^4 N_1^2}+\frac{x_k-x}{\eta_k^6N_1^3}\Big). 
\end{split}
\end{align}
    
    We compute the summation of each term of \eqref{S1^(3):logIx} using Lemma \ref{lem:3.9}. 
    For this purpose, let us write 
    \begin{align*}
        \mc S_0 &= -\frac12 \sum_{k=g_++1}^{k_+-1} \Big( \log N_1 + \log(2\pi)\Big), \qquad \mc S_1 = -\frac{N_1}{2} \sum_{k=g_++1}^{k_+-1} \eta_k^2, \qquad \mc S_2 = \sum_{k=g_++1}^{k_+-1} \Big(-\frac12 \log \eta_k^2 + \log \al_k \Big), 
        \\
        \mc S_3 &= -\frac1{N_1} \sum_{k=g_++1}^{k_+-1} \frac{\beta_k}{\al_k}, \qquad  \mc S_4 = \frac{5}{2N_1^2} \sum_{k=g_++1}^{k_+-1} \frac{1}{\eta_k^4}, \qquad \mc S_5 = -\frac{37}{3N_1^3} \sum_{k=g_++1}^{k_+-1} \frac{1}{\eta_k^6}.
    \end{align*}
    First, observe that since 
    \[
    \sum_{k=g_++1}^{k_+-1} 1 = k_+-g_+-1 = (1-\del-x)N_1 - M_1\sqrt{N_1}+ \theta_++ \theta_\del' -1,
    \]
    we have 
    \begin{align}\label{S1^(3):part0}
    \begin{split}
        \mc S_0 &= -\frac12 (1-\del-x) N_1 \log N_1 - \frac12 (1-\del-x) \log(2\pi) N_1 + \frac12 M_1\sqrt{N_1}\log N_1 \\
        & \quad + \frac12 \log(2\pi) M_1 \sqrt{N_1} + \frac12 (1-\theta_+-\theta_\del') \log N_1 + \frac12 (1-\theta_+-\theta_\del') \log(2\pi).
    \end{split} 
    \end{align}
    We take account of the following variation of Lemma \ref{lem:3.9} which can be obtained by putting $h \circ T$ instead of $h$: 
    \begin{align*}
        \sum_{k=g_++1}^{k_+-1} h(\eta_k^2) &= N_1 \int_{x+\frac{M_1}{\sqrt{N_1}}}^{1-\del}h( T(t)) \,dt +\Big(\theta_+-\frac12\Big)h(T_{M_1}) + \Big(\theta_\del'-\frac12\Big)h(T_{1-\del}) \\
    &\quad - \frac{1-6\theta_++6\theta_+^2}{12N_1}h'(T_{M_1})T_{M_1}' + \frac{1-6\theta_\del'+6\theta_\del'^2}{12N_1}h'(T_{1-\del}) T_{1-\del}' + \OO\Big(\frac{\mf m((h\circ T)'') + \mf m((h\circ T)''')}{N_1^2}\Big).
    \end{align*}
    By using this formula, we obtain 
    \begin{align} \label{S1^(3):part145}
    \begin{split}
        \mc S_1 
        &= -\frac{N_1^2}2 \int_{x+\frac{M_1}{\sqrt{N_1}}}^{1-\del} T(t) \,dt + \frac{N_1}{2}\Big[\Big(\frac12-\theta_+\Big)T_{M_1}+\Big(\frac12 -\theta_\del'\Big)T_{1-\del}\Big] \\
        &\quad + \frac{1-6\theta_++6\theta_+^2}{24}T_{M_1}' - \frac{1-6\theta_\del'+6\theta_\del'^2}{24}T_{1-\del}' + \OO\Big(\frac{1}{N_1}\Big), \\
        \mc S_4
        &= \frac{5}{2N_1}\int_{x+\frac{M_1}{\sqrt{N_1}}}^{1-\del} \frac{1}{T(t)^2}\, dt + \OO\Big(\frac{1}{M_1^4}\Big), \qquad \mc S_5 = -\frac{37}{3N_1^2} \int_{x+\frac{M_1}{\sqrt{N_1}}}^{1-\del}\frac{1}{T(t)^3}\,dt + \OO\Big(\frac{1}{M_1^6}\Big).
    \end{split}
    \end{align} 
    In order to analyse $\mc S_2 $, note that
    \[
    \frac{\al_k}{\eta_k} = \frac{1-\eta_kc_0(\eta_k)}{\eta_k} = -\frac{\sqrt{x_k(1-x_k)}}{x_k-x}.
    \]
    Then by using Lemma \ref{lem:3.9} with $h_1$ in \eqref{def:h1h2}, we have 
    \begin{equation} \label{S1^(3):part2}
    \mc S_2 = \frac{N_1}2\int_{x+\frac{M_1}{\sqrt{N_1}}}^{1-\del} h_1(t) \,dt + \frac12\Big(\theta_+-\frac12 \Big) h_1\Big(x+\frac{M_1}{\sqrt{N_1}}\Big) + \frac12 \Big(\theta_\del'-\frac12 \Big) h_1(1-\del) + \OO\Big(\frac{M_1}{\sqrt{N_1}}\Big).
    \end{equation}
    Finally, note that by \eqref{def of c1}, we have 
    \[
    \frac{\beta_k}{\al_k} = \frac{x_k(1-x_k)}{(x_k-x)^2}-\frac{1-2x_k}{x_k-x}+\frac{1-13x_k+13x_k^2}{12x_k(1-x_k)} = h_2(x_k). 
    \]
    Then it follows from Lemma \ref{lem:3.9} that 
    \begin{equation} \label{S1^(3):part3}
        \mc S_3 = -\int_{x+\frac{M_1}{\sqrt{N_1}}}^{1-\del} h_2(t) \,dt + \OO\Big(\frac{1}{M_1^2}\Big). 
    \end{equation} 
    Note also that Lemma \ref{lem:3.9} shows 
    \[
     \sum_{k=g_++1}^{k_+-1}  \OO\Big(\frac{1}{\eta_k^8N_1^4}+\frac{1}{\eta_k^2N_1^2(x_k-x)}+\frac{x_k-x}{\eta_k^4N_1^2}+\frac{x_k-x}{\eta_k^6N_1^3}\Big) = \OO\Big(\frac{\sqrt{N_1}}{M_1^7}+\frac{1}{M_1^2}\Big). 
    \]
    Combining \eqref{S1^(3):part0}, \eqref{S1^(3):part145}, \eqref{S1^(3):part2} and \eqref{S1^(3):part3}, the lemma follows. 
\end{proof}

To state the asymptotic behaviour of $S_1$, we define 
\begin{equation} \label{J1J2}
    \JJ_1^{(\del)} = \int_x^{1-\del}\Big[h_1(t) - 2\log \Big(\frac{\sig}{t-x} \Big) \Big]\,dt, \qquad \JJ_2^{(\del)} = \int_x^{1-\del}\Big[h_2(t) - \frac{\sig^2}{(t-x)^2}\Big]\,dt,
\end{equation}
where $h_1$ and $h_2$ are defined by \eqref{def:h1h2}. 

\begin{lem}[\textbf{Asymptotic behaviour of $S_1$}] \label{S1}
    As $N \to \infty$, we have 
    \begin{equation}
     S_1 = B_1 N_1^2 + B_2 N_1 \log N_1 + B_3 N_1 + B_4 \sqrt {N_1} + B_5 \log N_1 + B_6 + \OO\bigg(\frac{M_1^5}{\sqrt{N_1}}+\frac1{M_1^2}+\frac{\sqrt{N_1}}{M_1^7}\bigg),
    \end{equation}
where 
    \begin{align*}
        B_1 &= -\frac12 \int_x^{1-\del} T(t) \,dt, \qquad B_2 = -\frac12(1-x-\del), \\
        B_3 &= \frac12 \JJ_1^{(\del)}-\frac12 \Big( \theta_\del'-\frac12\Big)T_{1-\del}+(1-x-\del)\Big[1+\log \Big( \frac{\sig}{\sqrt{2\pi}(1-x-\del)} \Big) \Big], \\
        B_4 &=  \sig \, \II_1 , \qquad B_5 = \frac14-\frac12\theta_\del', 
        \\
        B_6 &= \frac{2(1-2x)}{3}\II_2-\JJ_2^{(\del)} +\Big(\frac12 -  \theta_\del'\Big)\log \sqrt{2\pi} - \frac{1-6\theta_\del'+6\theta_\del'^2}{24}T_{1-\del}'+\frac12\Big(\theta_\del'-\frac12\Big)h_1(1-\del) + \frac{\sig^2}{1-x-\del}.
    \end{align*}
    Here, $T$ and $h_1$ are given by \eqref{def:T} and \eqref{def:h1h2}; $T_{1-\del}$, $T_{1-\del}'$ are given by \eqref{def of TM1 del}; and $\II_1, \II_2, \JJ_1^{(\del)}$ and $\JJ_2^{(\del)}$ are given by \eqref{I1}, \eqref{I2} and \eqref{J1J2}.
\end{lem}

\begin{proof}
    By combining Lemmas~\ref{S1^(1)}, \ref{S1^(2)} and \ref{S1^(3)}, we obtain
    \begin{align}
    \begin{split}  \label{S1 sum 123}
        S_1 &= B_1^{(3)}N_1^2 + B_2^{(3)}N_1\log N_1 + B_3^{(3)} N_1 + B'_4\sqrt{N_1}+B_5^{(3)}\log N_1 + B'_6
        \\
    &\quad + B_7^{(3)} \frac{1}{N_1}+B_8^{(3)} \frac{1}{N_1^2} + \OO\bigg(\frac{\sqrt{N_1}}{M_1^7}+\frac{1}{M_1^2}+\frac{M_1^5}{\sqrt{N_1}}+\frac{M_1^{10}}{N_1}\bigg),
    \end{split}
    \end{align}
    where
    \begin{equation}
      B'_4 = B_4^{(2)} +B_4^{(3)}, \qquad 
   B'_6 = B_6^{(2)}+B_6^{(3)}.
    \end{equation}
    We claim that every $M_1$-dependent terms cancel out. 
    
    Let us begin with the leading order term. Note that as $t \to x$, we have
    \[
    T(t) = \frac1{\sig^2}(t-x)^2 - \frac{1-2x}{3\sig^4}(t-x)^3+\frac{1-3\sig^2}{6\sig^6}(t-x)^4+\OO((t-x)^5).
    \]
    Then we obtain
    \begin{align} 
    \begin{split}
        B_1^{(3)}N_1^2 &= -\frac12 N_1^2 \int_{x+\frac{M_1}{\sqrt{N_1}}}^{1-\del} T(t)\, dt = -\frac12 N_1^2\int_x^{1-\del}T(t) \,dt + \frac12 N_1^2 \int_x^{x+\frac{M_1}{\sqrt{N_1}}}T(t) \,dt 
        \\
        &= -\frac12 N_1^2 \int_x^{1-\del}T(t)\, dt + \frac1{6\sig^2}M_1^3\sqrt{N_1} - \frac{1-2x}{24\sig^4}M_1^4 + \OO\bigg(\frac{M_1^5}{\sqrt{N_1}}\bigg). \label{S1:part1}
    \end{split}
    \end{align}
    Note that by Lemma~\ref{S1^(3)}, 
    \begin{align}
    B_2^{(3)}N_1\log N_1 &=  -\frac12(1-\del-x) N_1 \log N_1, 
    \\ 
    B_5^{(3)} \log N_1 &= \frac12(1-\theta_+-\theta_\del') \log N_1. 
    \end{align}
    
   For the term $B_3^{(3)}N_1$, note that
    \[
    h_1(t) = 2\log \Big(\frac{\sig}{t-x}\Big)+\frac{1-2x}{\sig^2}(t-x) + \OO((t-x)^2), \qquad \text{as } t \to x.
    \]
   This gives rise to 
    \begin{align*}
      &\quad   \frac12\int_{x+\frac{M_1}{\sqrt{N_1}}}^{1-\del} h_1(t) \,dt = \frac12 \int_{x+\frac{M_1}{\sqrt{N_1}}}^{1-\del} \Big(h_1(t) - 2\log \Big(\frac{\sig}{t-x}\Big)\Big) \,dt + \int_{x+\frac{M_1}{\sqrt{N_1}}}^{1-\del}\log \Big(\frac{\sig}{t-x}\Big)\,dt 
      \\
        &= \frac{ \JJ_1^{(\del)} }{2}- \frac12 \int_x^{x+\frac{M_1}{\sqrt{N_1}}}\Big[\frac{1-2x}{\sig^2}(t-x) + \OO((t-x)^2)\Big]\,dt + \int_{x+\frac{M_1}{\sqrt{N_1}}}^{1-\del} \log \Big(\frac{\sig}{t-x}\Big)\,dt \\
        &= \frac{ \JJ_1^{(\del)} }{2} +(1-\del-x)\Big[1+\log\Big(\frac{\sig}{1-\del-x}\Big)\Big] + \frac{M_1}{\sqrt{N_1}}\log\Big(\frac{M_1}{\sqrt{N_1}}\Big) -\frac{M_1(1+\log \sig) }{\sqrt{N_1}}- \frac{1-2x}{4\sig^2}\frac{M_1^2}{N_1} + \OO\Big(\frac{M_1^3}{N_1\sqrt{N_1}}\Big).
    \end{align*}
    Using this and $T_{M_1} = \frac{M_1^2}{\sig^2N_1}+\OO(M_1^3N_1^{-3/2})$, we obtain
    \begin{align} \label{S1:part3}
        \begin{split}
            B_3^{(3)}N_1 &= N_1\bigg[\frac12 \JJ_1^{(\del)}-\frac12\Big(\theta_\del'-\frac12\Big)T_{1-\del}+(1-\del-x)\Big[1+\log\Big(\frac{\sig}{\sqrt{2\pi}(1-\del-x)}\Big)\Big]\bigg] 
            \\
            &\quad + \sqrt{N_1} \bigg[M_1\log M_1 - \frac12 M_1\log N_1 - M_1(1+\log \sig) \bigg] - \frac{\theta_+-x}{2\sig^2}M_1^2 + \OO\Big(\frac{M_1^3}{\sqrt{N_1}}\Big).
        \end{split}
    \end{align}
    
    Recall that by Lemmas~\ref{S1^(2)} and \ref{S1^(3)}, 
    \[
    B'_4 = -\int_{-M_1}^{M_1}\log\Big[\frac12\erfc\Big(\frac{t}{\sqrt2\sig}\Big)\Big]\,dt + \frac12 M_1 \log(2\pi N_1).
    \]
    Here, we have 
    $$
      \log\Big[\frac12\erfc\Big(\frac{t}{\sqrt 2 \sig}\Big)\Big] = -\frac{t^2}{2\sig^2}-\frac12 \log\Big(1+\frac{t^2}{2\sig^2}\Big) - \frac12\log(4\pi) + \frac{3\sig^4}{2(1+t^4)}-\frac{11\sig^6}{1+t^6}+\OO\Big(\frac1{t^8}\Big) \qquad \text{ as } t \to +\infty, 
    $$
    and 
    $$
      \log\Big[\frac12\erfc\Big(\frac{t}{\sqrt 2 \sig}\Big)\Big] = \OO\Big(\frac1{|t|}e^{-\frac{t^2}{2\sig^2}}\Big), \qquad \text{ as } t \to -\infty. 
    $$
    Define 
    \[
    w(t) = \log\Big[\frac12 \erfc\Big(\frac{t}{\sqrt 2 \sig}\Big)\Big]+\Big[\frac{t^2}{2\sig^2}+\frac12 \log\Big(1+\frac{t^2}{2\sig^2}\Big)+\frac12\log(4\pi)\Big] \ind_{(0,\infty)}(t). 
    \] 
    We also write 
    \begin{equation}
     \widetilde{\II}_1 = \int_{-\infty}^\infty \bigg[\log\Big(\tfrac12 \erfc(t) \Big) + \Big(t^2+\frac12\log(1+t^2)+\frac12 \log(4\pi)\Big)\ind_{(0,\infty)}(t)\bigg]\, dt.
    \end{equation}
    Note that $\widetilde{\II}_1$ is related to \eqref{I1} as 
    \begin{equation}
    \II_1= \sqrt{2} \Big(  \widetilde{\II}_1- \frac{\pi}{2} \Big). 
    \end{equation}
    Then we have
    \begin{align*}
    &\quad     \int_{-M_1}^{M_1} \log\Big[\frac12 \erfc\Big(\frac{t}{\sqrt2 \sig}\Big)\Big]\,dt = \int_{-M_1}^{M_1} w(t) \,dt - \int_0^{M_1}\Big(\frac{t^2}{2\sig^2}+\frac12\log\Big(1+\frac{t^2}{2\sig^2}\Big)+\frac12\log(4\pi)\Big)\,dt 
    \\
     &= \sqrt 2 \sig  \widetilde{\II}_1 - \int_{M_1}^\infty \Big(\frac{3\sig^4}{2(1+t^4)}-\frac{11\sig^6}{1+t^6}\Big)\,dt - \int_0^{M_1}\Big(\frac{t^2}{2\sig^2}+\frac12\log\Big(1+\frac{t^2}{2\sig^2}\Big)+\log\sqrt{4\pi}\Big)\,dt +\OO\Big(\frac{1}{M_1^7}+\frac{e^{-M_1^2}}{M_1}\Big).
    \end{align*}
    Then by direct computations, it follows that
    \begin{align}\label{S1:part4}
        \begin{split}
            B'_4\sqrt{N_1} &= -\frac{M_1^3\sqrt{N_1}}{6\sig^2}+\frac12 M_1\sqrt{N_1}\log N_1 - M_1\log M_1\sqrt{N_1}+(1+\log \sig)M_1\sqrt{N_1} \\
            &\quad +   \sig \, \II_1  \sqrt{N_1}+\sig^2 \frac{\sqrt{N_1}}{M_1}-\frac{5\sig^4}{6}\frac{\sqrt{N_1}}{M_1^3}+\frac{37\sig^6}{15}\frac{\sqrt{N_1}}{M_1^5}+\OO\Big(\frac{\sqrt{N_1}}{M_1^7}\Big).
        \end{split}
    \end{align}

  By employing similar computations as above, one can obtain 
    \begin{align}
        B_6^{(2)}&= \frac{1-2x}{24\sig^4}M_1^4+\frac{\theta_+-x}{2\sig^2}M_1^2 + \Big(\theta_+-\frac12\Big)\log \Big(\frac{\sqrt{2\pi}M_1}{\sig}\Big) + \frac{2(1-2x)}{3}\II_2 +\OO\Big(\frac1{M_1^2}\Big), \label{S1:part6-2} \\
        \begin{split}\label{S1:part6-3}
            B_6^{(3)}&= -\frac{\sig^2}{M_1}\sqrt{N_1} + \frac12\Big(\theta_+-\frac12\Big)\log N_1 - \Big(\theta_+-\frac12\Big)\log \Big( \frac{\sqrt{2\pi}M_1}{\sig} \Big)  + \Big(\frac12 - \theta_\del'\Big)\log\sqrt{2\pi} 
            \\
        &\quad - \frac{1-6\theta_\del'+6\theta_\del'^2}{24}T_{1-\del}' + \frac12\Big(\theta_\del'-\frac12\Big)h_1(1-\del) + \frac{\sig^2}{1-\del-x}-\JJ_2^{(\del)} + \OO\Big(\frac{M_1}{\sqrt{N_1}}\Big),
        \end{split}
         \\
        \frac{1}{N_1}B_7^{(3)} &= \frac56\sig^4\frac{\sqrt{N_1}}{M_1^3}+\OO\Big(\frac1{M_1^2}+\frac1{N_1}+\frac{M_1}{N_1\sqrt{N_1}}\Big),
        \label{S1:part7}\\
        \frac{1}{N_1^2}B_8^{(3)} &= -\frac{37}{15}\sig^6\frac{\sqrt{N_1}}{M_1^5}+\OO\Big(\frac{1}{M_1^4}+\frac{1}{N_1^2}+\frac{M_1}{N_1^2\sqrt{N_1}}\Big).\label{S1:part8}
    \end{align}
    Here, we have used the asymptotic behaviour
    $$
    g(t) = \frac{1-2x}{6\sig^4}\Big[t^3+3\sig^2t+\OO\Big(\frac1{t^3}\Big)\Big]
    $$
    as $t \to +\infty$, as well as the behaviours 
    $$
    h_2(t) = \frac{\sig^2}{(t-x)^2}+\OO(1), \qquad \frac{1}{T(t)^2} = \frac{\sig^4}{(t-x)^4}+\OO((t-x)^{-3}), \qquad \frac{1}{T(t)^3} = \frac{\sig^6}{(t-x)^6}+\OO((t-x)^{-5}) 
    $$
    as $ t \to x$. 
    Combining all of the above, after straightforward simplifications, one can observe that all the terms depending on $M_1$ cancel out, leading to the desired asymptotic behaviour.
\end{proof}

\subsection{Proof of Theorem~\ref{thm:1.1}}

This subsection culminates the proof of Theorem~\ref{thm:1.1}. 
The remaining task is to combine Lemmas~\ref{S2+S3} and ~\ref{S1}, and then collect all the coefficients.

\begin{proof}[Proof of Theorem \ref{thm:1.1}] 
    Recall that $N_1 = N+1$. Thus we have
    \begin{align*}
        S_1 &= B_1 N^2 + B_2 N \log N + (2B_1 + B_3) N + B_4 \sqrt{N}\\
        &\quad + (B_2 + B_5) \log N + B_1+B_2+B_3+B_6 + \OO\Big(\frac1{\sqrt N}\Big).
    \end{align*}
    By combining \eqref{sum division S0123} with Lemmas~\ref{S2+S3} and \ref{S1}, we have
    \begin{align*}
        \log \mc P_n^\CC(R;\al,c) 
        &= A_1' N^2 + A_2' N \log N + A_3' N + A_4' \sqrt{N}+A_5' \log N + A_6' \\
        &\quad + \OO\Big(\frac1M + \frac{M^2}{N}+\frac{M^4}{N^2}+\frac{\sqrt{N_1}}{M_1^7}+\frac{1}{M_1^2}+\frac{M_1^5}{\sqrt{N_1}}+\frac{M_1^{10}}{N_1}\Big)
    \end{align*}
    as $N \to \infty$, where
    \begin{align*}
        A_1' &= B_1 + E_1, \qquad A_2' = B_2 + E_2, \qquad A_3' = 2B_1 + B_3 + E_3, \\
        A_4' &= B_4, \qquad A_5' = B_2+B_5 + E_5, \qquad A_6' = B_1+B_2+B_3+B_6+E_6.
    \end{align*}
    
    Note that since 
    \[
    B_1 = -\frac12 \int_x^{1-\del}T(t)\, dt = \frac12 \del^2 \log \Big(\frac{\del}{1-x}\Big)-\frac12 (1-\del)^2 \log \Big(\frac{1-\del}{x}\Big)+\frac12(1-x-\del),
    \]
    we have 
    \begin{equation} \label{A1A2}
        A_1' = \frac12 (1-x+\log x), \qquad A_2' = -\frac12(1-x). 
    \end{equation}
    For $A_3'$, recall the definitions of $f_1$, $T_{1-\del}$ from \eqref{def:f1f2}, Lemma \ref{S1^(3)} and $\theta_\del':= \theta_\del-(1-\del)$. 
    Note also that $\JJ_1^{(\del)}$ in \eqref{J1J2} is computed as 
    \begin{align} \label{J1_calculation}
    \begin{split}
        \JJ_1^{(\del)} &:= \int_x^{1-\del}\Big(h_1(t) - 2\log \Big(\frac{\sig}{t-x}\Big)\Big)\,dt = \int_x^{1-\del} \log \Big(\frac{t(1-t)}{\sig^2}\Big) \,dt  \\
        &\,= -\del \log \del + (1-\del) \log(1-\del) - 2(1-x-\del) + \log(1-x) - 2(1-\del)\log \sig.
    \end{split}
    \end{align}
    Then after straightforward computations, we obtain 
    \begin{equation} \label{A3}
        A_3' = (1-x)\Big[1-\frac12 \log (2\pi) + \frac12 \log \Big( \frac x{1-x} \Big) \Big]-c\log x.
    \end{equation}
    The $\OO(\sqrt{N})$ term immediately follows from $B_4$, and $A_5'$ is readily obtained from $B_2$, $B_5$, $E_5$ with the fact that $\theta_\del' = \theta_\del - (1-\del)$. Consequently, we have 
    \begin{equation} \label{A4A5}
        A_4' = B_4 = \sqrt{x(1-x)} \, \II_1  , \qquad A_5' = \frac12 x - \frac{1-3c+3c^2}{6}.
    \end{equation}
    For the $\OO(1)$ term, notice that
    \begin{align*}
        \JJ_2^{(\del)} &:= \int_x^{1-\del}\Big(h_2(t) - \frac{\sig^2}{(t-x)^2}\Big) \,dt = \frac1{12}\int_x^{1-\del}\Big(\frac{1}{t(1-t)}-1\Big)\,dt \\
        &= \frac{1}{12}\log \Big(\frac{(1-\del)(1-x)}{\del x}\Big) - \frac1{12}(1-x-\del).
    \end{align*}
    Then by using  
    \[
    f_1(t) = \frac12 T(1-t), \qquad T'(t) = 2\log \Big(\frac{(1-x)t}{x(1-t)}\Big),
    \]
    after some computations, one can show that
    \begin{equation} \label{A6}
        A_6' = \Big(\frac{1-3c+3c^2}{6}-\frac x2\Big) \log \Big(\frac{x}{1-x}\Big) + \Big(\frac x2 - \frac14\Big) \log (2\pi) + \frac{2(1-2x)}{3}\II_2 + \frac{1+11x}{12} - \zeta'(-1) + \log G(c+1). 
    \end{equation} 
    Finally, the facts that $N = n+\al+c$, $M = N^{\frac13}$ and $M_1 = N_1^{\frac1{12}}$ infer that
    \begin{align} \label{logPnfinal}
    \begin{split}
        \log \mc P_n^\CC(R;\al,c) &= A_1' n^2 + A_2' n \log n + (2(\al+c)A_1'+A_3') n + A_4' \sqrt{n} \\
        & \quad + ((\al+c)A_2'+A_5') \log n + (\al+c)^2 A_1'+(\al+c)(A_2'+A_3') + A_6' + \OO(n^{-\frac1{12}}).
    \end{split}
    \end{align}
    Using \eqref{A1A2}, \eqref{A3}, \eqref{A4A5}, \eqref{A6}, together with $x=1/(1+R^2)$, one can observe that the coefficients in \eqref{logPnfinal} are expressed as in Theorem \ref{thm:1.1}. This completes the proof.
\end{proof}

\section{Gap probabilities of the symplectic ensemble} \label{Section_symplectic}

This section is organised in parallel with the previous section, and we provide the proof of Theorem~\ref{thm:1.2}.
As before, we shall use the division \eqref{H:sum division S0123} of $\log \mathcal{P}_n^{ \mathbb{H} }$, as shown in Figure~\ref{Fig_sum division H}.
Subsection~\ref{Subsection S0123 H} is devoted to the analysis of $\widehat{S}_0$, $\widehat{S}_2$, and $\widehat{S}_3$, while Subsection~\ref{Subsection S1 H} focuses on $\widehat{S}_1$.
During the proof, it is convenient to redefine some of the notations in Section~\ref{Section_complex}, such as $N$, $M$, $N_1$, $M_1$, $\theta_M$, $\theta_{\delta}$, and $\theta_{\delta}'$, which allows us to reuse certain computations.

\subsection{Analysis of $\wh S_0$, $\wh S_2$ and $\wh S_3$} \label{Subsection S0123 H}

In this subsection, we provide the asymptotic behaviours of $\wh S_0$, $\wh S_2$ and $\wh S_3$. 

\begin{lem}[\textbf{Asymptotic behaviour of $\wh S_0$}] \label{H:S0}
    There exists $C>0$ such that
    \begin{equation} \label{H:S0 asymptotic}
        \wh S_0 = \OO(e^{-CN}), \qquad \text{as } N\to \infty.
    \end{equation}
\end{lem}

\begin{proof}
Recall that $m_k$ is given by \eqref{def of m_k}. 
Note that the sum $\wh S_0$ in \eqref{H:sum division S0123} corresponds to $\al+1 \le m_k \le \lfloor \del N \rfloor$. 
Let $X \sim B(2N,x)$, where $N$ is given by \eqref{def of N} and $x=1/(1+R^2)$ as in Lemma~\ref{Lem_finite expression}. Then since 
    \begin{align*}
        \log I_x(2k+2\al+2, 2n+2c-2k) = \log I_x(2m_k, 2N+1-2m_k) = \log (1-\PP(X<2m_k)),
    \end{align*}
   it follows from Lemma~\ref{Lemma:2.1} that 
    \[
    \PP(X<2m_k)\le \PP(X<2\del N) = \frac12 \erfc\Big(\frac{x-\del}{\sqrt{2x}}\sqrt{2N}\Big) \le \frac12 e^{-\frac{(x-\del)^2}{x}N} \le e^{-CN}.
    \]
    This completes the proof. 
\end{proof}

We write 
\begin{equation} \label{H:theta_Mdel}
    \theta_M := \lceil M \rceil - M, \qquad \theta_{\del} := \del N - \lfloor \del N-\tfrac12 \rfloor.
\end{equation}
Then we have the following. 

\begin{lem}[\textbf{Asymptotic behaviour of $\wh S_3$}] \label{H:S3}
    As $N \to \infty$, we have 
    \begin{equation}
    \wh S_3 = \wh D_3N + \wh D_5 \log N + \wh D_6 -\frac13 \frac{M^3}{N} + \OO\Big(\frac1M + \frac{M^2}N + \frac{M^4}{N^2}\Big),
    \end{equation}
    where 
    \begin{align*}
            \wh D_3 &= 2(M + \theta_M - c) \log x,\quad \wh D_5 = M^2+2\theta_M M + \theta_M^2 - c^2,  
            \\
            \wh D_6 &= -M^2 \log M + \Big(\frac32 + \log \Big(\frac{1-x}{x}\Big)\Big)M^2 - \Big(2\theta_M+\frac12\Big)M\log M \\
            & \quad + \Big[\frac12 - \frac12 \log(4\pi) + \Big(2+2\log \Big(\frac{1-x}x\Big)\Big)\theta_M\Big]M -\Big(\theta_M^2+\frac12 \theta_M - \frac1{24}\Big)\log M - \Big(\frac14 + \theta_M\Big) \log 2 \\
            &\quad + (\theta_M^2-c^2)\log \Big(\frac{1-x}x\Big) - \frac{1+2c+2\theta_M}{4}\log \pi - 2\zeta'(-1) + \log\Big(G(c+1)G(c+\tfrac32)\Big).
    \end{align*}
\end{lem}
\begin{proof}
     By letting $X \sim B(2N,x)$ and $Y \sim B(2N,1-x)$, the summand of $S_3$ is expressed as
    \[
    \log I_x(2m_k, 2N+1-2m_k) = \PP(X \ge 2m_k) = \PP(Y \le 2N-2m_k).
    \]
    Note that by \eqref{def of m_k}, the indices of the summation of $S_3$ correspond to 
    $$
    2c+1 \le 2N-2m_k \le 2\lceil M \rceil -1. 
    $$ 
    By the argument used to derive \eqref{eq:3.2.1}, \eqref{eq:3.2.2} and \eqref{eq:3.2.3}, we obtain
    \begin{equation} \label{H:S3_summand}
        \log \PP(Y \le 2\ell+1) = 2N\log x + (2\ell+1)\log(2N) +(2\ell+1) \log \Big( \frac{1-x}{x} \Big) -\log [(2\ell+1)!] - \frac{\ell^2}{N}+\OO\Big(\frac{\ell}{N}+\frac{\ell^3}{N^2}\Big),
    \end{equation}
    which leads to 
    \begin{align} \label{H:S3_with_ceilM}
        \begin{split}
            &\quad  \wh S_3 = \sum_{\ell = c}^{\ceil M -1} \log \PP(Y \le 2\ell+1) \\
            &= [2(\ceil M - c) \log x] N+ (\ceil M^2 - c^2)\log\Big(2N \frac{1-x}x\Big)- \sum_{\ell=c}^{\ceil M - 1}\log [(2\ell+1)!] - \frac{M^3}{3N}+\OO\Big(\frac{M^2}{N}+\frac{M^4}{N^2}\Big).
        \end{split}
    \end{align}
   Here, by using the well-known duplication formula for the gamma function and the definition of the Barnes $G$-function \eqref{Barnes G def}, 
     \[
     \prod_{\ell=0}^{n-1}(2\ell+1)! = \prod_{\ell=1}^{n}\Gamma(2\ell) = \prod_{\ell=1}^n \frac{2^{2\ell-1}}{\sqrt \pi} \Gamma(\ell)\Gamma(\ell+\tfrac12) = \frac{2^{n^2}}{\pi^{n/2}}\frac{G(n+1)G(n+\frac32)}{G(\frac32)}, 
     \]
     which implies
     \begin{align*}
         \sum_{\ell = c}^{\ceil{M}-1} \log[(2\ell+1)!] &= (\ceil{M}^2-c^2)\log 2 - \frac{\ceil M - c}{2}\log \pi + \log \bigg(\frac{G(\ceil M + 1) G(\ceil M + \frac32)}{G(c+1)G(c+\frac32)}\bigg).
     \end{align*}
     Then by using the asymptotic formula \eqref{Barnes G asymp}, we obtain 
     \begin{align} \label{H:prod_odd_factorial}
     \begin{split}
         \sum_{\ell = c}^{\ceil{M}-1} &\log[(2\ell+1)!] = \ceil M^2 \log \ceil M + \Big(\log 2 - \frac32 \Big) \ceil M^2 + \frac12 \ceil M \log \ceil M + \frac{\log (4\pi) -1}{2}\ceil M 
         \\
         & \quad -\frac1{24}\log \ceil M + \Big(\frac14 - c^2\Big) \log2 + \frac{1+2c}{4}\log \pi + 2\zeta'(-1) -\log\Big( G(c+1)G (c+\tfrac32 ) \Big) + \OO\Big(\frac1M\Big).
     \end{split} 
     \end{align}
     Now putting \eqref{H:prod_odd_factorial} and $\ceil M = M + \theta_M$ to \eqref{H:S3_with_ceilM} yields the lemma.
\end{proof}

Next, we analyse the summation $\widehat{S}_2$. 
For this, we shall use a slightly modified version of Lemma \ref{lem:3.3}. 

\begin{lem} \label{lem:4.3}
    Let $A, B, a_0, b_0, a_n, b_n, \mf m_{A,n},\mf m_{B,n}$ and $\mf m_{j,n}$ be as in Lemma \ref{lem:3.3}. Then as $n \to \infty$, we have
    \begin{align}
    \begin{split}
          \sum_{j=a_n}^{b_n}f\Big(\frac{2j+1}{2n}\Big) &= n \int_A^B f(x) \,dx - a_0f(A) + (1+b_0)f(B) + \frac{(1-12a_0^2)f'(A)-(1-12(1+b_0)^2)f'(B)}{24n} 
        \\
        &\quad + \OO\bigg(\frac{\mf m_{A,n}(f'') + \mf m_{B,n}(f'')}{n^2}+\sum_{j=a_n}^{b_n}\frac{\mf m_{j,n}(f''')}{n^3}\bigg).
    \end{split}
    \end{align}
\end{lem}
\begin{proof}
    This immediately follows from Lemma \ref{lem:3.3} with $\widetilde f(t) := f(t+\frac1{2n})$ and simplifications.
\end{proof}

\begin{lem}[\textbf{Asymptotic behaviour of $\wh S_2$}] \label{H:S2}
    As $N \to \infty$, we have
    \begin{equation}
        \wh S_2 = \wh C_1 N^2 + \wh C_2 N\log N +\wh C_3 N + \wh C_5\log N + \wh C_6+ \OO\Big(\frac1M+\frac{M}{N}\Big),
    \end{equation}
   where
    \begin{align*}
        \wh C_1 &= -2\int_{M/N}^\del f_1(t) \,dt, \quad \wh C_2 = -\frac{\del}{2}, \quad \wh C_3 = 2\theta_M f_1\Big(\frac MN \Big) - 2(1-\theta_\del)f_1(\del) - \frac\del2\log 2 -\int_{M/N}^\del f_2(t)\,dt\\
        \wh C_5 &= \frac12(M+\theta_{\del}+\theta_M -1) -\frac1{24}, \\
        \wh C_6 &= \Big(\theta_M^2-\frac1{12}\Big)f_1'\Big(\frac MN\Big) + \Big(\frac1{12}-(1-\theta_\del)^2\Big)f_1'(\del) + \theta_M f_2\Big(\frac MN\Big) - (1-\theta_\del)f_2(\del) \\
        &\quad +\frac1{24}\log M + \frac1{24}\Big(\del - \log \Big(\frac{\del}{1-\del}\Big)\Big) - \frac{\del x}{2(1-x-\del)}-\frac x2 \log \Big(\frac{1-x-\del}{1-x}\Big)+\frac{\log 2}{2}(M+\theta_\del+\theta_M-1).
    \end{align*}
    Here, $\theta_M$ and $\theta_{\del}$ are given by \eqref{H:theta_Mdel} and the functions $f_1$ and $f_2$ are given by \eqref{def:f1f2}.
\end{lem}
\begin{proof}
Let $X \sim B(2N,x)$. Then as in \eqref{S2_summation}, we have
\begin{equation} \label{H:S2 summation}
    \widehat{S}_2 = \sum_{k=k_+}^{k_M} \log I_x(2m_k, 2N+1-2m_k) = \sum_{k=k_+}^{k_M} \log \PP(X \ge 2m_k) = \sum_{j=\ceil M}^{\floor{\del N-\frac12}}\log \PP(X \ge 2N-2j-1).
\end{equation}
   Then by using Lemma \ref{lem:2.7}, we have
    \[
        \log \PP(X \ge 2N(1-t)) = -2N f_1(t) - \frac12\log (2N) -f_2(t) + \frac1{2N}f_3(t) + \OO\Big(\frac1{t^3N^3}+\frac{1}{N^2}\Big),
    \]   
   where $f_3$ is given by \eqref{def:f3}. 
   By employing Lemma~\ref{lem:4.3} with $A(N) = M/N$, $a_0(N) = \theta_M$, $B(N) = \del$ and $b_0(N) = -\theta_\del$, we obtain
    \begin{align*}
        -2N\sum_{j=\ceil M}^{\floor{\del N - \frac12}} f_1\Big(\frac{2j+1}{2N}\Big)&= -2N^2\int_{M/N}^\del f_1(t)\,dt + \Big[2\theta_Mf_1\Big(\frac MN\Big) - 2(1-\theta_\del)f_1(\del)\Big]N \\
        &\quad + \Big(\theta_M^2-\frac1{12}\Big)f_1'\Big(\frac MN\Big) + \Big(\frac1{12}-(1-\theta_\del)^2\Big)f_1'(\del) + \OO\Big(\frac1{M}\Big), \\
        -\sum_{j=\ceil M}^{\floor{\del N - \frac12}} f_2\Big(\frac{2j+1}{2N}\Big)&= -N \int_{M/N}^\del f_2(t)\, dt + \theta_M f_2\Big(\frac MN\Big) -(1-\theta_\del)f_2(\del) + \OO\Big(\frac 1M + \frac1N\Big), \\
        \frac{1}{2N}\sum_{j=\ceil M}^{\floor{\del N - \frac12}} f_3\Big(\frac{2j+1}{2N}\Big) &= \frac1{24}\bigg(\del - \log \Big(\frac{\del}{1-\del}\Big)\bigg)- \frac{\del x}{2(1-x-\del)} -\frac x2\log \Big(\frac{1-x-\del}{1-x}\Big) \\
        &\quad + \frac1{24}\log \Big(\frac MN\Big) + \OO\Big(\frac 1M + \frac MN\Big).
    \end{align*}
    Combining all of the above, we obtain the desired asympototic formula.
\end{proof}

As a counterpart of Lemma~\ref{S2+S3}, we have the following.

\begin{lem}[\textbf{Asymptotic behaviour of $\wh S_0+ \wh S_2 + \wh S_3$}] \label{H:S2+S3}
    As $N \to \infty$, we have
    \begin{equation}
    \wh S_0 + \wh S_2 + \wh S_3 = \wh E_1 N^2 + \wh E_2  N\log N + \wh E_3 N + \wh E_5 \log N + \wh E_6 + \OO \Big(\frac{1}{M}+\frac{M^2}{N}+\frac{M^4}{N^2}\Big),
    \end{equation}
where 
    \begin{align*}
        \wh E_1 &= \log x - \del^2 \log \Big(\frac\del {1-x}\Big) +  (1-\del)^2 \log \Big(\frac{1-\del}{x}\Big) + \del, \qquad \wh E_2 = -\frac\del2, \\
        \wh E_3 &= -2c \log x + (1-x-\del) \log \Big(\frac{1-x-\del}{1-x}\Big)+\del - \frac12 \del \log(4\pi\del) - \frac12(1-\del)\log(1-\del) - 2(1 - \theta_\del) f_1(\del), \\
        \wh E_5 &= \frac{1}{2}\theta_\del - c^2 - \frac{11}{24}, \\
        \wh E_6 &= \Big(c^2-\frac1{12}\Big)\log\Big(\frac{x}{1-x}\Big)-\frac14 \log(2\pi)-\frac c2 \log \pi -2\zeta'(-1)+\log \Big(G(c+1)G(c+\tfrac32)\Big) \\
        &\quad +\Big(\frac1{12}-(1-\theta_\del)^2\Big)f_1'(\del) - (1-\theta_\del)f_2(\del) + \frac1{24}\Big(\del-\log \Big(\frac{\del}{1-\del}\Big)\Big) \\
        &\quad - \frac{\del x}{2(1-x-\del)} -\frac x2\log\Big(\frac{1-x-\del}{1-x}\Big) - \frac{1-\theta_\del}2 \log 2.
    \end{align*}
    Here, $f_1$ and $f_2$ are given by \eqref{def:f1f2}.
\end{lem}

\begin{proof}
    Note that by Lemma~\ref{H:S0}, one can ignore the exponentially small $\wh S_0$. Using \eqref{integral of f1} and \eqref{integral of f2} again, the terms in Lemma \ref{H:S2} can be rewritten as
    \begin{align*}
        \wh C_1 N^2 &= \Big[-\del^2\log \Big(\frac{\del}{1-x}\Big)+(1-\del)^2\log\Big(\frac{1-\del}{x}\Big)+\del+\log x\Big]N^2 - (2M\log x)N - M^2 \log N \\
        &\quad + \Big[M^2 \log M + M^2 \log \Big(\frac{x}{1-x}\Big) - \frac32 M^2\Big]+\frac{M^3}{3N} + \OO\Big(\frac{M^4}{N^2}\Big), \\
        \wh C_3 N &= \Big[(1-x-\del)\log(1-x-\del) - \del \log \Big(\frac{\sqrt{4\pi}}{1-x}\Big)+\del-\frac12 \del \log \del - \frac12 (1-\del)\log(1-\del) - 2(1-\theta_\del)f_1(\del)\Big]N \\
        &\quad - (2\theta_M\log x +(1-x)\log(1-x)) N - \Big(2\theta_M+\frac12\Big) M \log N + \Big(2\theta_M+\frac12\Big)M \log M \\
        &\quad + \Big[\frac12 \log(2\pi) - 2\theta_M + 2\theta_M \log \Big(\frac{x}{1-x}\Big)- \frac12\Big]M + \OO\Big(\frac{M^2}{N}\Big), \\
        \wh C_6 &= \Big(\frac1{12}-\frac12 \theta_M -\theta_M^2\Big) \log N + \frac{\log 2}{2}M + \Big(-\frac1{24}+\frac12 \theta_M + \theta_M^2\Big)\log M + \frac12 \theta_M \log(4\pi) \\
        &\quad + \Big(\theta_M^2 - \frac1{12}\Big) \log \Big(\frac{x}{1-x}\Big) - \Big(\frac1{12}-(1-\theta_\del)^2\Big)f_1'(\del) - (1-\theta_\del)f_2(\del) \\
        &\quad+ \frac1{24}\Big(\del- \log \Big(\frac \del{1-\del}\Big)\Big) - \frac{\del x}{2(1-x-\del)} - \frac x2 \log \Big(\frac{1-x-\del}{1-x}\Big)-\frac{1-\theta_\del}{2}\log 2 + \OO\Big(\frac MN\Big).
    \end{align*}
    Then these terms cancel out with the terms of Lemma \ref{H:S3} up to the error $\OO(M^{-1}+M^2N^{-1}+M^4N^{-2})$. 
   After simplifications, only the coefficients depending on $\delta$ remain, and the lemma follows.
\end{proof}

\subsection{Analysis of $\wh S_1$} \label{Subsection S1 H}

In this subsection, we derive the asymptotic behaviour of $\wh S_1$, which corresponds to the symmetric regime of the symplectic ensemble. We shall follow the strategy used in Subsection~\ref{Subsection_S1}.
Recall that $N$ is given by \eqref{def of N}. We denote
\begin{equation}
N_1 = N+\frac12, \qquad M_1 = N_1^{\frac1{12}}, \qquad  x_k = \frac{m_k}{N_1}.
\end{equation}
Notice here that $N_1$ is an integer. It is convenient to redefine
\begin{equation} \label{H:g_pm}
    g_- = \Big\lceil N_1\Big(x+\frac{M_1}{\sqrt{N_1}}\Big)-\al-1\Big\rceil, \qquad g_+ = \Big\lfloor N_1\Big(x+\frac{M_1}{\sqrt{N_1}}\Big)-\al-1\Big\rfloor.
\end{equation}
As an analogue of \eqref{sum division S1 123}, we divide $\wh S_1$ by
\begin{equation} \label{H:sum division S1 123}
    \wh S_1 = \wh S_1^{(1)} + \wh S_1^{(2)} + \wh S_1^{(3)},
\end{equation}
where
\begin{align}
\begin{split}
    \wh S_1^{(1)} &= \sum_{k=k_-+1}^{g_--1} \log I_x(2m_k,2N_1-2m_k),\\
    \wh S_1^{(2)} &= \sum_{k=g_-}^{g_+} \log I_x(2m_k,2N_1-2m_k),\\
    \wh S_1^{(3)} &= \sum_{k=g_++1}^{k_+-1} \log I_x(2m_k,2N_1-2m_k).
\end{split}
\end{align}
With our new notations, we can efficiently reuse most of the calculations from the complex case. 
In particular, by setting $\eta_k$ as defined in \eqref{def of eta}, Lemma \ref{beta} yields
\begin{equation}\label{H:summand of S1}
    I_x(2m_k,2N_1-2m_k) = \frac12 \erfc(-\eta_k\sqrt{N_1}) + \frac{1}{\sqrt{4\pi N_1}} e^{-\eta_k^2 N_1}\Big(c_0(\eta_k) + \frac{1}{2N_1}c_1(\eta_k) + \OO\Big(\frac1{N_1^2}\Big)\Big).
\end{equation}
This also follows by putting $2N_1$ to $N_1$ in \eqref{summand of S1}.

\begin{lem} \label{H:S1^(1)}
    There exists a constant $C>0$ such that
    \begin{equation} \label{H:S1^(1) asymptotic}
        \wh S_1^{(1)} = \OO(e^{-CM_1^2}), \qquad \text{ as } N \to \infty.
    \end{equation}
\end{lem}
\begin{proof}
    An argument similar to Lemma \ref{S1^(1)} shows
    \[
    I_x(2m_k, 2N_1 - 2m_k) = 1 - \frac{1}{\sqrt{4\pi N_1} \eta_k}e^{-\eta_k^2 N_1}\OO(1) + \frac{1}{\sqrt{4\pi N_1}}e^{-\eta_k^2 N_1}\OO(1) = 1+\OO(e^{-cM_1^2}),
    \]
    which verifies the lemma.
\end{proof}

We write
\begin{equation} \label{H:theta_pm}
    \theta_- = g_- - \Big(N_1\Big(x-\frac{M_1}{\sqrt{N_1}}\Big)-\al-1\Big), \qquad \theta_+ = \Big(N_1\Big(x+\frac{M_1}{\sqrt{N_1}}\Big) - \al - 1\Big) - g_+.
\end{equation}
Define $M_k = \sqrt{N_1}(x_k-x)$. Note that Lemma \ref{Riemann sum of M_k} is still valid for the new $N_1$ and $M_1$ since $M_k, \theta_-$ and $\theta_+$ are defined in the same way as in the complex case, and its proof did not use any relations between $N_1, M_1$, and $n$.

\begin{lem} \label{H:S1^(2)}
    As $N \to \infty$,
    \begin{equation}
    \wh S_1^{(2)} = \wh B_4^{(2)}\sqrt{N_1} + \wh B_6^{(2)} + \OO\Big(\frac{M_1^5}{\sqrt{N_1}}+\frac{M_1^{10}}{N_1}\Big),
    \end{equation}
    where 
    \begin{align*}
        \wh B_4^{(2)} = \int_{-M_1}^{M_1} \log\Big[\frac12 \erfc\Big(\frac t\sig\Big)\Big]\,dt, \qquad \wh B_6^{(2)} = \int_{-M_1}^{M_1} \widehat g(t) \,dt + \Big(\frac12 - \theta_+\Big)\log\Big[\frac12 \erfc\Big(\frac{M_1}\sig\Big)\Big].
    \end{align*}
    Here, $\widehat g$ is a rescaled version of $g$ defined as \eqref{def:g}:
    \[
    \widehat g(t) = \frac{1}{\sqrt 2}g(\sqrt 2 t) = \frac{1-2x}{3\sig^3\sqrt \pi \erfc(t/\sig)}(t^2+\sig^2)e^{-\frac{t^2}{\sig^2}}.
    \]
\end{lem}
\begin{proof}
    By using \eqref{eq:3.9.1}, we have
    \[
    \Del_k := -\eta_k \sqrt{N_1} - \frac{M_k}{\sig} = -\frac{1-2x}{6\sig^3}\frac{M_k^2}{\sqrt{N_1}}+\frac{5-14\sig^2}{72\sig^5}\frac{M_k^3}{N_1}+\OO\Big(\frac{M_k^4}{N_1\sqrt{N_1}}\Big).
    \]
    This in turn implies that 
    \begin{align*}
       &\quad  \erfc(-\eta_k\sqrt{N_1}) = \erfc\Big(\frac{M_k}{\sig}\Big) - \frac{2}{\sqrt{\pi}}e^{-\frac{M_k^2}{\sig^2}}\Big[\Del_k - \frac{M_k}{\sig}\Del_k^2 + \OO((1+M_k^2)\Del_k^3)\Big] \\
        &= \erfc\Big(\frac{M_k}{\sig}\Big)-\frac{2}{\sqrt \pi}e^{-\frac{M_k^2}{\sig^2}}\Big[-\frac{1-2x}{6\sig^3}\frac{M_k^2}{\sqrt{N_1}}+\frac{5-14\sig^2}{72\sig^5}\frac{M_k^3}{N_1}-\frac{1-4\sig^2}{36\sig^7}\frac{M_k^5}{N_1}+\OO\Big(\frac{M_k^4+M_k^8}{N_1\sqrt{N_1}}\Big)\Big].
    \end{align*}
   Following the argument in Lemma \ref{S1^(2)}, we can then obtain
    \begin{align*}
        I_x(2m_k, 2N-2m_k) &= \frac12 \erfc(-\eta_k \sqrt{N_1}) + \frac{1}{2\sqrt{\pi {N_1}}}e^{-\eta_k^2{N_1}}\Big(c_0(\eta_k) + \OO\Big(\frac1{N_1}\Big)\Big) \\
        &= \widehat u_1(M_k) + \frac{1}{\sqrt{N_1}} \widehat  u_2(M_k) + \frac{1}{N_1} \widehat  u_3(M_k) + \OO\Big(\frac{1+M_k^8}{N_1\sqrt{N_1}}e^{-\frac{M_k^2}{\sig^2}}\Big),
    \end{align*}
    where
    \begin{align*}
        \widehat  u_1(t) &= \frac12 \erfc\Big(\frac{t}{ \sig}\Big), \qquad \widehat  u_2(t) = \frac{1-2x}{6\sig\sqrt{\pi}}e^{-\frac{t^2}{\sig^2}} \bigg(\frac{t^2}{\sig^2}+1\bigg), \\
       \widehat  u_3(t) &= \frac{1}{2\sqrt{\pi}}e^{-\frac{t^2}{\sig^2}}\bigg(\frac{1-4\sig^2}{18\sig^7}t^5 - \frac{1-2\sig^2}{36\sig^5}t^3 +\frac{1+\sig^2}{12\sig^3}t\bigg).
    \end{align*}
    Therefore, it follows that
    \begin{align*}
        \log I_x(2m_k, 2N_1-2m_k) &= \log \widehat 
 u_1(M_k) + \frac{1}{\sqrt{N_1}}  \widehat g(M_k) + \frac{1}{N_1} \widehat g_1(M_k) \\
        &\quad + \OO\bigg(\frac{1}{N_1\sqrt N_1}\bigg(\frac{\widehat  u_2(M_k)^3}{\widehat 
 u_1(M_k)^3}+\frac{\widehat  u_2(M_k) \widehat 
 u_3(M_k)}{\widehat  u_1(M_k)^2}+\frac{1+M_k^8}{\widehat  u_1(M_k)} e^{-\frac{M_k^2}{\sig^2}}\bigg)\bigg),
    \end{align*}
    where 
    \[
    \widehat g(t) = \frac{\widehat  u_2(t)}{\widehat  u_1(t)}, \qquad \widehat g_1(t) = -\frac{\widehat u_2(t)^2}{2\widehat u_1(t)^2} + \frac{\widehat u_3(t)}{\widehat u_1(t)}.
    \]
    Note that $\widehat g_1(t) = \OO(t^4)$ as $t \to +\infty$.
    Then by applying  Lemma \ref{Riemann sum of M_k} to $\log \widehat  u_1(t)$, $\widehat g(t)$, $\widehat g_1(t)$ and the error term, the lemma follows. 
\end{proof}

It now suffices to analyse  $\wh S_1^{(3)}$. 
We need a modified version of Lemma \ref{lem:3.9} since the definition of $\theta_\del$ is different from the complex case. 

\begin{lem} \label{lem:4.8}
    Let $h \in C^3((x,1-\del))$. Then as $N \to \infty$,
    \begin{align*}
        \sum_{k=g_++1}^{k_+-1} h(x_k) &= N_1 \int_{x+\frac{M_1}{\sqrt{N_1}}}^{1-\del} h(t) \,dt + \Big(\theta_+ - \frac12 \Big)h\Big(x+\frac{M_1}{\sqrt{N_1}}\Big) + \Big(\theta_\del'-\frac12\Big) h(1-\del) \\
        &\quad -\frac{1-6\theta_++6\theta_+^2}{12N_1}h'\Big(x+\frac{M_1}{\sqrt{N_1}}\Big) + \frac{1-6\theta_\del'+6\theta_\del'^2}{12N_1}h'(1-\del)+\OO\Big(\frac{\mf m(h''')+\mf m(h'')}{N_1^2}\Big),
    \end{align*}
    where $\theta_\del' := \theta_{\del}-(1-\frac\del 2)$ and $\mf m$ is given by \eqref{mf lem3.9}. 
    Here, $\theta_+$ and $\theta_\del$ are given by \eqref{H:theta_pm} and  \eqref{H:theta_Mdel}.
\end{lem}
\begin{proof}
   Let $A = x+\frac{M_1}{\sqrt{N_1}}$, $a_0 = 1-\theta_+$, $B = 1-\del$ and $b_0 = \theta_\del'-1$. Then we have
    \begin{align*}
        a_n &= AN_1 + a_0 = \Big(x+\frac{M_1}{\sqrt{N_1}}\Big)N_1 + 1-\theta_+ = g_++\al+2, \\
        b_n &= BN_1+b_0 = (1-\del)N_1+\theta_\del'-1 =  (1-\del)N + \theta_\del - \frac 32= k_++\al.
    \end{align*}
    Now Lemma \ref{lem:3.3} completes the proof.
\end{proof}

\begin{lem}\label{H:S1^(3)}
    As $N \to \infty$, we have
    \begin{align}
    \begin{split}
        \wh S_1^{(3)} &= \wh B_1^{(3)} N_1^2 + \wh B_2^{(3)} N_1 \log N_1 + \wh B_3^{(3)}N_1 + \wh B_4^{(3)} \sqrt{N_1} + \wh B_5^{(3)} \log N_1 + \wh B_6^{(3)} \\
        &\quad + \wh B_7^{(3)}\frac{1}{N_1}+\wh B_8^{(3)} \frac{1}{N_1^2}+\OO\Big(\frac{M_1}{\sqrt{N_1}}+\frac{1}{M_1^2}+\frac{\sqrt{N_1}}{M_1^7}\Big),  
    \end{split}
    \end{align}
    where
    \begin{align*}
        \wh B_1^{(3)} &= -\int_{x+\frac{M_1}{\sqrt{N_1}}}^{1-\del} T(t)\, dt, \qquad \wh B_2^{(3)} = -\frac12(1-\del-x), \\
        \wh B_3^{(3)} &= \frac12 \int_{x+\frac{M_1}{\sqrt{N_1}}}^{1-\del} h_1(t)\, dt -\frac12(1-\del-x)\log(4\pi) + \Big(\frac12 - \theta_+\Big)T_{M_1}+\Big(\frac12 - \theta_\del'\Big)T_{1-\del}, \\
        \wh B_4^{(3)} &= \frac12 M_1\log(4\pi N_1), \qquad \wh B_5^{(3)} = \frac12(1-\theta_+-\theta_\del'), \\
        \wh B_6^{(3)} &= \frac12 (1-\theta_+-\theta_\del') \log(4\pi) + \frac{1-6\theta_++6\theta_+^2}{12}T_{M_1}' - \frac{1-6\theta_\del'+6\theta_\del'^2}{12}T_{1-\del}' \\
        &\quad + \frac12\Big(\theta_+-\frac12\Big)h_1\Big(x+\frac{M_1}{\sqrt{N_1}}\Big) + \frac12\Big(\theta_\del'-\frac12\Big)h_1(1-\del) - \frac12\int_{x+\frac{M_1}{\sqrt{N_1}}}^{1-\del}h_2(t) \,dt, \\
        \wh B_7^{(3)} &= \frac58 \int_{x+\frac{M_1}{\sqrt{N_1}}}^{1-\del} \frac{1}{T(t)^2}\,dt, \qquad \wh B_8^{(3)} = -\frac{37}{24}\int_{x+\frac{M_1}{\sqrt{N_1}}}^{1-\del}\frac1{T(t)^3}\,dt.
    \end{align*}
    Here $T$ is given by \eqref{def:T}, 
    \begin{equation} \label{def of TM1 etc H}
      T_{M_1} = T(x+\tfrac{M_1}{\sqrt{N_1}}), \qquad T_{1-\del} = T(1-\del), \qquad T_{M_1}' = T'(x+\tfrac{M_1}{\sqrt{N_1}}), \qquad T_{1-\del}' = T'(1-\del).  
    \end{equation}
   Also, $h_1$ and $h_2$ are given by \eqref{def:h1h2}, $\theta_+$ is given by \eqref{H:theta_pm}, and $\theta_\del'$ is given by Lemma \ref{lem:4.8}.
\end{lem}

\begin{proof}

By putting $2N_1$ to $N_1$ in \eqref{S1^(3):logIx}, we have 
\begin{align}\label{H:S1^(3):logIx}
\begin{split}
    I_x(2m_k, 2N_1 - 2m_k) &= -\eta_k^2 N_1 - \frac12 \log N_1 - \frac12 \log \eta_k^2 - \frac12 \log(4\pi) + \log \al_k - \frac{\beta_k}{2\al_k N_1}\\
    &\quad +\frac{5}{8\eta_k^4N_1^2}-\frac{37}{24\eta_k^6N_1^3} + \OO\Big(\frac1{\eta_k^8N_1^4}+\frac{1}{\eta_k^2N_1^2(x_k-x)}+\frac{x_k-x}{\eta_k^4N_1^2}+\frac{x_k-x}{\eta_k^6N_1^3}\Big),
\end{split}
\end{align} 
where $\al_k$ and $\beta_k$ are defined by \eqref{def:alphabeta}. 
Let us write 
    \begin{align*}
        \wh{\mc S}_0 &= -\frac12 \sum_{k=g_++1}^{k_+-1} \Big( \log N_1 + \log(4\pi)\Big), \qquad\wh{\mc S}_1 = -N_1 \sum_{k=g_++1}^{k_+-1} \eta_k^2, \qquad \wh{\mc S}_2 = \sum_{k=g_++1}^{k_+-1} \Big(-\frac12 \log \eta_k^2 + \log \al_k \Big), \\
        \wh{\mc S}_3 &= -\frac1{2N_1} \sum_{k=g_++1}^{k_+-1} \frac{\beta_k}{\al_k}, \qquad  \wh{\mc S}_4 = \frac{5}{8N_1^2} \sum_{k=g_++1}^{k_+-1} \frac{1}{\eta_k^4}, \qquad \wh{\mc S}_5 = -\frac{37}{24N_1^3} \sum_{k=g_++1}^{k_+-1} \frac{1}{\eta_k^6}.
    \end{align*}
    Here, each sum can then be computed using Lemma~\ref{lem:4.8} and expressed with formulas similar to those in Lemma~\ref{S1^(3)}. We leave the details to the interested reader. 
    Then the desired asymptotic formula follows. 
\end{proof}

We are now ready to state the asymptotic behaviour of $\wh S_1. $

\begin{lem}[\textbf{Asymptotic behaviour of $\wh S_1$}] \label{H:S1}
    As $N \to \infty$, we have
    \begin{equation}
    \wh S_1 = \wh B_1 N_1^2 + \wh B_2 N_1 \log N_1 + \wh B_3 N_1 + \wh B_4 \sqrt {N_1} + \wh B_5 \log N_1 + \wh B_6 + \OO\Big(\frac{M_1^5}{\sqrt{N_1}}+\frac1{M_1^2}+\frac{\sqrt{N_1}}{M_1^7}\Big),
    \end{equation}
  where 
    \begin{align*}
        \wh B_1 &= -\int_x^{1-\del} T(t)\, dt, \qquad \wh B_2 = -\frac12(1-x-\del), \\
        \wh B_3 &= \frac12 \JJ_1^{(\del)}-\Big(\theta_\del'-\frac12\Big)T_{1-\del}+(1-\del-x)\Big(1+\log \Big( \frac{\sig}{\sqrt{4\pi}(1-\del-x)} \Big) \Big),
        \\
        \wh B_4 &= \frac{ \sig }{\sqrt 2} \,\II_1, \qquad \wh B_5 = \frac14 - \frac12 \theta_\del', \\
        \wh B_6 &= \frac{1-2x}{3}\II_2-\frac12 \JJ_2^{(\del)} +\Big(\frac14 - \frac12\theta_\del'\Big)\log(4\pi) - \frac{1-6\theta_\del'+6\theta_\del'^2}{12}T_{1-\del}'+\Big(\frac12\theta_\del'-\frac14\Big)h_1(1-\del) + \frac{\sig^2}{2(1-\del-x)}.
    \end{align*}
    Here $T$, $h_1$, $T_{1-\del}$ and $T_{1-\del}'$ are given by \eqref{def:T}, \eqref{def:h1h2} and \eqref{def of TM1 etc H}, whereas $\II_1, \II_2, \JJ_1^{(\del)}$ and $\JJ_2^{(\del)}$ are given by \eqref{I1}, \eqref{I2} and \eqref{J1J2}.
\end{lem}

\begin{proof}
    Combining Lemmas \ref{H:S1^(1)}, \ref{H:S1^(2)} and \ref{H:S1^(3)}, we have
    \begin{align}
    \begin{split}
        \wh S_1 &= \wh B_1^{(3)}N_1^2 + \wh B_2^{(3)}N_1\log N_1 + \wh B_3^{(3)} N_1 + \wh B_4'\sqrt{N_1}+\wh B_5^{(3)}\log N_1 + \wh B_6'\\
    &\quad + \wh B_7^{(3)} \frac{1}{N_1}+ \wh B_8^{(3)} \frac{1}{N_1^2} + \OO\Big(\frac{\sqrt{N_1}}{M_1^7}+\frac{1}{M_1^2}+\frac{M_1^5}{\sqrt{N_1}}+\frac{M_1^{10}}{N_1}\Big),
    \end{split}
    \end{align}
    where
    \[
    \wh B_4' = \wh B_4^{(2)} + \wh B_4^{(3)}, \qquad 
    \wh B_6' = \wh B_6^{(2)} + \wh B_6^{(3)}.
    \]
    Through long but straightforward computations similar to those in Lemma \ref{S1}, we obtain
    \begin{align*}
        \wh B_1^{(3)}N_1^2 &= -N_1^2\int_x^{1-\del} T(t) \,dt + \frac{1}{3\sig^2}M_1^3 \sqrt{N_1}-\frac{1-2x}{12\sig^4}M_1^4 + \OO\Big(\frac{M_1^5}{\sqrt{N_1}}\Big), \\
        \wh B_2^{(3)}N_1 \log N_1 &= -\frac12 (1-\del-x) N_1 \log N_1, \\
        \wh B_3^{(3)}N_1 &= \Big[\frac12 \JJ_1^{(\del)}-\Big(\theta_\del'-\frac12\Big)T_{1-\del}+(1-x-\del)\Big(1+\log \Big(\frac{\sig}{\sqrt{4\pi}(1-x-\del)}\Big)\Big)\Big]N_1 \\
        &\quad + \Big(M_1\log M_1 - \frac12 M_1 \log N_1 - M_1(1+\log \sig)\Big)\sqrt{N_1} + \frac{1+2x-4\theta_+}{4\sig^2}M_1^2 + \OO\Big(\frac{M_1^3}{\sqrt{N_1}}\Big), \\
        \wh B_4' \sqrt{N_1} &= -\frac{M_1^3\sqrt{N_1}}{3\sig^2}+\frac12 M_1 \sqrt{N_1}\log N_1 - M_1 \log M_1 \sqrt{N_1}+(1+\log \sig)M_1\sqrt{N_1} \\
        &\quad + \frac1{\sqrt 2}\sig \, \II_1 \sqrt{N_1}+\frac{\sig^2}{2}\frac{\sqrt{N_1}}{M_1}-\frac{5\sig^4}{24}\frac{\sqrt{N_1}}{M_1^3}+\frac{37\sig^6}{120}\frac{\sqrt{N_1}}{M_1^5}+\OO\Big(\frac{\sqrt{N_1}}{M_1^7}\Big), \\
        \wh B_5^{(3)} \log N_1 &= \frac12(1-\theta_+-\theta_\del') \log N_1, \\
        \wh B_6^{(2)} &= \frac{1-2x}{12\sig^4}M_1^4 - \frac{1+2x-4\theta_+}{4\sig^2}M_1^2 + \Big(\theta_+-\frac12\Big) \log \Big(\frac{\sqrt{4\pi}M_1}{\sig}\Big)+\frac{1-2x}{3}\II_2 + \OO\Big(\frac1{M_1^2}\Big), \\
        \wh B_6^{(3)} &= -\frac{\sig^2}{2}\frac{\sqrt{N_1}}{M_1} + \Big(\frac{\theta_+}{2}-\frac14\Big) \log N_1 - \Big(\theta_+-\frac12\Big) \log \Big(\frac{\sqrt{4\pi}M_1}{\sig}\Big)+\Big(\frac14 - \frac12 \theta_\del'\Big) \log(4\pi) \\
        &\quad -\frac{1-6\theta_\del'+6\theta_\del'^2}{12}T_{1-\del}'+\Big(\frac12 \theta_\del'-\frac14\Big) h_1(1-\del) + \frac{\sig^2}{2(1-\del-x)}-\frac12 \JJ_2^{(\del)} + \OO\Big(\frac{M_1}{\sqrt{N_1}}\Big), \\
        \frac{1}{N_1}\wh B_7^{(3)} &= \frac{5\sig^4}{24}\frac{\sqrt{N_1}}{M_1^3}+ \OO\Big(\frac{1}{M_1^4}+\frac{1}{N_1^2}+\frac{M_1}{N_1^2\sqrt{N_1}}\Big), \\
        \frac{1}{N_1^2}\wh B_8^{(3)} &= -\frac{37\sig^6}{120}\frac{\sqrt{N_1}}{M_1^5} + \OO\Big(\frac{1}{M_1^4}+\frac{1}{N_1^2}+\frac{M_1}{N_1^2\sqrt{N_1}}\Big).
    \end{align*}
  Combining these equations, the lemma follows. 
\end{proof}

\subsection{Proof of Theorem~\ref{thm:1.2}}

We now complete the proof of Theorem~\ref{thm:1.2}. 

\begin{proof}[Proof of Theorem \ref{thm:1.2}]
    Since $N_1 = N + \frac12$, Lemma~\ref{H:S1} gives rise to 
    \begin{align*}
        \wh S_1 &= \wh B_1 N^2 + \wh B_2 N \log N + (\wh B_1 + \wh B_3) N + \wh B_4 \sqrt{N}\\
        &\quad + \Big(\frac12\wh B_2 + \wh B_5\Big) \log N + \frac14 \wh B_1+\frac12 \wh B_2+\frac12 \wh B_3+ \wh B_6 + \OO\Big(\frac1{\sqrt N}\Big).
    \end{align*}
    Combining this with Lemma~\ref{H:S2+S3}, we have 
    \begin{align*}
        \log \mc P_n^\H(R;\al,c) &= \wh A_1' N^2 + \wh A_2' N \log N + \wh A_3' N + \wh A_4' \sqrt{N}+\wh A_5' \log N + \wh A_6' \\
        &\quad + \OO\Big(\frac1M + \frac{M^2}{N}+\frac{M^4}{N^2}+\frac{\sqrt{N_1}}{M_1^7}+\frac{1}{M_1^2}+\frac{M_1^5}{\sqrt{N_1}}+\frac{M_1^{10}}{N_1}\Big)
    \end{align*}
    as $N \to \infty$, where
    \begin{align*}
        \wh A_1' &= \wh B_1 + \wh E_1, \qquad \wh A_2' = \wh B_2 +\wh E_2, \qquad \wh A_3' = \wh B_1 + \wh B_3 +\wh E_3, \\
        \wh A_4' &= \wh B_4, \qquad \wh A_5' = \frac12\wh B_2+\wh B_5 + \wh E_5, \qquad \wh A_6' = \frac14\wh B_1+\frac12\wh B_2+\frac12 \wh B_3+\wh B_6+\wh E_6.
    \end{align*}
    After straightforward simplifications, reusing some computations from the proof of Theorem \ref{thm:1.1}, we obtain 
    \begin{align} \label{H:A tilde}
    \begin{split}
        \wh A_1' &= 1-x+\log x, \qquad \wh A_2' = -\frac12 (1-x), \\
        \wh A_3' &= (1-x)\Big[1-\frac12 \log(4\pi) + \frac12 \log \Big(\frac{x}{1-x}\Big)\Big]-\Big(2c+\frac12\Big)\log x, \\
        \wh A_4' &= \sqrt{\frac{x(1-x)}2}\,\II_1, \qquad \wh A_5' =  - c^2+ \frac14 x + \frac1{24}, \\
        \wh A_6' &= \Big(c^2-\frac14 x-\frac1{24}\Big) \log \Big( \frac{x}{1-x} \Big) - \Big(c+\frac{1-x}{2}\Big) \frac{\log(4\pi)}{2}+\Big(c+\frac14\Big) \log 2 \\
        &\quad + \frac{1-2x}{3}\II_2 + \frac1{24}(1+11x) - 2\zeta'(-1) + \log\Big(G(c+1)G(c+\tfrac32)\Big).
    \end{split}
    \end{align}
    Finally, using the relation $N = n+\al+c+\frac12$, we can recover the large $n$ asymptotics of the gap probability: 
    \begin{align} \label{H:logPnfinal}
    \begin{split}
        \log  \mc P_n^\HH(R;\al,c)& = \wh A_1'n^2 + \wh A_2' n \log n + ((2\al+2c+1)\wh A_1'+\wh A_3') n + \wh A_4' \sqrt{n} \\
        &\quad + \Big( (\al+c+\tfrac12) \wh A_2'+\wh A_5'\Big) \log n 
        \\
        &\quad +(\al+c+\tfrac12)^2 \wh A_1'+ (\al+c+\tfrac12)(\wh A_2'+\wh A_3') + \wh A_6' + \OO(n^{-\frac1{12}}).
    \end{split}
    \end{align}
    By putting \eqref{H:A tilde} to \eqref{H:logPnfinal}, the proof of Theorem~\ref{thm:1.2} is complete. 
\end{proof}

\appendix

\section{Tail Probabilities of the Binomial Distribution} \label{Appendix_binomial tails}

In this appendix, we compile some tail probabilities of the binomial distribution. 

The following lemma provides an upper bound for the lower tail probability, which can be found in \cite[p.405]{Slud1977}.

\begin{lem}
\label{Lemma:2.1}
     Let $X\sim B(n,p)$. Then for any integer $k \leq \frac{n^2p}{n+1}-1$, we have 
     \begin{equation} 
    \PP(X\geq k) \geq e^{-np}\sum_{j=k}^\infty \frac{(np)^j}{j!} \geq 1-\frac12 \erfc\Big(\frac{np-k}{\sqrt{2np}}\Big).    
     \end{equation}
\end{lem}

Let us also introduce a slightly improved version of the asymptotic behaviour in \cite{Ferrante2021} adapted to our purpose.

\begin{lem} \label{lem:2.7}
    Let $X \sim B(N,x)$ where $0<x<1$. Let $\del  \in (0, \min\{\frac14,1-x\}) $ be a constant. 
    Suppose that $\phi$ is an increasing function such that $0<\phi(N)<\del N$ and $\phi(N) \to \infty$ as $N \to \infty$. Then as $N \to \infty$,
    \begin{align}
    \begin{split}
        \log \PP[X\ge (1-t)N] &= -\frac12\Big(t\log \Big(\frac t{1-x}\Big) + (1-t)\log\Big(\frac{1-t}{x}\Big)\Big) N - \frac12 \log N - \log \Big(\frac{1-x-t}{(1-x)(1-t)}\Big)
        \\
        &\quad - \frac{1}{2}\log(2\pi t(1-t)) + \frac{1}{N}\bigg(\frac{1}{12}-\frac1{12}\frac{1}{t(1-t)}-\frac{xt}{(1-x-t)^2}\bigg) +\OO\Big(\frac{1}{t^3N^3}+\frac{1}{N^2}\Big) 
    \end{split}
    \end{align}
    uniformly for $t$ such that $tN$ is an integer with $\phi(N)<tN<\del N$. 
\end{lem}
\begin{proof}
    Write
    \begin{equation}
        s = 1-t, \qquad r = \frac{tx}{(1-t)(1-x)}.
    \end{equation}
    %We have an exact equality for $\PP[X\ge (1-t)N]$ as following.
    Then by definition, we have 
    \begin{equation} \label{tail prob exact formula}
        \PP[X\ge (1-t)N] = \sum_{k=0}^{tN} \binom{N}{tN-k}x^{(1-t)N+k}(1-x)^{tN-k} = \binom{N}{tN} (x^{1-t}(1-x)^{t})^N \sum_{k=0}^{tN}r^k\prod_{j=1}^k \frac{1-\frac{j-1}{tN}}{1+\frac{j}{sN}}.
    \end{equation}
    By using \eqref{log N!}, we have 
    \begin{align} \label{binomial coefficient asymptotic}
    \begin{split}
        \log \binom{N}{tN} &= -\Big(t\log t+(1-t)\log(1-t)\Big)N - \frac12 \log N \\
        &\quad - \frac12 \log(2\pi t(1-t)) +1 + \frac1{12N}\Big(1-\frac{1}{t(1-t)}\Big)+\OO\Big(\frac{1}{t^3N^3}\Big).
    \end{split}
    \end{align}
    
    Next, we estimate the summation in \eqref{tail prob exact formula}. We will compute explicit bounds for this sum. For this purpose, let
    \begin{align}
    \begin{split}
         L(k;N,t) &= 1-\frac{k^2+(2t-1)k}{2stN}, \\
        U(k;N,t) &= 1-\frac{k^2+(2t-1)k}{2stN} + \frac{3k^4+(20t-10)k^3+(9-24st)k^2+(4t-2)k}{24s^2t^2N^2}.
    \end{split}
    \end{align}
    We claim that for any $0 \le k \le tN$,
    \begin{equation}\label{bound of product term}
        L(k;N,t) \leq \prod_{j=1}^k \frac{1-\frac{j-1}{tN}}{1+\frac{j}{sN}} \leq U(k;N,t).
    \end{equation}
    Fix $N,t$ and let us write $L(k) := L(k;N,t)$, $U(k) := U(k;N,t)$.
    First we show the left inequality of \eqref{bound of product term}. Since the inequality is trivial if $L(k)\le0$, we may assume $L(k)>0$. 
    Note that it suffices to show 
    \[
    \frac{L(j)}{L(j-1)} \le \frac{1-\frac{j-1}{tN}}{1+\frac{j}{sN}}
    \]
    for every $j=1, \ldots, k$.
     Through straightforward computations, we have
    \begin{align*}
        L(j)\Big(1+\frac{j}{sN}\Big) &= 1-\frac{j^2-j}{2stN}-\frac{j^3+(2t-1)j^2}{2s^2tN^2}, \\
        L(j-1)\Big(1-\frac{j-1}{tN}\Big) &= 1-\frac{j^2-j}{2stN}+\frac{(j-1)^2(j+2t-2)}{2st^2N^2}.
    \end{align*}
   It is clear that the numerators of the last terms are non-negative for any positive integer $j$, thereby proving the lower bound part of \eqref{bound of product term}.
    
    Next, let us verify the right inequality of \eqref{bound of product term}. 
    Let $0 \le \al \le 1$. Note that if $tN$ is large enough, we have 
    \begin{align*}
        8s^2U(\al tN) &= t^2N^2\al^4 + \frac{20t-10}{3}tN \al^3 + (3-8st-4stN)\al^2+\frac{(2t-1)(2-12stN)}{3tN}\al + 8s^2 \\
        & = \Big(tN \al^2 + \frac{10t-5}{3}\al - 2s\Big)^2 + \Big(\frac29+\frac{28}{9}st\Big)\al^2+\Big(\frac{4t-2}{3tN}+\frac{16}{3}st - \frac{8}{3}s\Big)\al + 4s^2 \\
        & \ge -\frac2{3tN}-\frac83 s + 4s^2 \ge \frac14-\frac2{3tN} \ge 0.
    \end{align*}
    Therefore, it follows that 
    \begin{equation} \label{U positive}
        U(k) \ge 0
    \end{equation}
    for all integer $k = 0,1,\dots,tN$. Thus it is enough to show again
    \[
    \frac{1-\frac{j-1}{tN}}{1+\frac{j}{sN}} \le \frac{U(j)}{U(j-1)}
    \]
    for all $j=1,\dots, tN$. For this, note that 
    \begin{align*}
        U(j-1)\Big(1-\frac{j-1}{tN}\Big) &= 1-\frac{j^2-j}{2stN}+\frac{3j^4+(20t-10)j^3+(9-24st)j^2+(4t-2)j}{24s^2t^2N^2}-\frac{E_1(j)}{24s^2t^3N^3}, \\
        U(j)\Big(1+\frac{j}{sN}\Big) &= 1-\frac{j^2-j}{2stN}+ \frac{3j^4+(20t-10)j^3+(9-24st)j^2+(4t-2)j}{24s^2t^2N^2}+\frac{E_2(j)}{24s^3t^2N^3},
    \end{align*}
    where
    \begin{align*}
        E_1(j) %&= (j-1)\big[3j^4+(20t-22)j^3+(57-84t+24t^2)j^2+(-62+112t-48t^2)j + 24(1-t)^2\big], \\
        &= (j-1)^2\big[(j-2)(j-3)(3j-4)+4t(5(j-1)^2-6s(j-1)+1)\big], \\
        E_2(j) %&= 3j^5 + (20t-10)j^4+(9-24st)j^3+(4t-2)j^2 \\
        &= j^2[(j-1)(j-2)(3j-1) +4t(5j^2-6sj+1)].
    \end{align*}
    Since the quadratic function $f(u) = 5u^2 - 6su + 1$ attains its minimum at $u = \frac{3s}{5} \in (0,1)$ and both $f(0)$ and $f(1)$ are positive, we can deduce that $f(j) > 0$ for all integers $j$. Consequently, we can conclude that $E_1(j)$ and $E_2(j)$ are positive for all $j$, which leads to the upper bound part of \eqref{bound of product term}.

    Next, we assert that
    \begin{equation} \label{bound of sum prod term}
        \sum_{k=0}^{tN} r^k \prod_{j=1}^k \frac{1-\frac{j-1}{tN}}{1+\frac{j}{sN}} = \frac1{1-r}\Big(1-\frac1N\frac{tx}{(1-x-t)^2}+\OO\Big(\frac1{N^2}\Big)\Big).
    \end{equation}    
    For this, note that for large enough $tN$ and $k>tN$,
    \begin{align} \label{L negative U positive}
    \begin{split}
        L(k;N,t) &= 1-\frac{k^2+(2t-1)k}{2stN} \leq 1-\frac{(tN)^2+(2t-1)tN}{2stN} < 0, \\
       U(k;N,t) &= 1-\frac{k^2+(2t-1)k}{2stN} + \frac{3k^4+(20t-10)k^3+(9-24st)k^2+(4t-2)k}{24s^2t^2N^2} > 1
    \end{split}
    \end{align}
    since $3k^4 / (24s^2t^2N^2)$ term dominates the size of $U(k;N,t)$. Let $N$ be large enough so that $tN>\phi(N)$ satisfies \eqref{U positive} and \eqref{L negative U positive}. Then one can observe that
    \begin{align*}
        \sum_{k=0}^\infty r^k L(k;N,t) \le \sum_{k=0}^{tN} r^kL(k;N,t) \le \sum_{k=0}^{tN} r^k \prod_{i=1}^k \frac{1-\frac{i-1}{tN}}{1+\frac{i}{sN}} \le \sum_{k=0}^\infty r^k U(k;N,t) \le \sum_{k=0}^\infty r^k U(k;N,t).
    \end{align*}
    Here, we have assumed $t<\del<1-x$ so that $r<1$. Then we obtain 
    \begin{align*}
        \sum_{k=0}^\infty r^k L(t;N,t) &= \frac{1}{1-r}\Big(1-\frac1N\frac{tx}{(1-x-t)^2}\Big),\\
        \sum_{k=0}^\infty r^k U(k;N,t) &= \frac1{1-r}\Big(1-\frac1N\frac{tx}{(1-x-t)^2}+\frac1{N^2}\frac{tx(1-x^2-t(1-2x))}{(1-x-t)^4}\Big),
    \end{align*}
    which gives rise to \eqref{bound of sum prod term}.
    Then we conclude the proof by putting the asymptotics \eqref{binomial coefficient asymptotic} and \eqref{bound of sum prod term} to \eqref{tail prob exact formula}.
\end{proof}

%%%%%%%%%%bibliography%%%%%%%%%%%%%%%%%%%%%%%%%%%%%%%%%%%

%%%%%%%%%%%%%%%%%%%%%%%%%%%%%%%%%%%%%%%%%%%%%%%%%%%%%%%%%%%%%%	

\end{document}